\newcommand{\LONGVERSION}[1]{#1}
\newcommand{\SHORTVERSION}[1]{}

\newcommand{\LONGVERSIONREF}{the long version}
\ifdefined\LONGVERSION
  \relax
\else
 \newcommand{\SUBSTCLO}[1]{}
 \newcommand{\LONGVERSION}[1]{}
 \newcommand{\LONGVERSIONCHECKED}[1]{}
 \newcommand{\SHORTVERSION}[1]{#1}
\fi
\newcommand{\LONGSHORT}[2]{\LONGVERSION{#1}\SHORTVERSION{#2}}

\documentclass[acmsmall,screen]{acmart}

\setcopyright{rightsretained}
\SHORTVERSION{%
\acmPrice{}
\acmDOI{10.1145/3498700}
\acmYear{2022}
\copyrightyear{2022}
\acmSubmissionID{popl22main-p294-p}
\acmJournal{PACMPL}
\acmVolume{6}
\acmNumber{POPL}
\acmArticle{39}
\acmMonth{1}
}

\usepackage{booktabs}   
\usepackage{cancel}
\usepackage{changepage}
\usepackage{cprotect}
\usepackage[shortlabels]{enumitem}
\usepackage{graphicx}
\usepackage{amsmath}
\usepackage{listings}
\usepackage{multicol}
\usepackage{multirow}
\usepackage{natbib}
\bibpunct{[}{]}{;}{a}{}{,}

\usepackage{proof}
\usepackage{stmaryrd}
\usepackage{subcaption} 

\lstdefinelanguage{Moebius}%
{%
  keywords=%
  {%
    box, unbox, with,%
    mFn,tFn,fn,fix,%
    if,else,then,match,%
    let,in,%
  },%
  sensitive,
}

\newcommand{\Moebius}{M\oe{}bius}
\newcommand{\Beluga}{\textsf{Beluga}}

\newcommand{\ctxmerge}[2]{{{#1} \oplus {#2}}}
\newcommand{\ctxinsert}[2]{{{#1} \oplus {#2}}}
\newcommand{\ctxinsertatl}[3]{{{#1} \oplus {#2}^{#3}}}
\newcommand{\ctxappend}[2]{{{#1},{#2}}}

\newcommand{\ctxrestrict}[2]{#1|_{#2}}
\newcommand{\ctxignore}[2]{#1|^{#2}}

\newcommand{\bnfas}{\mathrel{::=}}
\newcommand{\bnfalt}{\mathrel{\mid}}
\newcommand{\m}[1]{\mathsf{#1}}
\newcommand{\arrow}{\rightarrow}

\newcommand{\id}{\m{id}}

\newcommand\ceil[1]{\lceil#1\rceil}
\newcommand{\cbox}[1]{\ceil {#1}}

\newcommand{\Psihat}{\hat{\Psi}}
\newcommand{\Gammahat}{\hat{\Gamma}}
\newcommand{\Phihat}{\hat{\Phi}}

\newcommand{\tto}{\rightarrow}
\newcommand{\tfn}[2]{\m{Fn}\;#1 \tto #2}
\newcommand{\fn}[2]{\m{fn}\;#1 \tto #2}

\newcommand{\letboxm}[3]{\m{let}\,\m{box}\;#1\;{=}\;#2\;\m{in}\;#3}
\newcommand{\boxn}[3]{\m{box}_{#1}(#2.\,#3)}
\newcommand{\boxm}{\boxn{}}

\newcommand{\caseof}[2]{\m{case}_{#2}\ {#1}\ \m{of}}

\newcommand{\atl}[1]{\vdash^{#1}}

\begin{document} 

\title
{\Moebius: Metaprogramming using Contextual Types}
\subtitle{The stage where System F can pattern match on itself 
\LONGVERSION{-- \textbf{Long version}}}



\author{Junyoung Jang}

\affiliation{
  \department{School of Computer Science} 
  \institution{McGill University}         
  \country{Canada}                    
}
\email{junyoung.jang@mail.mcgill.ca}          

\author{Samuel G{\'e}lineau}

\affiliation{
  \department{} 
  \institution{Simspace}         
  \country{Canada}                    
}
\email{}          

\author{Stefan Monnier}

\affiliation{
	\department{Computer Science} 
	\institution{Universit{\'e} de Montr{\'e}al}         
	\country{Canada}                    
}
\email{monnier@iro.umontreal.ca}

\author{Brigitte Pientka}

\affiliation{
	\department{School of Computer Science} 
	\institution{McGill University}         
	\country{Canada}                    
}
\email{bpientka@cs.mcgill.ca}         


\begin{abstract}
We describe the foundation of the metaprogramming language, \Moebius{},
which supports the generation of polymorphic code and, more
importantly the analysis of polymorphic code via pattern matching. 

\Moebius{} has two main ingredients: 1) we exploit contextual modal
types to describe open code together with the context in which it is
meaningful. In \Moebius, open code can depend on
type and term variables (level \(0\)) whose values are supplied at a later
stage, as well as code variables (level \(1\)) that stand for code
templates supplied at a later stage. This leads to a multi-level modal
lambda-calculus that supports System-F style polymorphism and forms
the basis for polymorphic code generation.
2) we extend the multi-level modal lambda-calculus to
support pattern matching on code. As pattern matching on polymorphic code
may refine polymorphic type variables, we extend our type-theoretic
foundation to generate and track typing constraints that arise. We
also give an operational semantics and prove type preservation.

Our multi-level modal foundation for \Moebius{} provides the appropriate
abstractions for both generating and pattern matching on open code
without committing to a concrete representation of variable binding
and contexts. Hence, our work is a step towards building a general
type-theoretic foundation for multi-staged metaprogramming that, on the
one hand, enforces strong type guarantees and, on the other hand, makes it
easy to generate and manipulate code. This will allow us to exploit
the full potential of metaprogramming without sacrificing the reliability
of and trust in the code we are producing and running.
\end{abstract}




\maketitle



\section{Introduction}\label{sec:introduction}
Metaprogramming is the art of writing programs that produce or manipulate other programs. This opens the possibility to eliminate boilerplate code and exploit domain-specific knowledge to build high-performance programs. Unfortunately, designing a language extension to support type-safe, 
multi-staged metaprogramming remains very challenging.

One widely used approach to metaprogramming going back to Lisp/Scheme
is using \emph{quasiquotation} which allows programmers to generate and
compose code fragments. This provides a simple and flexible mechanism
for generating and subsequently splicing in a code fragment. However,
it has been challenging to provide type safety guarantees about the
generated code, and although many statically typed programming languages such as Haskell
\cite{Sheard:Haskell02}, and even dependently-typed languages such as
Coq \cite{Anand:ITP18} or Agda \cite{vanderWalt:IFL12} support a form
of quasiquotation, they generate untyped code. Fundamentally, we lack
a type-theoretic foundation for quasiquotation that provides rich
static type guarantees about the generated code and allows programmers
to analyze and manipulate it.

This is not to say that no advances have been made in the area of
typed metaprogramming. Two decades ago, MetaML\cite{Taha:TCS00,Taha:POPL03}
pioneered the type-safe code generation from a practical
perspective. At the same time, \citet{Davies:ACM01} observed that
the \emph{necessity (box) modality} allows us to distinguish the
generated code from the programs that are 
generating it providing a logical foundation for metaprogramming.
For example, code fragments like \lstinline!box(2 + 2)! have boxed
types such as \lstinline![nat]!, while programs such as 
\lstinline!6 * 3!, which will evaluate to $18$ and have type \lstinline!nat!.
\citet{Nanevski:TOCL08} extend this idea to describe open code with
respect to a context in which it is meaningful using contextual
types. For example, the  code fragment \lstinline!box(x. x + 2)! has
the contextual box type \lstinline![x:nat |- nat]!. To use a piece of
code such as \lstinline!box(x. x + 2)!, we bind \lstinline!x.x+2! to
the contextual variable \lstinline!U! of type 
\lstinline!(x:nat |- nat)! using a \lstinline!let-box! expression. At the site where we use
\lstinline!U! to eventually splice in \lstinline!(x.x+2)!, we
associate \lstinline!U! with a delayed substitution giving a value to
\lstinline|x|. As soon as we know
what \lstinline!U! stands for, we can apply the substitution. For
example, \lstinline!let box (x.U) = box(x.x+2) in box(U with 3)!
will bind \lstinline!U! to the code \lstinline!(x.x+2)! during
runtime, and produce \lstinline!box(3 + 2)! as a result where
\lstinline!x! in the code \lstinline!x+2! has been replaced by
\lstinline!3!. We hence treat contextual variables as closures
consisting of a variable and a delayed substitution (written using the
keyword \lstinline!with!, in this example). 
This allows us not
only to instantiate open code fragments, but also ensures that
variables in code fragments are properly renamed when we splice them
in. Using closures and contextual types is in
contrast to using functions and function types where we could write a program
\lstinline!let box U = box(fn x -> x + 2) in box(U 3)!. Here, evaluation would produce the code
\lstinline!box ((fn x -> x + 2) 3)! containing an administrative redex,
which is particularly undesirable in \Moebius, where the program can
inspect the code. Arguably, contextual types lead to more compact generated code which is more in line with the programmer's 
intention. 

This line of work
cleanly separates local variables (such as \lstinline!x!) and global
variables (i.e. variables such as \lstinline!U!,
accessible at every stage) into two zones. This provides an
alternative to the Kripke-style view where a stack
of contexts models the different stages of computation. While the
latter is appealing, since \lstinline!box! (quote) and
\lstinline!unbox! (unquote) match common
metaprogramming practice, reasoning
about context stacks and variable dependencies is difficult. The
\lstinline!let box! formulation of the modal necessity leads to an arguably
much simpler formulation of the static and operational semantics, which
we see as an advantage.



In this paper, we introduce a core metaprogramming language,
\Moebius{}, which allows us to generate and, importantly, analyze
polymorphically-typed open code fragments using pattern
matching. Our starting point is the lambda-calculus presented by \citet{Nanevski:TOCL08} where
they distinguish between local and global variables. This formulation
makes evaluation order explicit using 
\lstinline!let box!-expressions. 
This allows us to more
clearly understand the behavior of multi-staged metaprograms that work
with open code. We extend their work in three main directions: 

\smallskip

\textbf{1) Generating polymorphic code.} First, we generalize their
language to support System F style polymorphism
and polymorphic code generation. In particular, 
\lstinline!['a:*, x:'a, f: 'a -> 'a |- 'a]!  describes a polymorphic code fragment such
as \lstinline!box('a, x, f. f (f x))!. A context in our setting keeps
track of both term and type variables, hence treating both assumptions
uniformly. This allows us to support code generation for polymorphic data
structures such as polymorphic lists.  
%
%

\textbf{2) Generating code that depends on other code and type templates.}
Second, we generalize contextual types to characterize code that depends
itself on other code fragments and support a composition of code
fragments which avoids creating administrative redexes due to
boxing.
This is achieved by generalizing the two-zone formulation by
\citet{Davies:ACM01} and \citet{Nanevski:TOCL08}, which distinguishes
between two levels, to an n-ary zone formulation and leads us to a
multi-level contextual modal lambda-calculus.  The two-zone
formulation by Davies, Pfenning, et. al. is then a special case of
our multi-level calculus. 
For example, \lstinline![c:(x:nat |- nat), x:nat |- nat]! can describe a
piece of code like \lstinline!box(c, x. x + c with x)! that eventually
computes a natural number, but depends on the variable
\lstinline!x:nat! and on another piece of open code
\lstinline!c:(x:nat |- nat)!. Note that \lstinline!c! simply stands for
a piece of open code and any reference to it will have to provide an
\lstinline!x! with which to close it -- this is in contrast to an assumption
of type  \lstinline![x:nat |- nat]! which stands for the boxed version of
a piece of code. This allows us to elegantly and concisely combine
code fragments:  

\begin{lstlisting}[numbers=left]
let box (y. R) = box (y. y + 2) in 
let box (c,x. U) = box (c,x. 3 * x + (c with (2 * x))) in 
  box (y. U with ((y. R with y), y))
\end{lstlisting}
which results in the code \lstinline!box(y. 3 * y + (2 * y + 2))!. We say
that \lstinline![c:(x:nat |- nat), x:nat |- nat]! is a level 2
contextual type, as it depends on \lstinline!x:nat! (level 0) and
\lstinline!c:(x:nat |- nat)! (level 1) assumptions. This view is in
contrast to the two-level modal lambda-calculus described in
\citet{Davies:ACM01} or \citet{Nanevski:TOCL08}, where we would need to
characterize code depending on other code using the type
\lstinline![c:[x:nat |- nat], x:nat |- nat]!. However, this would not
allow us to splice in \lstinline!(y. y + 2)! for \lstinline!c! in the code bound
to \lstinline!U!; instead, we would need to splice in 
\lstinline!box(y. y + 2)! for \lstinline!c! and retrieve the code \lstinline!(y. y + 2)!
using a \lstinline!let-box!-expression. The code
bound to \lstinline!U! in line 2 would need to be written as:
\lstinline!box (c,x. let box (y. R') = c in 3 * x + (R' with (2 * x)))!.
This would then generate the code 
\lstinline!box(y. let box (y. R') = box(y. y + 2) in 3 * y + (R' with (2 * y)))!
which contains an
administrative redex.
Our work avoids the generation of any
administrative redexes and hence generates code as intended and
envisioned by the programmer. 

We further extend this idea of code templates uniformly to type
templates, i.e.~we are able to describe the skeleton and shape not only of
code but also of types themselves. As a consequence, \Moebius{}
sits between System F and System F$_\omega$ where we support term and
type-level computation.


\textbf{3) Pattern matching on code.}
Third, we extend the multi-level contextual modal lambda-calculus with
  support for pattern matching on code. This follows ideas in \Beluga{}
  \cite{Pientka:POPL08,Pientka:PPDP08} 
  where pattern matching on higher-order abstract syntax (HOAS) trees
  is supported. However, in \Beluga{}, we separate the language in which
  we write programs from the language that describes syntax. In \Moebius,
  such a distinction does not exist and meta-programs can
  pattern match on code that represents another (meta-)program in the same
  language.
  We view the syntax of code through the lenses of higher-order abstract syntax
  treating variables abstractly and characterize pattern variables
  using multi-level contextual types. In fact, multi-level contextual
  types are the key to characterizing pattern variables in code --
  especially code that may itself contain
  \lstinline!box!-expressions and \lstinline!let box! expressions!

   Pattern matching on polymorphic code may refine type variables, as
   is typical in indexed or dependently typed systems. For example, when
   we pattern match on a piece of code of type
   \lstinline![x:'a |- 'a]!, then
   one of the branches may have the pattern \lstinline!box(x. 0)!,
   which has type \lstinline![x:int |- int]!. Hence, we need to use
   the constraint \lstinline!'a = int! to type check the body of the
   branch. We, therefore, extend our multi-level contextual modal lambda
   calculus to generate and track type constraints, and type-check a
   given expression modulo constraints. Despite the delicate issues
   that arise in supporting System F style polymorphism and pattern
   matching on code, our operational semantics and the accompanying
   type preservation proof is surprisingly compact and clean.

%

\medskip

We view our work as a step towards building a general
type-theoretic foundation for multi-staged metaprogramming, which enforces strong type guarantees and makes it
easy to generate and manipulate code. This will allow us to exploit
the full potential of metaprogramming without sacrificing the reliability
of and trust in the code we are producing and running.


\section{Motivation}\label{sec:motivation}
To illustrate the design and capabilities of \Moebius, we discuss
several examples below.

\subsection{Example: Generating Open Polymorphic Code}
First, we implement the function \lstinline!nth! which generates
the code to look up the \lstinline!i!-th element in a polymorphic list
\lstinline!v! where \lstinline!v! is supplied at a later (next)
stage. 

\begin{lstlisting}
nth : int -> ['a:*, v:'a list |- 'a]
nth n = if n <= 0 then
          box('a,v. hd v )
        else
          let box ('a,v. X) = nth (n - 1) in box('a,v. X with 'a, tl v) )
\end{lstlisting}

The result of the computation has 
the contextual type \lstinline!['a:*, v:'a list |- 'a]!. This
type describes open code at level \(1\) that has type
\lstinline!'a! and depends on variables from level \(0\), namely
\lstinline!'a:*, v:'a list!. In general, a contextual type
\lstinline![$\Psi$ $\atl{n}$ T]! characterizes a code template at level $n$
of type \lstinline!T!. This code template may depend on locally bound
variables in $\Psi$ which contains variables at levels strictly lower than
$n$. The code template may also refer to outer variables that are
greater or equal to $n$ and whose values are computed during run-time. This principle is what underlies the design
of the multi-level modal lambda-calculus and is worth highlighting:

\begin{center}
\fbox{
  \begin{tabular}{p{12cm}}
\emph{Mantra: a code template $\cbox{\Psi \atl{n} T}$ at level $n$, can depend on locally bound
    variables $\Psi$ from levels less than $n$ and outer variables that have levels greater or equal to $n$.}
  \end{tabular}}
\end{center}

In the above example, the result that we return in the recursive case
is \lstinline!box('a,v. X with 'a, tl v)!. The code inside the
\lstinline!box!, depends on the locally bound variables 
\lstinline!'a! and \lstinline!v! (all of which have level $0$) and uses the outer variable
\lstinline!X! (which has level $1$). During runtime, \lstinline!X!
will be bound to the result of the recursive call \lstinline!nth (n-1)!.
Given an
\lstinline!n!, the program \lstinline!nth! will recursively build up the code 
\begin{center}
\lstinline!box('b, v. hd (!$\underset {n}{\underbrace{\mbox{\lstinline!tl (.. (tl!}}}$\lstinline! v))))!
\end{center}
Ultimately producing a code template that depends only on variables at
level $0$. 

Subsequently, we write $\atl{n}$ in a contextual type, if we want to make
explicit the level at which a term is
well-typed, but we will mostly omit the levels when they can be
easily inferred.







\subsection{Example: Combining Code Templates}
Next, we 
generate code that depends on two other code templates, \lstinline!c! and
\lstinline!d!, both of which have the contextual type
\lstinline!(x:int |-$^1$ int)!. 
The type of these templates is at level $1$, as they depend on the
variable \lstinline!x:int!, which is at level $0$. Hence, the overall
generated code template lives at level $2$.

In this simple example, we combine the two templates \lstinline!c! and
\lstinline!d! in different ways depending on whether the input to the
function \lstinline!combine! is
\lstinline!true! or \lstinline!false!. From a computational view, if
the input evaluates to \lstinline!true!, then we generate code that, during
runtime, will first evaluate the template
\lstinline!d! and subsequently pass its result to the template
\lstinline!c!. If the input evaluates to \lstinline!false!, then we do
the opposite, i.e., we will first evaluate the template \lstinline!c! and
then pass its result to \lstinline!d!.



\begin{lstlisting}
combine : bool -> [c:(x:int |-$^1$ int) , d:(x:int |-$^1$ int), x:int |-$^2$ int]
combine p = if p then
              box(c,d,x. (fun y -> c with y) (d with x))
            else
              box(c,d,x. (fun y -> d with y) (c with x))   
\end{lstlisting}

One might wonder whether we could have generated the code
\lstinline!box(c,d,x. c with (d with x))! instead of
\lstinline!box(c,d,x. (fun y -> c with y) (d with x))!.
To understand the difference in the runtime behaviour, we use the
following instantiation:  \lstinline!(x. x+2*x)!
for \lstinline!c!, \lstinline!(x. x*3)! for \lstinline!d!, and
\lstinline!3! for \lstinline!x!.

When we use the code
\lstinline!box(c,d,x. (fun y -> c with y) (d with x))! with
\lstinline!(x. x+2*x), (x. x*3), 3!, then we run the code
\lstinline!(fun y -> y + 2*y) (3*3)!, which will first evaluate
\lstinline!3*3! to \lstinline!9! and then compute \lstinline!9 + 2*9!,
which returns $27$. When we use \lstinline!box(c,d,x. c with (d with x))!
with the same instantiation, we will evaluate the code
\lstinline!3*3 + 2*(3*3)!, i.e. we will evaluate the code
\lstinline!3*3!, which we obtain by instantiating \lstinline!d with x!
twice! Effectively, we are using a call-by-name evaluation strategy,
as the execution of \lstinline!3*3! gets delayed. This example hence
illustrates yet another difference between closures and
function applications.

Last, this example also highlights the difference between levels and
stages. The result for \lstinline!combine! is a contextual type at
level $2$, since it depends on other code templates. However, we are
generating the code in one stage.


The difference between levels and stages can also be understood from a
logical perspective: stages characterize \emph{when code is
generated}; as such they describe a property of a boxed type in a
positive position. Levels express a property of the assumptions
(i.e. negative occurrences) that are used in a boxed type, namely that
\emph{how code is used}. In particular, levels allow us to express
directly and accurately the fact that code may depend not only on
closed values, but also on open code. This allows us to describe code
that may itself depend on other code. 

\subsection{Example: Type Templates}
We support System-F style polymorphism in \Moebius. As a consequence,
we may not only want to describe open code, but also types
that are open (i.e. type templates). For example, consider
the contextual type

\begin{lstlisting}
['a:('c:* |- *), f: $\forall$'b. 'b -> ('a with 'b) -> int |- int]
\end{lstlisting}

This describes a piece of code that relies on the polymorphic
function \lstinline!f!, which computes an \lstinline!int!. Here \lstinline!'a! describes a
type template -- it stands for some type that has one free type
variable. We again associate \lstinline!'a! with a substitution which,
in this example, renames the variable \lstinline!'c! in the type template to be
\lstinline!'b!. As a consequence, \lstinline!f! in fact stands for a family
of functions! If we instantiate \lstinline!'a! with the type template
\lstinline!'c.'c!, then \lstinline!f! has type
\lstinline!$\forall$ 'b. 'b -> 'b -> int!. If we instantiate it with
\lstinline!'c. 'c -> c'!, then \lstinline!f! stands for a function of
type \lstinline!$\forall$ b'. b' -> (b' -> b') -> int!.
In System F$_\omega$, we would have declared this type as:
\lstinline!['a: * -> *, f:$\forall$ 'b. 'b -> 'a 'b -> int |- int]!

The support of type templates that have contextual kinds means that
\Moebius{} sits between System F$_\omega$ where we have type level
functions and System F. Using type variables that have a contextual
kind brings a distinct advantage over System F$_\omega$: as we apply
the substitution associated with the type variable \lstinline!'a! as
soon as we know its instantiation, we can compare two types simply by
structural equality, and we do not need to reason about type-level
computation. 

The ability to characterize holes in types and terms is particularly
important when we consider pattern matching on code. In our setting,
polymorphic programs may contain explicit type applications, and hence,
we not only pattern match on terms but also on types. 


\subsection{Example: Lift Polymorphic Data Structures}
In \Moebius, as in other similar frameworks, we need to lift values
explicitly to the next stage. Lifting integers is done by
simply traversing the input and turning it into its syntax
representation.  This is straightforward. We write here \lstinline![int]! as an abbreviation for
\lstinline![ |-$^1$ int]!, which describes a closed integer at level
$1$, i.e. the minimum level that generated code can have.

\begin{lstlisting}
lift_int : int -> [int]
lift_int n = if n = 0 then box(0) else let box(X) = lift_int (n - 1) in box(X + 1)
\end{lstlisting}

To lift polymorphic lists, we need to lift values of type
\lstinline!'a! to their syntactic representations and then lift lists
themselves. This generic lifting function for values of type
\lstinline!'a! will intuitively have the
type \lstinline!'a -> ['a]!. But how can we ensure that the type
variable \lstinline!'a! can be used inside a contextual type at level
$1$? -- To put it differently, how can we guarantee that the type
variable persists when we transition inside the contextual box type?
We again use the intuition that \emph{code and types at stage $n$, have
locally bound variables from levels less than $n$, but they may depend on outer
variables that have levels greater or equal to $n$.} For \lstinline!['a]!
to be a well-formed type, the type variable \lstinline!'a! must be
declared at level $1$ or higher. 


Hence, we declare the type variable \lstinline!'a! as a type template 
\lstinline!'a: ($~$|- *)!. We omit here again the level $1$, since it
can be inferred from the surrounding context. 

\begin{lstlisting}
lift_list : ('a:( $~$|- *)) -> 'a list -> ('a -> ['a]) -> ['a list]
lift_list l lift = match l with
| v[]v    -> box(v[]v)
| x::xs -> let box X = lift x in let box XS = lift_list xs lift in box(X::XS)
\end{lstlisting}

For terms, we have a natural interpretation of running and evaluating
code. This rests on the idea that a term at stage 1 can always be used
at stage 0.

\begin{lstlisting}
eval_int : [int] -> int
eval_int x = let box X = x in X

eval_list : ('a:( |- *)) -> ['a list] -> 'a list
eval_list v = let box V = v in V
\end{lstlisting}

Similarly, we can run and evaluate polymorphic code. Again we rely on
the idea that the type variable \lstinline!a! is describing types at
level \(1\) or below.

This concept also allows us to lift the result of a function
\lstinline!rev ('a:*) -> 'a list -> 'a list!. In the code below, we
simply call \lstinline!rev! with \lstinline!'a! (or more precisely
with \lstinline!'a with .!).

\begin{lstlisting}
lift_rev_list : ('a:( $~$|- *)) -> 'a list -> ('a -> ['a]) -> ['a list] 
lift_rev_list l = lift_list (rev l)  
\end{lstlisting}







\cprotect\LONGVERSION{
\subsection{Example:  Staged Polymorphic Map}
With the ability to generate polymorphic code, there are many useful
programs we can write, for example a staged polymorphic map
function. We omit again the levels, as they can be inferred. 
\begin{lstlisting}
map : (a:( |- *)) -> a list -> (a -> [a]) -> [b:*, f: a -> b |- b list]
map l lift = match l with
| v[]v    -> box( b, f . v[]v )
| x::xs -> let box X = lift x in
           let box (b,f. XS) = (map xs lift) in
             box(b, f . f (X :: XS))
\end{lstlisting}

We again exploit the fact that we declare type variable \lstinline!a! at level
\(1\), so we can use it inside a boxed contextual type. While the type
variable \lstinline!b! is locally bound, the type variable
\lstinline!a! is bound outside.
}




\subsection{Example: Multi-Staged Polymorphic Code}
We now give an example of multi-staged polymorphic code generation. Here, names
\lstinline!l! and \lstinline!liftA! in the function type are to help
the understanding, and are not part of our syntax. It is a multi-staged version of
\lstinline!map_reduce!. We annotate all the contextual types and kinds
with their level, although the level can be inferred from where the
type variable is used. It serves as another illustration that the
level manages variable dependencies and cross-stage persistence\cite{Taha:TCS00}, while
the staging manages when code is generated. In the code below, we omit
writing the identity substitution that is associated with variables
writing for example simply \lstinline!R! instead of \lstinline!(R with 'b,f,liftB)!.

\begin{lstlisting}
map_reduce : ('a:( |-$^2\;$*)) -> (l:'a list) -> (liftA:'a -> [ |-$^2\;$'a])
          -> ['b:( |-$^1\;$*), f:'a -> 'b, liftB:'b ->[ |-$^2\;$'b] 
                                  |-$^2\;$['c:*, g:'b ->'c ->'c, base:'c |-$^1\;$'c]]
map_reduce l liftA = match l with
| v[]v    -> box ('b,f,liftB. box('c,g,base. base) )
| x::xs ->
  let box ('b,f,liftB. R) = map_reduce xs liftA in
  let box A = liftA x in box ('b,f,liftB.
    let box ('c,g,base. M) =  R in
    let box X' = liftB (f A) in box ('c,g,base. M with 'c, g, (g X' base)))
\end{lstlisting}
Multi-stage code generation generates
code incrementally in several stages. For example,
\lstinline!map_reduce! generates the final code in \(2\) stages. This is
evident, since the return type has a boxed type nested inside another
boxed type.


In the function declaration of \lstinline!map_reduce!, we declare
\lstinline!('a:( |-$^2$ *))!, as we want to use \lstinline!a! at any
stage \(2\) or below. However, the other declarations
\lstinline!l! and \lstinline!liftA! are declared at level $0$ -- they
will be inaccessible inside a box-expression. Therefore, this
forces us to lift the head of the list (\lstinline!liftA x!) first
before we build our result, as when we enter the scope of
\lstinline!'b:( |-$^1$ *); f:'a -> 'b, liftB:'b -> [ |-$^2$ 'b] !
we only keep the declaration \lstinline!('a:( |-$^2$ *))!, but the
outer declarations at level \(0\) get dropped as they become
inaccessible. Similarly, when we enter the inner contextual box
\lstinline!['c:*, g:'b -> 'c -> 'c, base:'c |-$^1$ 'c]!, we will drop all the
outer assumptions at level \(0\) (as they get overwritten by the new local
context \lstinline!'c:*, g:'b -> 'c -> 'c, base:'c!, but assumptions at levels higher than
\(0\) remain accessible. The levels, therefore, manage scope dependencies --
especially as we cross different stages.

\subsection{Example: Working with Church Encodings and Pattern Matching}


The last example illustrates the benefits of generating code as
intended by programmers and the use of pattern matching. We represent
the type of natural numbers using Church encoding 
as \lstinline!['a:*, x:'a, f:'a -> 'a |- 'a]!.  
However, unlike the usual Church encoding where we
use functions and function application, we exploit the power of
contextual types instead. This then allows for elegant implementations
of addition or other arithmetic operations. We will again omit writing
the identity substitutions, i.e. for example, writing simply
\lstinline!N! instead of \lstinline!(N with 'a,x,f)!.

\begin{lstlisting}
gen_church: int -> ['a:*,x:'a, f:'a -> 'a |- 'a] 
gen_church n = if n = 0 then box('a,x,f. x) 
               else let box('a,x,f. N) = gen_church (n-1) in box('a,x,f. f N)

add: ['a:*,x:'a, f:'a -> 'a |- 'a] -> ['a:*,x:'a, f:'a -> 'a |- 'a] 
  -> ['a:*,x:'a, f:'a -> 'a |- 'a]
add n m = let box ('a,x,f. N) = n in 
          let box ('a,x,f. M) = m in box('a,x,f. N with 'a, M, f)
\end{lstlisting}

Working with the Church encoding directly as syntactic representations
again allows us to generate code as the programmer intended. For
example adding the Church encoding of $2$ and the Church encoding of
$3$ will yield the Church encoding of $5$ -- not some program that
will evaluate to $5$ when we run the result. As we know the structure
of the Church encoding for numbers, we can also inspect it via pattern
matching. This opens new possibilities to implement the predecessor
function directly via pattern matching instead of the usual more
painful solution when we work with functions and function applications.  

\begin{lstlisting}
pred n = case n of
| box('a,x,f. x) -> box('a,x,f. x)
| box('a,x,f. f X) -> box('a,x,f. X)
\end{lstlisting}

The first case captures the fact that the input \lstinline!n! is a
representation of $0$, and we simply return it. In the second case, we
know that it is an application \lstinline!f X! where $X$ is the
pattern variable denoting a code template of type
\lstinline!('a:*,x:'a,f:'a -> 'a |- 'a)!. 
We hence simply strip off the \lstinline!f!. 

We delay a more in-depth discussion of pattern matching on code to
Sec.~\ref{sec:PatMatch}.

\section{A Multi-level Modal Lambda-Calculus with Polymorphism}
We describe first a modal polymorphic lambda-calculus where we use
multi-level contextual types to characterize open code. This calculus
serves as a foundation to generate polymorphic open code. We then extend this
calculus to support pattern matching on polymorphic code in Sec.~\ref{sec:PatMatch}.

\begin{figure}[ht]
\[
\begin{array}{llcl}
\mbox{Types} & T,S & \bnfas & \alpha[\sigma] \bnfalt T_1 \arrow T_2 \bnfalt
                   (\alpha{:}(\Psi\atl{n} *)) \arrow T \bnfalt \cbox{\Psi \atl{n} T}\\
\mbox{Terms} & e   & \bnfas & x[\sigma] \bnfalt \fn{x}{e} \bnfalt e_1\;e_2
\bnfalt \tfn{\alpha^n}{e} \bnfalt e\;(\Psihat^n.T) \bnfalt
\\
& &              & \boxm {\Gammahat^n} e \bnfalt \letboxm{(\Gammahat^n.u)}{e_1}{e_2}
\\[0.5em]
\hline
\\[-0.5em]
\mbox{Substitution} & \sigma & \bnfas & \cdot \bnfalt \sigma,
                                        \Psihat^n.e \bnfalt \sigma; x^n \bnfalt
                                        \sigma, \Psihat^n.T \bnfalt \sigma; \alpha^n \\
\mbox{Context} & \Gamma, \Psi, \Phi & \bnfas &
   \cdot \bnfalt \Gamma, x{:}(\Psi \atl{n} T)
         \bnfalt \Gamma, \alpha{:}(\Psi\atl{n}*)
\\
\mbox{Erased context} & \Gammahat,\Psihat, \Phihat & \bnfas &
\cdot \bnfalt \Gammahat, x^n \bnfalt \Gammahat, \alpha^n
\end{array}
\]
  \caption{Syntax of Multi-Level Modal Lambda-Calculus}
  \label{fig:synmoebius}
\end{figure}

\subsection{Syntax}
We describe the syntax of \Moebius{} in Fig.~\ref{fig:synmoebius}.
Central to \Moebius{} are multi-level contextual types
and kinds that describe code and type templates together with the context
in which they are meaningful. A multi-level contextual type $(\Psi
\atl{n} T)$ describes code of type $T$ in a context $\Psi$ where all
assumptions in $\Psi$ are themselves at levels below $n$. Variables at
level $n$ or above are viewed as global variables whose values will be
computed during run-time. Similarly, a multi-level
contextual kind $(\Psi \atl{n} *)$ describes a type fragment in a
context $\Psi$.

We treat all variable declarations in our context uniformly; hence our typing context
keeps track of variable declarations $x:(\Psi \atl{n} T)$ and type
declarations $\alpha:(\Psi \atl{n} *)$.
If $n=0$, then we recover our ordinary bound variables of type $T$
from $x:(\cdot \atl{0} T)$. As there are no possible declarations at a level below
$0$, this context is necessarily empty (in fact it doesn't even exist). Similarly, a type
declaration $\alpha:(\cdot \atl{0} *)$ denotes simply a type variable
$\alpha$.

A contextual type or term variable is associated with a substitution, written here as
$\alpha[\sigma]$ and $x[\sigma]$. 
Intuitively, if a contextual variable $x:(\Psi \atl{n} \_)$ is declared in some context
$\Phi$, then the substitution $\sigma$ provides instantiations for
variables in $\Psi$ in terms of $\Phi$. As soon as we know the instantiation for $x$, we apply the
substitution $\sigma$. We previously wrote  \lstinline!(x with $\sigma$) !
as concrete source syntax in our code examples in
Sec.~\ref{sec:motivation}. 

%
\Moebius{} is a generalization of the modal lambda-calculus described by \citet{Davies:ACM01} or
\citet{Nanevski:TOCL08} which separates the global assumptions in a
meta-context from the local assumptions. This gives us a two-zone
representation of the modal lambda-calculus. \Moebius{} generalizes
this work to an $n$-ary zone representation. If $n=1$, then we obtain
the two-zone representation from the previous work. Variables at level
$1$ correspond to the meta-variables which live in the global
meta-context, and variables at level $0$ correspond to the ordinary
bound variables. 

Our notion of multi-level contextual types is also similar to the work by 
\citet{Boespflug:LFMTP11}. However, this work was concerned with
developing a multi-level contextual logical framework. We adopt the
ideas to polymorphic multi-staged programming.

\Moebius{} supports boxed contextual types, $\cbox{\Psi \atl{n} T}$
which describe a code template of type $T$ in the context $\Psi$. The
level $n$ enforces that the given template can only depend locally on
variables at levels below $n$. This will allow us to cleanly manage
variable dependencies between stages and will serve as a central guide
in the design of \Moebius.
In addition, we support function types (written as
$T_1 \arrow T_2$) and polymorphic function $(\alpha{:}(\Psi\atl{n} *))
\arrow T$ which abstract over type variables $\alpha$. The ordinary polymorphic function space is a special
case where $n=0$. As a type variable $\alpha$ stands in general for a
type-level template, we associate it with a substitution $\sigma$,
written as $\alpha[\sigma]$. As soon as we know what type 
variable $\alpha$ stands for, we splice it in and apply $\sigma$. This ensures that
the type makes sense in the context where it is used.
When $n=1$, we can model types in System F$_\omega$. For
example, in System F$_\omega$, we might define:
$(\alpha{:} * \arrow *) \arrow (\beta:*) \arrow \alpha ~\beta$. In
\Moebius, we declare it instead as:
$(\alpha{:}(\_{:}* \atl{1} *)) \arrow (\beta:*) \arrow \alpha[\beta]$
 avoiding type-level functions and creating a type-level function 
 application. 

Our core language of \Moebius{} includes
functions ($\fn{x}{e}$), function application ($e_1\;e_2)$,
type abstractions ($\tfn{\alpha^n}{e}$) and type application
($e\;(\Psihat^n.T)$). We carry the level $n$ of $\alpha$ as an
annotation; this is needed for technical reasons when we define
simultaneous substitution operations.
Further, we include $\m{box}$ and
$\m{let}\,\m{box}$ expressions to generate code and use code.
In particular, $\boxm{\Gammahat^n} e$ describes code $e$ together with
the list of (local) variables $\Gammahat^n$ which may be used in
$e$. The level $n$ states again that these bound variables are below level
$n$ and we ensure that $n$ here is greater than 0 during typing,
as an expression \(e\) at level \(0\) is an ordinary expression, not a code fragment.
To use a program $e_1$ that generates a code template, we bind a contextual
variable $u$ to the result of evaluation of $e_1$ and use the result
in a body $e_2$. This is accomplished via
$\letboxm{(\Gammahat^n.u)}{e_1}{e_2}$.

A substitution is formed by extending a given substitution with a
contextual term $(\Psihat^n.e)$, a contextual type $(\Psihat^n.T)$, a term variable $x$,
or a type variable $\alpha$.
The latter extensions are necessary to allow us to treat substitutions as
simultaneous substitutions. In particular, when pushing a substitution $\sigma$
inside a function $\tfn{\alpha^n}{e}$, we need to extend it with a mapping for
$\alpha$. However, we would need to eta-expand $\alpha$ based on its type to
obtain a proper contextual object, as $\alpha$ by itself would not be a
legitimate type. Therefore, we allow extensions of substitutions with
an identity mapping, which is written as $\sigma ; \alpha^n$. 
Similarly,
if we push the substitution inside $\boxm {\Gammahat^n}{e}$ then we
may need to extend the simultaneous substitution with mappings for the
variables listed in $\Gammahat$, but we lack the type information to
expand a term variable $x$ to a proper contextual term. This issue is not new and is
handled similarly in \citet{Nanevski:TOCL08} in a slightly different
setting.

In contrast to the two-zone formulation of the modal lambda-calculus
in \citet{Nanevski:TOCL08}, which has different substitution
operations for ordinary bound variables (variables at level \(0\)) and
meta-variables (variables at level \(1\)), our substitution encompasses
both kinds of substitutions.

\paragraph{Remark 1:}
We define the level of a context: $\m{level}(\Psi) = n$ iff for all
declarations $x{:}(\_ \atl{k} \_)$ in $\Psi$ have $k < n$.
In other words, the level of context $\Psi$ ensures that all variables in $\Psi$ are at
strictly smaller levels. We often make the level explicit as a
superscript, $\Psi^n$, which describes a context $\Psi$ where
$\m{level}(\Psi) = n$. 
Note that we abuse the functional notation here, since $\m{level}$ 
does not have a unique result for a given context $\Psi$.\\[-2em]


\paragraph{Remark 2:} The erasure of type
annotations from the context $\Psi^n$ drops all the type and kind
declarations from $\Psi$, but retains the level to obtain $\Psihat^n$.  
\newcommand{\erase}[1]{\m{erase}(#1)}
\LONGVERSION{
\[
  \begin{array}{lcl}
    \erase{\cdot} & = & \cdot\\
    \erase{\Psi,x{:}(\Phi \atl{n}T)} & = & \erase{\Psi}, x^n\\
    \erase{\Psi,\alpha{:}(\Phi \atl{n}*)} & = & \erase{\Psi}, \alpha^n
  \end{array}
\]

}
Abusing notation, we often simply write $\Psihat^n$ for the erasure of type annotations
from the context $\Psi^n$\LONGVERSION{ instead of $\erase{\Psi} = \Psihat$}.\\[-2em]

\paragraph{Remark 3:} We write $\id(\Phihat)$ for the identity substitution
that has domain $\Phi$. It is defined as follows:
\newlength{\oldabovedisplayskip}
\newlength{\oldbelowdisplayskip}
\newlength{\oldabovedisplayshortskip}
\newlength{\oldbelowdisplayshortskip}
\setlength{\oldabovedisplayskip}{\abovedisplayskip}
\setlength{\oldbelowdisplayskip}{\belowdisplayskip}
\setlength{\oldabovedisplayshortskip}{\abovedisplayshortskip}
\setlength{\oldbelowdisplayshortskip}{\belowdisplayshortskip}
\setlength{\abovedisplayskip}{0.5em}
\setlength{\belowdisplayskip}{0.5em}
\setlength{\abovedisplayshortskip}{0.5em}
\setlength{\belowdisplayshortskip}{0.5em}
\[
  \begin{array}{lcl}
\id(\cdot)           & = & \cdot\\
\id(\Phihat,~x^n)         & = & \id(\Phihat); x^n\\
\id(\Phihat,~\alpha^n) & = & \id(\Phihat);\alpha^n
  \end{array}
\]

\LONGVERSION{\paragraph{Remark 4:} We sometimes write simply $x{:}T$ for
$x{:}(\Psi \atl{0} T)$ for readability; similarly we write simply
$x$ for $x[\id{(\Psihat)}]$, i.e. when $x$ is associated with the
identity substitution and $\m{erasure}(\Psi) = \Psihat$.}
\setlength{\abovedisplayskip}{\oldabovedisplayskip}
\setlength{\belowdisplayskip}{\oldbelowdisplayskip}
\setlength{\abovedisplayshortskip}{\oldabovedisplayshortskip}
\setlength{\belowdisplayshortskip}{\oldbelowdisplayshortskip}

Before moving to the typing rules, we take a closer look first at our typing
contexts and then at our substitutions, because they are both key parts of
the design of our language and they hopefully give a good general intuition
about the way the system works.

\subsection{Context operations}
%
First, we note that we maintain order in a context and hence
contexts must be sorted according to the level 
of assumptions $x{:}(\Psi \atl{n} T)$. One should conceptualize this
ordered context as a stack of sub-contexts, one for each level of
variables. Let $\Psi(k)$ be the subcontext of $\Psi^n$ with only
assumptions of level $k$. Then, $\Psi^n = \Psi(n-1), \Psi(n-2),
\ldots, \Psi(1), \Psi(0)$. Slightly abusing notation we write
$\ctxappend{\Psi(i)}{\Psi(i-1)}$ for appending the context $\Psi(i)$
and the context $\Psi(i-1)$. 
Note that $\Psi(0)$ contains only
declarations of type $(\cdot \atl{0} T)$ and we recover our ordinary typing context which
simply contains variable declarations $x{:}T$. The
context $\Psi^n$ is essentially an $n$-ary zone generalization of the two-zone context present in
\citet{Davies:ACM01} work. We can recover their two zones when
$n=1$. We opt here for a flattened presentation of the $n$-ary zones
of contexts in order to simplify operations on contexts. Keeping the
context sorted not only provides conceptual guidance, but also helps
us to handle
cleanly the dependencies among variable declarations.

Keeping the context sorted comes at a cost: extending a context must
preserve this invariant. However,
restricting a context, which is necessary as we move up one stage,
i.e. we move inside a box, is simpler. 
%
We therefore define and explain two main operations on
context: restricting a context and merging two contexts.
To chop off all variables below level $n$ from an ordered context $\Gamma$,
we write $\ctxrestrict{\Gamma}{n}$. 
%
To merge two ordered contexts $\Gamma$ and $\Psi$ we write $\ctxmerge{\Gamma}{\Psi}$.
With merging defined, insertion of a new assumption
into a context is a special case. For compactness, we write $K$ for a type or a
kind in the definition of the context operations in Fig.~\ref{fig:ctxoperations}.

\begin{figure}[tb]
\[
\begin{array}{lcll}
\multicolumn{3}{l}{\mbox{Merging contexts:} \quad \ctxmerge{\Psi}{\Phi} = \Gamma }\\
    \ctxmerge{\cdot}{\Phi} &= &\Phi \\
    \ctxmerge{\Psi}{\cdot} &= &\Psi \\
    \ctxmerge{\Psi, x {:} (\Gamma \atl{n} K)}{\Phi, y {:}({\Gamma'} \atl{k} K')} &= &
    (\ctxmerge{\Psi, x {:} (\Gamma\atl{n} K)}{\Phi}), y {:}({\Gamma'} \atl{k} K')
    &\text{if $k  \leq n$} 
    \\
    \ctxmerge{\Psi, x {:}(\Gamma\atl{n} K)}{\Phi, y {:}({\Gamma'}\atl{k} K')} &=&
    (\ctxmerge{\Psi}{\Phi, y {:}({\Gamma'}\atl{k} K')}), x {:}({\Gamma}\atl{n} K)
    &\text{otherwise}
\\[0.75em]
\multicolumn{3}{l}{\mbox{Chopping lower context:} \quad\ctxrestrict{\Psi}{n}~ =~  \Phi
}\\
    \ctxrestrict{(\cdot)}n &= & \cdot \\ \relax
    \ctxrestrict{(\Psi, x {:} (\Phi \atl{k}K)}n &= & \ctxrestrict{\Psi}n
    &\text{if $k < n$}
    \\
    \ctxrestrict{(\Psi, x {:} (\Phi \atl{k}K))}n &= & \Psi, x {:} (\Phi \atl{k} K)
    &\text{otherwise}
  \end{array}
\]
  \caption{Merging and chopping of contexts}  \label{fig:ctxoperations}
\end{figure}
%
The chopping operation for the lower context allows us to drop all variable
assumptions below the given level from a context. If $k \leq n$, then
$\ctxrestrict{(\Psi^k)}{n} = \cdot$. 
%

Merging of two independent, sorted contexts is akin to the merge step of the
merge-sort algorithm and therefore inherits many of its properties. In
particular, the merge of two sorted independent contexts is again a
sorted context. It is also stable, in the sense that the relative
positions of any two assumptions in $\Psi^n$ or in $\Phi^k$ is
preserved in $\ctxmerge{\Psi^n}{\Phi^k}$.

When we merge two sorted contexts \(\Psi\) and \(\Phi\), i.e.~when for every
declaration $x{:}(\_ \atl{n} \_)$ in \(\Psi\), we have
\(\m{level}(\Phi) \le n + 1\), we simply write $\ctxappend{\Psi}{\Phi}$
(see \LONGSHORT{Appendix \ref{sec:ctxop}}{\LONGVERSIONREF} for the definition).
%
While defining a separate
append operation on sorted contexts is not necessary, it is often conceptually
easier to read and use. It is also more concise.

\subsection{Simultaneous Substitution Operation}

\newcommand \SSubst [2] {[#1]#2}
\newcommand \LKP [2] {\textrm{lkp}(#1)~#2}


We concentrate here on the simultaneous substitution operation, but
the single substitution operation follows similar ideas. 

In our grammar, simultaneous substitutions are defined without their
domain. However, when applying a simultaneous substitution $\sigma$,
we can always recover its domain $\Psihat$. The idea is that $\sigma$
provides instantiations for all variables in $\Psihat$ where
$\m{level}(\Psihat) = n$. All variables at levels greater or
equal to $n$ are treated as global variables and are untouched by the
substitution operation. We note that since contexts are ordered,
simultaneous substitutions are also ordered. 

Applying the substitution operation is written in prefix as
$\SSubst{\sigma/\Psihat}{T}$ (see \LONGSHORT{Fig.~\ref{fig:typsubst}}{\LONGVERSIONREF}). 
Note that the closure of a variable with a
substitution is written in postfix (see $x[\sigma]$ and
$\alpha[\sigma]$ resp.). 

We describe here in detail applying a substitution to a type 
$\SSubst{\sigma/\Phihat}{T}$. Due to space, we omit the definition for
$\SSubst{\sigma/\Phihat}{e}$ which applies the substitution $\sigma$
to a term $e$, composing substitutions, and applying a substitution to a context. They
can be found in \LONGSHORT{the Appendix \ref{sec:substerm}}{\LONGVERSIONREF}.
In the definition of these substitution operations, we rely on 
chopping, merging, and appending simultaneous substitutions, written
$\ctxrestrict{(\sigma/\Phihat)}{k}$,
$\ctxmerge{(\sigma/\Phihat)}{(\sigma'/\Psihat)}$ and
$\ctxappend{(\sigma/\Phihat)}{(\sigma'/\Psihat)}$. Those operations
correspond to the equivalent context operations and are also defined in
\LONGSHORT{the Appendix \ref{sec:substop}}{\LONGVERSIONREF}.


\begin{figure}[ht]
  \begin{displaymath}
    \renewcommand{\arraystretch}{1.2}
    \begin{array}{@{}l@{~=~}l@{}}
      \SSubst{\sigma/\Psihat}{(\alpha[\sigma'])} &
          \alpha[\sigma''] \hfill
          \alpha \not\in \hat\Psi \text{ and } \SSubst{\sigma/\Psihat}{\sigma'} = \sigma''
          \\
      \SSubst{\sigma/\Psihat}{(\alpha[\sigma'])} &
          T' \qquad\quad \hfill
          \LKP{\sigma /\Psihat}{\alpha} = (\Phihat^n.T)~ \mbox{and}~
          \m{level}(\Psihat) > n\\
     \multicolumn{2}{r}{\qquad\quad     \mbox{and}~
               \SSubst{\sigma/\Psihat}{\sigma'} = \sigma''
      ~\mbox{and}~ \SSubst{\sigma''/\Phihat}{T} =T'}
          \\
      \SSubst{\sigma/\Psihat}{(\alpha[\sigma'])} &
          \beta[\sigma''] \hfill  
          \LKP{\sigma /\Psihat}{\alpha} = \beta^n~\mbox{and}~
          \m{level}(\Psihat) > n~\mbox{and}~
               \SSubst{\sigma/\Psihat}{\sigma'} = \sigma''
          \\
      \SSubst{\sigma/\Psihat}{(T \arrow S)} &
          T' \arrow S' \hfill
          \SSubst{\sigma/\Psihat}{T} = T' ~\mbox{and}~ \SSubst{\sigma/\Psihat}{S} = S'
          \\
      \SSubst{\sigma/\Psihat}{((\alpha{:}(\Phi\atl{n} *)) \arrow T)} &
          ((\alpha{:}(\Phi\atl{n} *)) \arrow T')
          \hfill\qquad 
          \mbox{level}(\Psihat) > n \text{ and } \SSubst{\ctxinsert{(\sigma/\Psihat)}{(\alpha/\alpha^n)}}{T} = T'
          \\
      \SSubst{\sigma/\Psihat}{((\alpha{:}(\Phi\atl{n} *)) \arrow T)} &
          ((\alpha{:}(\Phi\atl{n} *)) \arrow T')
          \hfill
          \mbox{level}(\Psihat) \le n \text{ and } \SSubst{\sigma/\Psihat}{T} = T'
          \\
      \SSubst{\sigma/\Psihat}{(\cbox{\Phi \atl{n} T})} &
          \cbox{\Phi' \atl{n} T'} \hfill
          \begin{array}[t]{@{}r@{}}
            \mbox{level}(\Psihat) \ge n \text{ and } \ctxrestrict{(\sigma/\Psihat)}{n} = \sigma'/\Psihat' \text{ and } \\
            {\SSubst{\sigma'/\Psihat'}{\Phi} = \Phi'} \text{ and }
            \SSubst{\ctxappend{(\sigma'/\Psihat)}{(\id(\Phihat)/\Phihat)}}{T} = T'
          \end{array}
          \\
      \SSubst{\sigma/\Psihat}{(\cbox{\Phi \atl{n} T})} &
          \cbox{\Phi \atl{n} T} \hfill
          \begin{array}[t]{@{}r@{}}
            \mbox{level}(\Psihat) < n
          \end{array}
    \end{array}
  \end{displaymath}
  \caption{Simultaneous Substitution Operation for Types:
    $[\sigma/\Psihat]T = S$}
  \label{fig:typsubst}
\end{figure}

The substitution operation $[\sigma/\Psihat]T$ is then mostly straightforward and applied
recursively to $T$. In the variable case $\alpha[\sigma']$, we distinguish between three cases: 

 If $\alpha$ is not in $\Psihat$, then $\alpha$ denotes a ``global''
 variable; we leave $\alpha$ untouched and only apply $\sigma$ to $\sigma'$.

If $\alpha$ is in $\Psihat$, then we have either
a corresponding instantiation $(\Phihat^n.T)/\alpha^n$ or simply
$\beta^n$ for $\alpha$ in $\sigma/\Psihat$ (see \LONGSHORT{Appendix \ref{sec:substop}}{\LONGVERSIONREF} for the
definition of the lookup operation, $\LKP{(\sigma/\Psihat)} \alpha$). In both cases, we apply the
simultaneous substitution $\sigma$ to $\sigma'$ giving us some substitution
$\sigma''$. In the former case (where $(\Phihat^n.T)/\alpha^n$),
we now apply the substitution $\sigma''$ to $T$. Applying the
substitution will terminate, as $\sigma''$ provides instantiations for variables
at lower levels than $n$. 
In the latter case (where $\beta^n/\alpha^n$), we simply pair $\beta$ with
$\sigma''$ creating a new closure.

For polymorphically
quantified types, we push the substitution inside and extend the
substitution with the identity, if its level \(n\) satisfies
\(n < \mbox{level}(\Psihat)\). Although we quantify over (contextual)
type variables that have type $(\Phi \atl{n} *)$, we do not have to apply
the substitution to $\Phi$ itself, as $\Phi$ is a context containing only
type variable declarations\footnote{Generalizing this to contain both type
and term variable declarations is straightforward and follows the similar
principle as in the boxed type case where we apply to $\Phi$ the substitution
$\ctxrestrict{(\sigma/\Psihat)}{n}$.}.
For boxed types such as $\cbox{\Phi \atl{n} T}$, we push $\sigma$
inside the box, if $\m{level}(\Psihat) \geq n$. To do this,
we first drop
all the mappings for variables below $n$ from the substitution
$\sigma$,  since those variables will be replaced by $\Phi$. Then,  
we apply 
$\sigma'/\Psihat'$ to $\Phi$ 
(\LONGSHORT{see Fig.~\ref{fig:ctxsubst} in the Appendix \ref{sec:substop}}{see \LONGVERSIONREF for the definition}). Further, we extend $\sigma'$  with the
identity mapping for variables in $\Phi$ before applying it to $T$.
If \(\m{level}(\Psihat) \le n \),
then we do not apply the substitution to $\Phi$ or $T$, since it
concerns variables that will be replaced by $\Phi$ and those variables
are locally bound.


\begin{theorem}[Termination]\quad
The substitution operation $[\sigma/\Phihat]T$, $[\sigma/\Phihat]e$, and $[\sigma/\Phihat]\sigma'$ terminates.
\end{theorem}
\begin{proof}
By induction on the $\m{level}(\Phihat)$ and structure of $T$
and $e$ and $\sigma'$. Either $\m{level}(\Phihat)$  stays the same and
$T$ (resp. $e$ or $\sigma$) is decreasing or $\m{level}(\Phihat)$ is decreasing.
\end{proof}

\subsection{Typing Rules}\label{sec:typrules}
\newcommand{\vde}{\mathrel{\vdash{\!\!\!\!\!}\vdash}}
With the context and substitution operations in place, we can now
define the typing rules for the multi-level contextual modal
lambda-calculus in an elegant way. The levels provide us with enough
structure to keep track of the scope of variables, terms and types.

To distinguish the notation for contextual types, written
as $\cbox{\Psi \atl{n} T}$ from the typing judgment, we use $\vde$ for
all the typing judgments.

%

We begin with defining well-formed contexts.
A context
$\Gamma, x{:}(\Phi \atl{n} T)$ is well-formed, if $\Gamma$ is 
well-formed, $\Phi$ is well-formed with respect to
$\ctxrestrict{\Gamma}{n}$, and $T$ is well-kinded in the
context $(\ctxappend{\ctxrestrict{\Gamma}{n}}{\Phi})$.
\[
\begin{array}{c}
\multicolumn{1}{l}{\mbox{Well-formed contexts: \fbox{~$\vde \Gamma$}}}\\[0.75em]
\infer{\vde \ctxappend{\Gamma}{x{:}(\Phi \atl{n} T)}}
{\vde \Gamma &
 \m{level}(\Phi) \leq n &  
 \vde \ctxappend{\ctxrestrict{\Gamma}{n}}{\Phi} &
 \ctxrestrict{\Gamma}n,\Phi \vde T
}
\quad
\infer{\vde \ctxappend{\Gamma}{\alpha{:}(\Phi \atl{n} *)}}
{\vde \Gamma &
 \m{level}(\Phi) \leq n &  
 \vde \ctxappend{\ctxrestrict{\Gamma}{n}}{\Phi}
}
\quad
\infer{\vde \cdot}{}
\end{array}
\]

With the context operations in place, the kinding rules for
types
are straightforward. 
\[
\begin{array}{c}
\raisebox{2ex}{\mbox{Kinding rules for types: \fbox{$\Gamma \vde T$}}}\quad\qquad\qquad
\infer{\Gamma \vde \alpha[\sigma]}
      {\Gamma(\alpha) = (\Phi \atl{n} *) &
       \Gamma \vde \sigma : \Phi}
\\[0.75em]
\infer{\Gamma \vde S \arrow T}
      {\Gamma \vde S &
       \Gamma \vde T}
\quad
\infer{\Gamma \vde (\alpha {:}(\Phi \atl{n} *)) \arrow T}
{\vde \ctxrestrict{\Gamma}{n},\Phi &
 \ctxinsert{\Gamma}{\alpha {:}(\Phi \atl{n} *)} \vde T}
\quad
\infer[n > 0]{\Gamma \vde \cbox{\Phi \atl{n} T}}
{ \vde \ctxrestrict{\Gamma}{n},\Phi &
  \ctxrestrict{\Gamma}{n},\Phi \vde T
}
\end{array}
\]
\smallskip

Type variables $\alpha[\sigma]$ where $\alpha: (\Phi \atl{n} *)$ in
$\Gamma$ are well-kinded, if the associated substitution $\sigma$
maps variables from $\Phi$ to the present context $\Gamma$. Function
types, $S \arrow T$ are well-kinded, if both $S$ and $T$ are well-kinded.
When we check that
$(\alpha {:}(\Phi \atl{n} *)) \arrow T$ is well-kinded in a context $\Gamma$,
we ensure that $\Phi$ is well-formed with respect to $\ctxrestrict{\Gamma}{n}$,
as the local context $\Phi$ will replace all the variables at levels below $n$ in
the original $\Gamma$; this is accomplished by $\vde \ctxrestrict{\Gamma}{n},\Phi$.
We then check that $T$ is well-formed, in the context $\Gamma$
extended with the assumption $\alpha {:}(\Phi \atl{n} *)$ where the
assumption is inserted at the appropriate position in $\Gamma$.
For the kinding of $\cbox{\Phi \atl{n} T}$, we proceed similarly. We 
again replace all the variables at levels below $n$ in
the original $\Gamma$ with $\Phi$ using the previously defined context
operations.

Using the context operations for restricting and extending contexts,
the typing rules for expressions (Fig.~\ref{fig:typing}) are mostly straightforward and follow
the general principle that we have seen before. To show that a function
$\fn{x}{e}$ has type $S \arrow T$, we extend the context $\Gamma$ with the
declaration $x{:}(\cdot \atl{0} S)$ and check that the body $e$ has
  type $T$. As we remarked before, no typing context with assumptions
  below level $0$ exists, and this declaration can be viewed as
  $x{:}T$. However, giving it level $0$ allows for a uniform treatment.
To show that a type abstraction,  $\tfn{\alpha^n}{e}$, has the
appropriate type, we extend the context with type variable declaration
for $\alpha$.

Applications, $e_1~e_2$ are defined as expected. For type
applications, $ e\;(\Phihat^n.T)$, we need to be a bit more careful:
in addition to verifying that $e$ has type $(\alpha{:}(\Phi \atl{n} *)) \arrow S$,
we again verify that $T$ is well-kinded in a new context, where we
replace the assumptions at level below $n$ in $\Gamma$ with the
declarations from $\Phi$. To accomplish this, we again rely on our
context operations, first restricting $\Gamma$ and then appending
$\Phi$.  As the type of $e$ will be polymorphic, namely
$(\alpha{:}(\Phi \atl{n} *)) \arrow S$, we return as the type of
 $e\;(\Phihat^n.T)$ the type $S$ where we have replaced $\alpha$
with $(\Phihat^n.T)$.

The rules for $\m{box}$ and $\m{let box}$ again appropriately restrict
and extend $\Gamma$ based on the level where $n > 0$. This side
condition simply allows us to distinguish between code, i.e. terms
that are boxed and represent syntax, and programs, i.e. terms that
will be evaluated and run.

\begin{figure}
\[
\begin{array}{c}
\infer{\Gamma \vde x[\sigma] : [\sigma/\Phihat]T}
{\Gamma(x) = (\Phi \atl{n} T) & \Gamma \vde \sigma : \Phi}
\quad
\infer{\Gamma \vde \fn{x}{e} : S \arrow T}
      {\Gamma, x{:}(\cdot \atl{0} S) \vde e : T}
\qquad
\infer{\Gamma \vde e_1\;e_2 : T}
      {\Gamma \vde e_1 : S \arrow T
     & \Gamma \vde e_2 : S}
\\[0.75em]
\infer{\Gamma \vde \tfn{\alpha^n}{e} : (\alpha{:}(\Phi \atl{n} *)) \arrow T}
      {\ctxinsert{\Gamma}{\alpha{:}(\Phi \atl{n}*)} \vde e : T}
\qquad
\infer{\Gamma \vde e\;(\Phihat^n.T) : [\Phihat^n.T/\alpha^n]S}
      {\Gamma \vde e : (\alpha{:}(\Phi \atl{n} *)) \arrow S
     & \ctxappend{\ctxrestrict{\Gamma}{n}}{\Phi} \vde T}
\\[0.75em]
\infer{\Gamma \vde \boxm{\Phihat^n}{e} : \cbox{\Phi \atl{n} T}}{
\ctxappend{\ctxrestrict{\Gamma}{n}}{\Phi} \vde e : T}
\qquad
\infer{\Gamma \vde \letboxm{(\Phihat^n.u)}{e_1}{e_2} : T}
      {\Gamma \vde e_1 : \cbox{\Phi \atl{n} S} &
      \ctxinsert{\Gamma}{u{:}(\Phi \atl{n} S)} \vde e_2 : T}
\end{array}
\]
\caption{Typing rules for expressions \fbox{$\Gamma \vde e : T$}}
\label{fig:typing}
\end{figure}

A substitution $(\sigma, \Psihat^n.e)$ provides a mapping from
$(\Gamma', x{:}(\Psi \atl{n} T))$ to the context $\Gamma$, if $\sigma$
maps variables from $\Gamma'$ to $\Gamma$ and if $e$ is well-typed.
As the type $(\Psi \atl{n} T)$ is well-formed with respect
to $\Gamma'$, we will need to transport both $\Psi$ and $T$ to
$\Gamma$. Since $\Psi$ depends only on $\ctxrestrict{\Gamma'}{n}$, we
only need to apply the restricted substitution
$\ctxrestrict{(\sigma/\Gammahat')}{n} = (\sigma''/\Gammahat'')$ to $\Psi$. Since we work with
simultaneous substitution, we need to extend this restricted
substitution with the identity mapping for $\Psi$, before we can apply
it to the type $T$ and effectively move the type to the new context
$(\ctxrestrict{\Gamma}{n},[\sigma''/\Gammahat'']\Psi)$ in
which $e$ will be well-typed.
The cases where we extend a substitution with a
variable (i.e. $\sigma;x^n$) is a special case
of the previous cases. 
The extension of a substitution with a type $(\Psihat^n.T)$ or a type
variable $\alpha^n$ follows the similar idea and are a special case of
handling the term extension. We omit them here due to space, but their
definition is given in \LONGSHORT{the Appendix \ref{sec:subtyp}}{\LONGVERSIONREF}.

\begin{figure}
\[
\begin{array}{c}
\infer{\Gamma \vde \cdot : \cdot}{}
\quad
\infer{\Gamma \vde (\sigma, \Psihat^n.e) ~:~ (\Gamma', x{:}(\Psi \atl{n} T))}{
 \Gamma \vde \sigma : \Gamma' &
\ctxrestrict{(\sigma/\Gammahat')}{n} = (\sigma'/\Gammahat'') &
\ctxappend{\ctxrestrict{\Gamma}{n}}{[\sigma'/\Gammahat'']\Psi} \vde
   e : [\ctxappend{(\sigma'/\Gammahat'')}{(\id(\Psihat)/\Psihat)}]T}
\\[0.75em]
\infer{\Gamma \vde (\sigma;x) ~:~ (\Gamma', x{:}(\Psi \atl{n} T))}
{\Gamma \vde \sigma : \Gamma' &
 \ctxrestrict{(\sigma/\Gammahat')}{n} = (\sigma''/\Gammahat'') &
 \Gamma(x) = ([\sigma''/\Gammahat'']\Psi \atl{n} [\sigma''/\Gammahat'']T)
}
\LONGVERSION{\\[0.75em]
\infer{\Gamma \vde (\sigma, \Psihat^n.T) ~:~ (\Gamma', \alpha{:}(\Psi \atl{n} *))}{
 \Gamma \vde \sigma : \Gamma' &
\ctxrestrict{(\sigma/\Gammahat')}{n} = (\sigma'/\Gammahat'') &
\ctxappend{\ctxrestrict{\Gamma}{n}}{[\sigma'/\Gammahat'']\Psi} \vde T
}
\\[0.75em]
\infer{\Gamma \vde (\sigma;\alpha) ~:~ (\Gamma', \alpha{:}(\Psi \atl{n} *))}
{\Gamma \vde \sigma : \Gamma' &
 \ctxrestrict{(\sigma/\Gammahat')}{n} = (\sigma''/\Gammahat'') &
 \Gamma(\alpha) = ([\sigma''/\Gammahat'']\Psi \atl{n} *)
}
}
\end{array}
\]
  \caption{Typing rules for substitutions \fbox{$\Gamma \vde \sigma : \Phi$}}
  \label{fig:substtyping}
\end{figure}



\subsection{Substitution properties}
\newcommand{\ctxapp}[2]{\ctxappend{#1}{#2}}


We now establish substitution properties. Recall that a context $\Psi^k = \Psi(k-1), \Psi(k-2), \ldots, \Psi(1), \Psi(0)$ where $\Psi(j)$ is the subcontext of $\Psi^k$ with only assumptions of level $j$.
For easier readability, we will simply write $\ctxapp{\ctxapp{\Gamma_1}{\alpha{:}(\Phi \atl{n})}}{\Gamma_0}$ for a context where all assumptions in $\Gamma_1$ are at level $n$ or above and  all assumptions in
  $\Gamma_0$ are at level $n$ or below.

We first state that the generation of identity substitutions yields
well-typed substitutions. 

\begin{lemma}[Identity Substitution]\label{lem:idsubst}
$\ctxapp{\Gamma}{\Phi} \vde \id(\Phihat) : \Phi$    
\end{lemma}
\begin{proof}
Induction on $\Phi$.    
\end{proof}

Next, we generalize a property that already exists in
\citet{Nanevski:TOCL08} and \citet{Pientka03phd}. In this previous
work, it states that the substitution for locally bound variables and
the substitution operation for meta-variables commute. Here we state
more generally that substitution for a variable with level $m$ and
substitution for a variable with level $n$ commute, as long as $m$ is
greater or equal to $n$. 

\begin{lemma}[Commuting Substitutions]\label{lem:commsubst}\quad
  \begin{enumerate}
  \item If \(m \ge n\), then \([(\hat\Psi^m.S_2)/\beta^m][(\Phihat^n.S_1)/\alpha^n]T = [(\Phihat^n.[(\hat\Psi^m.S_2)/\beta^m]S_1)/\alpha^n][(\hat\Psi^m.S_2)/\beta^m]T\).
  \item If \(m \ge n\), then \([(\hat\Psi^m.e_2)/y^m][(\Phihat^n.e_1)/x^n]e_0 = [(\Phihat^n.[(\hat\Psi^m.e_2)/y^m]e_1)/x^n][(\hat\Psi^m.e_2)/y^m]e_0\).
  \end{enumerate}
\end{lemma}
\begin{proof}
By induction on \(T\) and \(e_0\).
\end{proof}

Last, we state the substitution properties for type
variables, term variables and simultaneous substitutions. Since the
single substitution operation relies on simultaneous substitution, we
prove them all mutually. However, for clarity, we state them
separately.

\begin{lemma}[Type Substitution Lemma]\label{lem:typesubst}
Assuming
  $\vde \ctxapp{\ctxapp{\Gamma_1}{\alpha{:}(\Phi \atl{n}
      *)}}{\Gamma_0}$ and $\ctxapp{\Gamma_1}{\Phi} \vde T$.
  \begin{enumerate}
  \item\label{itm:typesubst-context}
Then $\vde \ctxapp{\Gamma_1}{[(\Phihat^n.T)/\alpha^n]\Gamma_0}$.
  \item\label{itm:typesubst-type} If $\ctxapp{\ctxapp{\Gamma_1}{\alpha{:}(\Phi \atl{n} *)}}{\Gamma_0} \vde S$
then $\ctxapp{\Gamma_1}{[(\Phihat^n.T)/\alpha^n]\Gamma_0} \vde [(\Phihat^n.T)/\alpha^n]S$.
  \item\label{itm:typesubst-term} If $\ctxapp{\ctxapp{\Gamma_1}{\alpha{:}(\Phi \atl{n} *)}}{\Gamma_0} \vde e_2 : S$
then $\ctxapp{\Gamma_1}{[(\Phihat^n.T)/\alpha^n]\Gamma_0} \vde [(\Phihat^n.T)/\alpha^n]e_2 : [(\Phihat^n.T)/\alpha^n]S$.
  \item\label{itm:typesubst-subst} If $\ctxapp{\ctxapp{\Gamma_1}{\alpha{:}(\Phi \atl{n} *)}}{\Gamma_0} \vde \sigma : \Psi$
then $\ctxapp{\Gamma_1}{[(\Phihat^n.T)/\alpha^n]\Gamma_0} \vde [(\Phihat^n.T)/\alpha^n]\sigma : [(\Phihat^n.T)/\alpha^n]\Psi$.
  \end{enumerate}
\end{lemma}

\begin{lemma}[Term Substitution Lemma]\label{lem:termsubst} Assuming 
$~\vde \ctxapp{\ctxapp{\Gamma_1}{u{:}(\Phi \atl{n} T)}}{\Gamma_0} $
and  $\ctxapp{\Gamma_1}{\Phi} \vde e : T$.
  \begin{enumerate}
  \item\label{itm:termsubst-term} If $\ctxapp{\ctxapp{\Gamma_1}{u{:}(\Phi \atl{n} T)}}{\Gamma_0} \vde e_2 : S$
then $\ctxapp{\Gamma_1}{\Gamma_0} \vde [(\Phihat^n.e)/u]e_2 : S$.
  \item\label{itm:termsubst-subst} If $\ctxapp{\ctxapp{\Gamma_1}{u{:}(\Phi \atl{n} T)}}{\Gamma_0} \vde \sigma : \Psi$
then $\ctxapp{\Gamma_1}{\Gamma_0} \vde [(\Phihat^n.e)/u]\sigma : \Psi$.
  \end{enumerate}
\end{lemma}

\begin{lemma}[Simultaneous Substitution Lemma]\label{lem:simsubst}
  Assuming $\ctxapp{\Gamma_1}{\Phi} \vde \sigma : \Gamma_0$.
  \begin{enumerate}
  \item\label{itm:simsubst-context} If $~\vde
    \ctxapp{\ctxapp{\Gamma_1}{\Gamma_0}}{\Psi}$ 
then $\vde \ctxapp{\ctxapp{\Gamma_1}{\Phi}}{([\sigma/\hat\Gamma_0]\Psi)}$.
  \item\label{itm:simsubst-type} If $\ctxapp{\Gamma_1}{\Gamma_0} \vde S$ 
then $\ctxapp{\Gamma_1}{\Phi} \vde [\sigma/\hat\Gamma_0]S$.
  \item\label{itm:simsubst-term} If $\ctxapp{\Gamma_1}{\Gamma_0} \vde
    e : S$
then $\ctxapp{\Gamma_1}{\Phi} \vde [\sigma/\hat\Gamma_0]e : [\sigma/\hat\Gamma_0]S$.
  \item\label{itm:simsubst-subst} If $\ctxapp{\Gamma_1}{\Gamma_0} \vde
    \sigma_2 : \Psi$ 
then $\ctxapp{\Gamma_1}{\Phi} \vde [\sigma/\hat\Gamma_0]\sigma_2 : [\sigma/\hat\Gamma_0]\Psi$.
  \end{enumerate}
\end{lemma}
\begin{proof}
  By induction on the first derivation. 
For more details, see \LONGSHORT{Appendix \ref{sec:proofsubsts}}{\LONGVERSIONREF}.
\end{proof}


\subsection{Local Soundness and completeness}\label{sec:locsound}
With the substitution properties in place, we establish local
soundness and completeness properties of our calculus. Local soundness
guarantees that our typing rules are not too
strong, i.e. they do not allow us to infer more than we
should. Dually, local completeness guarantees that the rules are
sufficiently strong.

The local soundness property also gives natural rise to reduction
rules in our operational semantics and can be viewed as showing
that types are preserved during reduction.  

\paragraph{Function type}
For local soundness, we refer to the term substitution lemma
\ref{lem:termsubst} to obtain $\mathcal{D'}$.
\smallskip

\emph{Local Soundness}: \\[-1.75em]
\[
%
 \begin{array}[b]{c}
\infer{\Gamma \vde (\fn{x}{e})~e' : T \arrow S}
{
   \begin{array}[b]{c}
     \mathcal{D}\\
     \Gamma, x{:}(\cdot \atl{0} T) \vde e
   \end{array}
    &
      \begin{array}[b]{c}
        \mathcal{E}\\
        \Gamma \vde e' : T
      \end{array}
}
  \end{array}
    \Longrightarrow
  \begin{array}[b]{c}
  \mathcal{D'}\\
  \Gamma \vde  [(\cdot . e') / x^0]e : S
  \end{array}
\]

Local completeness derivation follows directly from the given typing
rules.
\smallskip

\emph{Local Completeness}: \\[-1.75em]
\[
  \begin{array}[b]{c}
\mathcal{D}\\
\Gamma \vde e: T \arrow S
\end{array}
\Longrightarrow
\begin{array}[b]{c}
\infer{\Gamma \vde \fn{x}{e~x[\cdot]} : T \arrow S}
{\infer{\Gamma,x{:}(\cdot \atl{0} T) \vde e ~ x[\cdot] : S}
  {
  \begin{array}[b]{c}
    {\mathcal{D}} \\
    {\Gamma,x{:}(\cdot \atl{0} T) \vde e : T \arrow S}
  \end{array}
        &
  \begin{array}[b]{c}
    \infer{\Gamma, x{:}(\cdot \atl{0} T) \vde x[\cdot] : T}{
    {\infer{\Gamma, x{:}(\cdot \atl{0} T) \vde \cdot : \cdot }{}
    }
    }
  \end{array}
  }
}
\end{array}
\]

\paragraph{Polymorphic type}
The local soundness and completeness property for the polymorphic type
are similar. 

\smallskip
\emph{Local Soundness}: \\[-1.75em]
\[
 \begin{array}[b]{c}
\infer{\Gamma \vde (\tfn{\alpha^n}{e})~(\Phihat^n.T) : [\Phihat^n.T/\alpha^n]S}
{
  \infer{\Gamma \vde \tfn{\alpha^n}{e} : (\alpha{:}(\Phi \atl{n} *)) \arrow S}{
   \begin{array}[b]{c}
     \mathcal{D}\\
     \ctxinsert{\Gamma}{\alpha{:}(\Phi  \atl{n} *)} \vde e : S
   \end{array}}
    &
      \begin{array}[b]{c}
        \mathcal{E}\\
         \ctxappend{\ctxrestrict{\Gamma}{n}}{\Phi} \vde T
      \end{array}
}
  \end{array}
\Longrightarrow
\begin{array}[b]{c}
\mathcal{D'}\\
\Gamma \vde  [(\Phihat^n . T) / \alpha^n]e : [\Phihat^n.T/\alpha^n]S
\end{array}
\]

We note that $\ctxinsert{\Gamma}{\alpha{:}(\Phi  \atl{n} *)}$ results in an ordered context of the form $\Gamma_1, \alpha{:}(\Phi  \atl{n} *),\Gamma_0$ where all declarations in $\Gamma_1$ are at level $n$ or above and all declarations in $\Gamma_0$ are below $n$.
The derivation $\mathcal{D'}$ is then obtained using the type substitution lemma \ref{lem:typesubst}.
\smallskip

    {Local Completeness}: \\[-1.95em]
\[
  \begin{array}[b]{c}
\mathcal{D}\\
\Gamma \vde e: (\alpha{:}(\Phi \atl{n} *)) \arrow S
\end{array}
\Longrightarrow
\hspace{-5pt}
\begin{array}[b]{c}
\infer{\Gamma \vde \tfn{\alpha^n}{e~(\Phihat^n.\alpha[\id(\Phihat)])} :  (\alpha{:}(\Phi \atl{n} *)) \arrow S}
{\infer{\ctxinsert{\Gamma}{\alpha{:}(\Phi \atl{n} *)} \vde e
  ~(\Phihat^n.\alpha[\id(\Phihat)]) : S}
  {
  \begin{array}[b]{@{}c}
    {\mathcal{D}} \\
    \ctxinsert{\Gamma}{\alpha{:}(\Phi \atl{n} *)} \vde e : (\alpha{:}(\Phi \atl{n} *)) \arrow S
  \end{array}
        &
  \begin{array}[b]{c@{}}
    \infer{\ctxappend{\ctxrestrict{\Gamma}{n}}{\Phi} \vde \alpha[\id(\Phihat)] }
    {\infer{\ctxappend{\ctxrestrict{\Gamma}{n}}{\Phi} \vde \id(\Phihat) : \Phi}{}
    }
  \end{array}
          }
}
\end{array}
\]

We use the identity substitution lemma \ref{lem:idsubst} to justify the derivation 
$\ctxappend{\ctxrestrict{\Gamma}{n}}{\Phi} \vde \alpha[\id(\Phihat)]$ and exploit the fact that $[(\Phihat^n.\alpha[\id(\Phihat)])/\alpha^n]S = S$.

\paragraph{Multi-level contextual type} Local soundness and
completeness follows similar ideas.
For local soundness, we note that $\ctxinsert{\Gamma}{u{:}(\Phi \atl{n} T)}$ results in an ordered context $\Gamma'$ where the declaration $u{:}(\Phi \atl{n} T)$ is inserted at the appropriate level. The derivation $\mathcal{D'}$ is then obtained by referring to the term substitution lemma \ref{lem:termsubst} using $(\Phihat^n.e)$ for $u$, $\mathcal{E}$, and $\mathcal{D}$.
\smallskip

\emph{Local Soundness}: \\[-1.5em]
\[
  \begin{array}[b]{c}
\infer{\Gamma \vde \letboxm{(\Phihat^n.u)}{\boxm{\Phihat^n}{e}}{e_2} : S}
{\infer{\Gamma \vde \boxm{\Phihat^n}{e} : \cbox{\Phi \atl{n} T}}{
    \begin{array}[b]{c}
      \mathcal{E}\\
      \ctxapp{\ctxrestrict{\Gamma}{n}}\Phi \vde e : T
    \end{array}
 }
  &
   \begin{array}[b]{c}
     \mathcal{D}\\
     \ctxinsert{\Gamma}{u{:}(\Phi \atl{n} T)} \vde e_2 : S
   \end{array}
    }
  \end{array}
  \Longrightarrow
  \begin{array}[b]{c}
\mathcal{D'}\\
\Gamma \vde [(\Phihat^n.e)/u^n]e_2 : S
  \end{array}
\]

For local completeness, we again use Lemma \ref{lem:idsubst} for identity substitutions to justify the derivation
$\ctxapp{\ctxrestrict{(\ctxinsert{\Gamma}{{u{:}(\Phi \atl{n} T)}})}{n} }\Phi \vde u[\id(\Phihat)] : T$.
\medskip

\emph{Local Completeness}: \\[-1.75em]
\[
  \begin{array}[b]{c}
\mathcal{D}\\
\Gamma \vde e : \cbox{\Phi \atl{n} T}
  \end{array}
  \Longrightarrow
  \begin{array}[b]{c}
\infer{\Gamma \vde \letboxm{(\Phihat^n.u)}{e}{\boxm{\Phihat}{u[\id(\Phihat)]}} : \cbox{\Phi \atl{n} T}}
{
    \begin{array}[b]{c}
      \mathcal{D}\\
      \Gamma \vde e : \cbox{\Phi \atl{n} T      }
    \end{array}
&
\infer
{\ctxinsert{\Gamma}{u{:}(\Phi \atl{n} T)} \vde \boxm{\Phihat}{u[\id(\Phihat)]} : \cbox{\Phi \atl{n} T}}
{\infer
  {\ctxapp{\ctxrestrict{(\ctxinsert{\Gamma}{{u{:}(\Phi \atl{n} T)}})}{n} }\Phi \vde u[\id(\Phihat)] : T}
  {\infer{\ctxapp{\ctxrestrict{(\ctxinsert{\Gamma}{{u{:}(\Phi \atl{n} T)}})}{n} }\Phi \vde \id(\Phihat) : \Phi}{}}
}
}
  \end{array}
\]

%






\section{Extension to pattern matching}\label{sec:PatMatch}
We now extend the multi-level programming foundation to support pattern matching on code, i.e. expressions of type $\cbox{\Gamma \atl{n} T}$. This generalizes the expressions $\letboxm{(\Gammahat^n.u)}{e_1}{e_2}$ to support not only binding of code to a variable $u$, but also inspecting this piece of code by pattern matching. In general, one may view the expression  $\letboxm{(\Gammahat^n.u)}{e_1}{e_2}$, as a match on $e_1$ where the pattern consists only of a pattern variable $u$ and $e_2$ is the branch of the case-expressions. We view code patterns as higher-order abstract syntax trees; this abstraction allows us to choose any encoding internally in the implementation. 

Adding pattern matching on code in a type-safe manner has been a 
challenge for two reasons: 

\paragraph{1) Representation of Code and Refinement of Types}  In
general, we want to pattern match on $\cbox{\Gamma \atl{n} \alpha}$
(where we omit the identity substitution that is associated with the
type variable $\alpha$ for better readability). As a consequence, a
pattern may impose some constraints on
$\alpha$. For example, a code pattern $\boxm{\Gammahat}{\fn{x}{p}}$
which has type $\cbox{\Gamma \atl{n} \beta_1 \arrow \beta_2}$
generates the constraint that $\alpha := \beta_1 \arrow \beta_2$. A
similar situation exists in dependently typed systems. 
Hence, it is not surprising that constraints arise.

\paragraph{2) Characterizing pattern variables at different stages. }
 A key question is what type to assign to pattern variables in code and
 type patterns. When we capture a code pattern $\boxm{\Gammahat}{\fn{x}{p}}$, it seems
natural that we assign $p$ the type $\beta_2$ in the context $\Gamma$
extended with the declaration for $x$ of type $\beta_1$. This follows,
for example, the approach taken in the proof environment \Beluga
\cite{Pientka:POPL08,Pientka:IJCAR10,Pientka:CADE15} where the pattern variable $p$
depends on the bound variables from $\Gammahat$ extended with
$x$. However, in our setting, we want to also allow code patterns of
the form 
$\boxm{\Gammahat}{\letboxm{(\Phihat^k.u)}{p}{q}}$. As a consequence,
we need to capture the type of $q$, which may depend not only on the
variables from $\Gammahat$, but also on the variable $u$. Or we might
want to allow code patterns that themselves contain code! For example,
$\boxm{\Gammahat}{\boxm{\Phihat}{p}}$. This is where our multi-level
context approach pays off. Since we associate every
variable with a level, we simply view $q$ in the extended context with
the variable $u$ where $u$ has been inserted at the appropriate position.
Similarly, in the code pattern
$\boxm{\Gammahat^k}{\boxm{\Phihat^n}{p}}$, the levels associated with
$\Gammahat$ and $\Phihat$ give us the structure to determine in which
context $p$ is making sense.  

\smallskip

These two considerations lead us to generalize our term language and
our contexts. In particular, we track and reason about two kinds of
constraints: the first one $ \alpha := (\Psi \atl{n} T) $ refines a type declaration $\alpha:(\Psi \atl{*})$ and the second one $\#$ denotes an inconsistent set of constraints. The latter allows for elegant handling of impossible cases that may arise. 
We note that these constraints arise only during typing, and programmers will write only pure contexts, i.e. contexts without any constraints.

\[
\begin{array}{llcl}
\mbox{Terms} & e , p, q  & \bnfas & \ldots 
                  \bnfalt \caseof{e}{\cbox{\Phi \atl{k} T}} ~\overrightarrow{(\Psi_i.(\Phihat_i.p_i) : (\Phi_i \atl{k} T_i) \tto e_i)}\\
\mbox{Context} & \Gamma, \Psi, \Phi & \bnfas & \ldots \bnfalt \Gamma,
                                               \alpha := (\Psihat^n.T) : (\Psi \atl{n} *) \bnfalt \Gamma, \#
\end{array}
\]

For simplicity and clarity, we define pattern matching and branches in
a case-expression in such a way that all pattern variables are
explicitly listed as part of the branch in the context $\Psi_i$. The
context $\Phi_i$ contains all the locally bound variables that may be
used in the pattern $p_i$. We concentrate here on patterns that are
terms, but not case-expressions themselves to explain the main
ideas. 
 In practice, we would infer the type associated with those pattern
 variables. We note that for all variables $x{:}( \_ \atl{n} \_)$  in
 $\Psi_i$, we have that $n  > \m{level}(\Phi_i)$,
 since pattern variables abstract over holes/variables in a code
 pattern.
Further, these pattern variables themselves may depend on the context
 $\Phi_i$ where all variables in $\Phi_i$ denote bound variables
 within the code pattern and hence are variables defined at levels
 higher than $n+1$.
 We also add a type annotation $(\Phi_i \atl{k} T_i) $ to each pattern
 which states that the pattern $p_i$ has type $T_i$ in the context
 $\ctxapp{\Psi_i}{\Phi_i}$.

 Last but not least, we add a type annotation to the case-expression
 itself. It describes the type of the scrutinee $e$. This is used
 during run-time where we first find the compatible branch by matching
 the type of the scrutinee against one of type annotations in branches and
 subsequently match the scrutinee against the pattern in the selected
 branch. Before giving a more detailed explanation of the static and
 dynamic semantics for case-expressions and turning our previous kinding and
 typing rules from Sec.~\ref{sec:typrules} 
 into shallow pattern rules (see Sec.~\ref{sec:pattyp}), we consider
 the generation and handling of type constraints next.

\subsection{Typing modulo}
To handle constraints during type checking, we extend our definition
of well-formed contexts s.t. $\Gamma, \alpha := (\Psi \atl{n} T)$ is a
well-formed context if $\Gamma$ is well-formed, and the type $T$ is
well-formed in $\ctxapp{\ctxrestrict{\Gamma}{n}}{\Psi}$. We also must
ensure that there are no circularities in $\Gamma$. A context $\Gamma,
\#$ which contains a contradiction marked by $\#$ is well-formed, if
$\Gamma$ is. 

\[
  \begin{array}{c}
\multicolumn{1}{l}{\text{Additional rules for context
    well-formedness}:\fbox{$\vde \Gamma$}}
\\[0.75em]
\infer{\vde \Gamma,  \alpha := (\Psihat^n. T) : (\Psi \atl{n} *)}
{\vde \Gamma & \vde (\ctxapp{\ctxrestrict{\Gamma}{n}}{\Psi}) & (\ctxapp{\ctxrestrict{\Gamma}{n}}{\Psi}) \vde T}
\quad 
\infer{\vde \Gamma, \#}
{\vde \Gamma}
  \end{array}
\]

We also revisit the typing rules for
substitutions. To ensure that $\sigma,~(\Psihat^k.T)$ is a well-typed
substitution for a context 
$\Phi, \alpha{:=}(\Psihat^k.S) : (\Psi \atl{k} *)$, we check in
addition to the fact that $\sigma$ and $T$ are well-typed, that $T$ and $S$ are equal
in the range of the substitution. We also note that, if the domain of a
substitution contains a contradiction, then the range must also
contain a contradiction. 
%
\[
  \begin{array}{c}
\multicolumn{1}{l}{\text{Additional rules for well-typed
    substitutions:}~ \fbox{$\Gamma \vde \sigma : \Phi$}}
\\[0.75em]
\infer
{\Gamma \vde \sigma,~(\Psihat^k.T) : \Phi, \alpha{:=}(\Psihat^k.S) : (\Psi \atl{k} *)}
{\Gamma \vde \sigma : \Phi ~~\quad& 
 \ctxapp{\ctxrestrict{\Gamma}{k}}{\Psi} \vde T = [\ctxapp{(\ctxrestrict{(\sigma/\hat\Phi)}{k})}{(\id(\hat\Psi)/\hat\Psi)}]S ~~\quad&
 \ctxapp{\ctxrestrict{\Gamma}{k}}{\Psi} \vde T
}
\quad
\infer
{\Gamma \vde \sigma : \Phi, \#}
{\# \in \Gamma& 
 \Gamma \vde \sigma : \Phi
}
  \end{array}
\]
Last, we add a type conversion rule that lets us
prove that two types are equal modulo the constraints in
$\Psi$. 
\[
\begin{array}{c}
\infer{\Psi \vde e : T}
{\Psi \vde e : S & \Psi \vde S = T}
\end{array}
\]




\subsection{Type equality modulo constraints} \label{sec:typeq}

Since we accumulate equality constraints in our context $\Psi$ during
type checking, declarative equality is not purely structural -- it can
also exploit the equality constraints. Structural equality on types
modulo constraints is defined using the judgment: $\Psi \vde T = S$.

Most structural equality rules are as expected, and we omit
them for brevity. To, for example, compare two function types for equality, we
simply compare their components. For polymorphic types, we similarly
compare their components and make
sure to extend the context with the declaration for $\alpha$ when
comparing the result types. If $\Psi$ contains a contradiction, then
we simply succeed. The only
interesting case is : $\alpha[\sigma] = T$ where $\alpha{:=}(\Phihat^n.S) : (\Phi \atl{n} *) \in \Psi $.
In this case, we continue to compare $T$ with $[\sigma/\Phihat]S$.
The complete set of structural equality rules can be found in \LONGVERSION{the Appendix \ref{appendix:typeq}}{\LONGVERSIONREF}.



\subsection{Typing Rule for Case-Expressions}
We now discuss the typing rule for case-expressions which is the centrepiece of supporting typed code analysis. 
\[
\begin{array}{c}
\infer{\Psi \vde \caseof{e}{\cbox{\Phi \atl{k} T}}~ {\overrightarrow{(\Psi_i.(\Phihat_i.p_i) : (\Phi_i \atl{k} T_i) \tto e_i)}} : S}
{
\begin{array}{l}
\qquad\quad\quad\!\Psi \vde e : \cbox{\Phi \atl{k} T} \qquad \qquad\quad
\Psi \vde \cbox{\Phi \atl{k} T} \qquad \qquad\\[0.5em]
\mbox{For all $i$}~
\left\{
 \begin{array}{l}
  \vde (\ctxapp{\Psi_i}{\Phi_i}) \qquad \qquad
  \Psi_i ; (\Phi_i)^k \vde T_i \qquad \quad
  \Psi_i ; (\Phi_i)^k \vde p_i : T_i \\
\vde (\ctxinsert{\Psi_i}{\ctxrestrict{\Psi}{k}}) \qquad
(\ctxinsert{\Psi_i}{\ctxrestrict{\Psi}{k}}) \vde (\Phi \atl{k} T) = (\Phi_i \atl{k} T_i) \searrow \Gamma_i
\qquad \ctxapp{\Gamma_i}{\ctxignore{\Psi}{k}} \vde e_i : S
\end{array}\right.
\end{array}
}
  \end{array}
\]

We first check that the guard $e$ has type $\cbox{\Phi \atl{k} T} $
and that the given type annotation is a well-kinded type. For each
branch, we then check the following conditions: 

\begin{enumerate}

\item The type annotation in the pattern is
  well-kinded. This entails verifying that the context
  $\ctxapp{\Psi_i}{\Phi_i}$ is well-formed and that the type
$T_i$ is well-kinded. Recall that $\ctxapp{\Psi_i}{\Phi_i}$ is only
defined, when all declarations in $\Psi_i$ are at higher levels than $k$, i.e. the level of $\Phi_i$.
\item The pattern $p_i$ has type $T_i$ in the context $\ctxapp{\Psi_i}{\Phi_i}$. In fact, we ensure something stronger, namely that $\Psi_i$ contains the pattern variables and the code pattern $\boxm{\Phihat_i}{p_i}$ has type $\cbox{\Phi_i \atl{k} T_i}$. This is accomplished by the judgment $\Psi_i ; \Phi_i \vde p_i : T_i$ (see Sec. \ref{sec:pattyp}).
\item We match the type of the pattern, $(\Phi_i \atl{k} T_i)$ against
  the type of the scrutinee $(\Phi \atl{k} T)$ using the judgment
\[
(\ctxinsert{\Psi_i}{\ctxrestrict{\Psi}{k}}) \vde (\Phi \atl{k} T) =
(\Phi_i \atl{k} T_i) \searrow \Gamma_i
\]

This generates a new
  context $\Gamma_i$ which constrains some type variables in
  $\ctxrestrict{\Psi}{k}$. 
 Recall that all variables in $\Psi$ which are below $k$ will be
 replaced by $\Phi_i$, and hence only the declarations in
 $\ctxrestrict{\Psi}{k}$ matter. Further,
 $\ctxinsert{\Psi_i}{\ctxrestrict{\Psi}{k}}$ only contains variable
 declarations at levels above $k$ and therefore also $\Gamma_i$
 contains only variables above $k$.

 Last, we ensure that the resulting context $\Gamma_i$ is
 well-formed, by starting with the joint context
 $(\ctxinsert{\Psi_i}{\ctxrestrict{\Psi}{k}})$. The order will in fact
 ensure that type variables in $\ctxrestrict{\Psi}{k}$ (the type of
 the scrutinee) can be
 constrained by variables in $\Psi_i$ (the type of the pattern).
 We describe matching and the generation of constraints in
 Sec.~\ref{sec:typconstraints}. 

\item Finally, we check that the body $e_i$ of the branch has type $S$ in
  the constrained context $\Gamma_i$ extended with the
  $\ctxignore{\Psi}{k}$ where we remove all assumptions that are at
  level $k$ or higher. Note that all removed assumptions are in fact
  present in $\Gamma_i$  and have possibly been refined in the previous step.
\end{enumerate}

We now describe pattern typing and constraint generation.

\subsection{Pattern Typing}\label{sec:pattyp}

Pattern kinding and typing rules are derived from
kinding and typing rules for types and terms in Sec.~\ref{sec:typrules}. They are a special case of those
rules. 
In all the pattern typing judgments, we separate
the pattern variables \(\Psi\) from the bound variables \(\Gamma\), writing $\Psi ;
\Gamma^n$ instead of $\Psi,\Gamma$ where $\m{level}(\Gamma) = n$ and
for every declaration $x{:}(\_ \atl{j} \_)$ in $\Psi$, we have $j >
n$. Although this separation is not strictly necessary for ensuring
that a pattern is well-typed, it is necessary when we define unification
on code and types, since only pattern variables in $\Psi$ can be
instantiated, while bound variables in $\Gamma$ remain fixed. Further,
to ensure that unification falls into the decidable higher-order pattern
fragment (see \citet{Miller91iclp}), pattern variables must be
associated with a variable substitution. For simplicity, we choose here the identity
substitution, which we in fact omit for better readability.

\begin{figure}
  \begin{displaymath}
    \begin{array}{c}
      \infer{\Psi ; \Gamma^n \vde \alpha[\sigma]}
            {\alpha:(\Phi \atl{k} *) \in \Gamma &
             \Psi ; \Gamma \vde \sigma : \Phi}
      \quad
      \infer{\Psi ; \Gamma^n \vde \alpha}
            {\alpha:(\Gamma \atl{n} *) \in \Psi}
      \\[0.75em]
      \infer{\Psi;\Gamma^n \vde \cbox{\Phi \atl{k} \beta}}
            { k > 0 &
              \vde (\ctxapp{\ctxrestrict{\Gamma}{k}}{\Phi}) &
              \beta:((\ctxapp{\ctxrestrict{\Gamma}{k}}{\Phi}) \atl{\m{max}(n,k)} *) \in \Psi
            }
      \\[0.75em]
      \infer{\Psi ; \Gamma^n \vde \beta \arrow \alpha }
            {\beta:(\Gamma \atl{n} *) \in \Psi &
             \alpha:(\Gamma \atl{n} *) \in \Psi}
      \quad
      \infer{\Psi ; \Gamma^n \vde (\alpha {:}(\Phi \atl{k} *)) \arrow \beta}
            {\vde \ctxapp{\Psi}{\ctxapp{\ctxrestrict{\Gamma}{k}}{\Phi}} &
             \beta:((\ctxinsert{\Gamma}{\alpha {:}(\Phi \atl{k} *)}) \atl{\m{max}(n,k)} *) \in \Psi}
    \end{array}
  \end{displaymath}
  \caption{Type Pattern rules: \fbox{$\Psi ; \Gamma^n \vde T$}}
  \label{fig:patkinding}
\end{figure}

We first consider type patterns (see  Fig. \ref{fig:patkinding}). For
type variables, we distinguish two cases. When the type variable $\alpha$ 
is in $\Psi$, it must have type $(\Gamma \atl{n} *)$ and describes a
pattern type variable. A bound variable occurrence, $\alpha[\sigma]$, is well-kinded, if $\alpha$ is declared in $\Gamma$ to have kind $(\Phi \atl{k} *)$ and $\sigma$ is a substitution mapping variables from $\Phi$ to $\Psi;\Gamma$. Pattern kinding for function types $\beta \arrow \alpha$ is straightforward: each pattern variable must be declared in $\Psi$ with kind $(\Gamma \atl{n} *)$. For polymorphic type patterns, $(\alpha {:}(\Phi \atl{k} *)) \arrow \beta$, we note that $\beta$ is the pattern variable that makes sense in the extended context $(\ctxinsert{\Gamma}{\alpha {:}(\Phi \atl{k} *)})$. Note that in general, $k$ could be greater than $\m{level}(\Gamma) = n$. Hence, the type variable $\alpha$ can increase the overall level of the extended context, if $k > n$, 
 and the extended context has level $\m{max}(n,k)$. Therefore, the pattern type variable $\beta$ must be declared in $\Psi$ with kind $((\ctxinsert{\Gamma}{\alpha {:}(\Phi \atl{k} *)}) \atl{\m{max}(n,k)} *)$. 
Last, we consider the contextual type pattern, $\cbox{\Phi \atl{k}
  \beta}$, where $\beta$ is the pattern type variable. It is
meaningful in the context
$(\ctxapp{\ctxrestrict{\Gamma}{k}}{\Phi})$. The level of this
resulting context is $\m{max}(n,k)$ depending on whether $n$ is
greater or less than $k$. Hence, it must be declared in $\Psi$ with kind $((\ctxapp{\ctxrestrict{\Gamma}{k}}{\Phi}) \atl{\m{max}(n,k)} *)$.

\begin{figure}
\[
\begin{array}{c}
\infer{\Psi ; \Gamma^n \vde x[\sigma] : [\sigma/\Phihat]T}
{x:(\Phi \atl{k} T) \in \Gamma&
 \Psi;\Gamma^n \vde \sigma : \Phi}
\quad
\infer{\Psi ; \Gamma^n \vde x : T}
{x:(\Gamma \atl{n} T) \in \Psi}
%
\quad
\infer{\Psi ; \Gamma^n \vde \fn{x}{p} : S \arrow T}
      {p{:}( \ctxappend{\Gamma}{x{:}(\cdot \atl{0} S)} \atl{n}T) \in \Psi
      }
\\[0.75em]
\infer{\Psi ; \Gamma^n \vde p\;q : T}
      {p{:}(\Gamma \atl{n} S \arrow T) \in \Psi &
       q{:}(\Gamma \atl{n} S) \in \Psi}
\qquad
\infer{\Psi ; \Gamma^n \vde \tfn{\alpha^k}{p} : (\alpha{:}(\Phi \atl{k} *)) \arrow T}
      {p{:}(\ctxinsert{\Gamma}{\alpha{:}(\Phi \atl{k}*)} \atl{\m{max}(n,k)}  T) \in \Psi}
\\[0.75em]
\infer{\Psi ; \Gamma^n \vde p\;(\Phihat^k.\beta) : [\Phihat^k.\beta[\id(\Phihat)]/\alpha^k]S}
      {p:(\Gamma \atl{n} (\alpha{:}(\Phi \atl{k} *)) \arrow S) \in \Psi
     & \beta{:}(\ctxapp{\ctxrestrict{\Gamma}{k}}{\Phi} \atl{\m{max}(n,k)} *) \in \Psi}
\\[0.75em]
\infer{\Psi ; \Gamma^n \vde \boxm{\Phihat^k}{p} : \cbox{\Phi \atl{k} T}}{
k > 0 &
p{:}(\ctxapp{\ctxrestrict{\Gamma}{k}}\Phi \atl{\m{max}(n,k)} T) \in \Psi}
\quad
\infer{\Psi ; \Gamma^n \vde \letboxm{(\Phihat^k.u)}{p}{q} : T}
      {p{:}(\Gamma \atl{n} \cbox{\Phi \atl{k} S}) \in \Psi &
      q{:}(\ctxinsert{\Gamma}{u{:}(\Phi \atl{k} S)} \atl{\m{max}(n,k)}  T) \in \Psi}
\end{array}
\]
  \caption{Code Pattern Typing \fbox{$ \Psi ; \Gamma^n \vde p : T$}}
\label{fig:pattyping}
\end{figure}

We describe code pattern typing in Fig.~\ref{fig:pattyping}. A code
pattern $p$ is the sub-term in $\boxm{\Gammahat^n}{p}$ which has type $\cbox{\Gamma \atl{n} T}$. Our
code pattern typing rules follow the typing rules for terms given
earlier, and we assume that $\cbox{\Gamma \atl{n} T}$ is a well-formed
pattern type.  We again separate the pattern variables (declared in
$\Psi$) from the bound variables (declared in $\Gamma$) (see
Fig.~\ref{fig:pattyping}). A pattern variable $x$ of type $T$ is
implicitly associated with the identity substitution and hence must be
declared in $\Psi$ with the type $(\Gamma \atl{n} T)$. A bound
variable occurrence, written as a closure $x[\sigma]$, is well-typed,
if $x$ is declared to have type $(\Phi \atl{k} T)$ in $\Gamma$ and
$\sigma$ is a substitution pattern that maps variables from $\Phi$ to $\Psi;\Gamma$. In code patterns for functions, $\fn{x}{p}$, of type $S \arrow T$, the pattern variable $p$ is meaningful in the extended context $\ctxappend{\Gamma}{x{:}(\cdot \atl{0} S)}$. It hence must be declared in $\Psi$ with type 
$(\ctxappend{\Gamma}{x{:}(\cdot \atl{0} S)} \atl{n}) T$. Code patterns for applications, $p~q$ are straightforward given that we adapt a declarative formulation. Polymorphic code patterns, $\tfn{\alpha^k}{p} $, and type application patterns, $p\;(\Phihat^k.\beta)$, follow the same principles as in polymorphic type patterns discussed earlier; similarly, boxed code patterns, $\boxm{\Phihat^k}{p}$, follow the same principles as for contextual type patterns. In each of these cases, the pattern variable is declared in the extended context at level $\m{max}(n,k)$. 
For $\letboxm{(\Phihat^k.u)}{p}{q}$, we follow the typing rules given earlier: $p$ must be declared with type $(\Gamma \atl{n} \cbox{\Phi \atl{k} S})$ and $q$ with type $(\ctxinsert{\Gamma}{u{:}(\Phi \atl{k} S)} \atl{\m{max}(n,k)}  T)$ in $\Psi$.

Last, we describe typing rules for substitution patterns
(Fig.~\ref{fig:substpattyping}) which follow the corresponding typing
rules for simultaneous substitutions.

\begin{figure}[htb]
\[
\begin{array}{c}
\infer{\Psi ; \Gamma^n \vde \cdot : \cdot}{} \quad
\infer{\Psi ; \Gamma^n \vde \sigma, (\Phihat^k.p) : (\Phi', x{:}(\Phi \atl{k} T))}
{\Psi ; \Gamma^n \vde \sigma : \Phi' &
 p{:}(\ctxapp{\ctxrestrict{\Gamma}{k}}{[\sigma/\Phihat']\Phi} \atl{\m{max}(n,k)} [\sigma/\Phihat', \id(\Phihat)/\Phihat]T) \in \Psi}
\\[0.75em]
\infer{\Psi ; \Gamma^n \vde \sigma, (\Phihat^k.\alpha) : (\Phi', \alpha{:}(\Phi \atl{k} *))}
{\Psi ; \Gamma^n \vde \sigma : \Phi' &
 \alpha{:}(\ctxapp{\ctxrestrict{\Gamma}{k}}{\Phi} \atl{\m{max}(n,k)} *) \in \Psi}
\end{array}
\]
  \caption{Substitution Pattern Typing \fbox{$ \Psi ; \Gamma^n \vde \sigma : \Phi$}}
\label{fig:substpattyping}
\end{figure}

Whenever a pattern is well-typed using the pattern kinding and typing rules, the pattern is also
well-typed when viewed as an expression. 

\pagebreak[3]
\begin{lemma}[Pattern Reflection]\label{lem:patrefl}\mbox{} \quad Let \(\vde
                                 \ctxappend{\Psi}{\Gamma^n}\).
  \begin{enumerate}
  \item\label{itm:patrefl-type} If \(\Psi ; \Gamma^n \vde T\) then \(\Psi \vde \cbox{\Gamma \atl{n} T}\)
  \item\label{itm:patrefl-pattern} If \(\Psi ; \Gamma^n \vde p : T\) then \(\Psi \vde \boxm{\hat\Gamma^n}{p} : \cbox{\Gamma \atl{n} T}\)
  \item\label{itm:patrefl-subst} If \(\Psi ; \Gamma^n \vde \sigma : \Phi\) then \(\Psi, \Gamma^n \vde \sigma : \Phi\)
  \end{enumerate}
\end{lemma}

\begin{proof}
  For (\ref{itm:patrefl-subst}), by induction on the first derivation. For others, by case analysis.
\end{proof}

\subsection{Unification on Types: constraint generation}\label{sec:typconstraints}

\begin{figure}
  \begin{displaymath}
    \begin{array}{@{}c@{}}
      \mbox{}\hfill
      \infer%
          {\Gamma ; \Phi \vde \alpha = \alpha \searrow \Gamma}%
          {\alpha{:}(\Phi \atl{k} *) \in \Gamma}%
      \hfill
      \infer{\Gamma ; \Phi \vde \alpha = T \searrow \Gamma, \#}
            {\Gamma = \Gamma_1, \alpha{:}(\Phi \atl{k} *), \Gamma_0 &
             \Gamma ; \Phi \vde \alpha \in T
            }
      \hfill\mbox{}
      \\[0.75em]
      \infer{\Gamma ; \Phi \vde \alpha = T \searrow \Gamma'}
            {\Gamma = \Gamma_1, \alpha{:}(\Phi \atl{k} *), \Gamma_0 &
             \ctxapp{\Gamma_1}{\Phi} \vde T &  \Gamma ; \Phi \vde \alpha \not\in T &
             \Gamma' = \Gamma_1, \alpha{:=}(\Phihat.T) : (\Phi \atl{k} *), \Gamma_0
            }
      \\[0.75em]
      \infer{\Gamma ; \Phi \vde \alpha = T \searrow \Gamma'}
            { \Gamma = \Gamma_1, \alpha{:=}(\Phihat.T') : (\Phi \atl{k} *),\Gamma_0 &
              \Gamma ; \Phi \vde \alpha \not\in T &
              \Gamma ; \Phi \vde T' = T \searrow \Gamma'
            }
      \\[0.75em]
      \infer%
          {\Gamma ; \Phi \vde \alpha[\sigma] = \alpha[\sigma'] \searrow \Gamma'}%
          {\alpha{:}(\Psi \atl{k} *) \in \Phi&%
           \Gamma ; \Phi \vde \sigma = \sigma' : \Psi \searrow \Gamma'
          }%
      \quad
      \infer{\Gamma ; \Phi \vde \cbox{\Psi \atl{n} T} = \cbox{\Psi' \atl{n} S} \searrow \Gamma''}
            { 
              \Gamma ; \ctxrestrict{\Phi}{n} \vde \Psi = \Psi' \searrow \Gamma' &
                  \Gamma'; (\ctxapp{\ctxrestrict{\Phi}{n}}{\Psi}) \vde T = S \searrow \Gamma''
            }
      \\[0.75em]
      \infer{\Gamma ; \Phi \vde T_1 \arrow T_2 = S_1 \arrow S_2 \searrow \Gamma_0}
            {\Gamma ; \Phi \vde T_1 = S_1 \searrow \Gamma_1 &
             \Gamma_1 ; \Phi \vde T_2 = S_2 \searrow \Gamma_0}
      \quad 
      \infer{\Gamma ; \Phi \vde (\alpha{:}(\Psi\atl{n} *)) \arrow T =
             (\alpha{:}(\Psi'\atl{n} *)) \arrow S \searrow \Gamma'}
            {
              \Gamma ; \ctxrestrict{\Phi}{n} \vde \Psi = \Psi' &
                  \Gamma ; \ctxinsert{\Phi}{(\alpha{:}(\Psi\atl{n} *))} \vde T = S \searrow \Gamma'
            }
      \\[0.75em]
    \end{array}
  \end{displaymath}
  \caption{Type Unification:\fbox{$\Gamma ; \Phi \vde T = S \searrow \Gamma'$ }}
  \label{fig:typunify}
\end{figure}

We define the refinement of types using unification via the judgment
 \[
\Psi \vde (\Phi \atl{k} T) = (\Phi_i \atl{k} T_i) \searrow \Gamma
\]

Without loss of generality, we assume that $\Psi$ is a context that
contains variables at levels higher than $k$. In particular, it contains the pattern
variables. 

\[
\infer{\Gamma \vde (\Phi \atl{k} T) = (\Phi_i \atl{k} T_i) \searrow \Gamma_2}
{\Gamma ; \cdot \vde \Phi = \Phi_i \searrow \Gamma_1 &
 \Gamma_1 ; \Phi \vde T = T_i \searrow \Gamma_2}
\]

To unify a contextual type $(\Phi \atl{k} T)$ against $(\Phi_i \atl{k} T_i) $, we first unify
$\Phi$ against $\Phi_i$ and subsequently unify $T$ against $T_i$ given
the (bound) variable context $\Phi$. To separate between the bound variables \(\Phi\) that are fixed and bound and
pattern variables \(\Gamma\) that can be refined and instantiated, we again use the
symbol $;$ and write $\Gamma ; \Phi$.
Context and type unification are then defined in Fig.~\ref{fig:typunify} and Fig.~\ref{fig:contextunify}
using the judgments $\Gamma ; \Phi^k \vde T = S \searrow \Gamma'$ and
$\Gamma ; \Gamma_0 \vde \Psi = \Phi \searrow \Gamma'$.  

We only unify well-kinded types $T$ and $S$; in particular, if $T$ and
$S$ are well-typed in $\ctxapp{\Gamma}{\Phi^k}$, then they remain
well-typed in $\ctxapp{\Gamma'}{\Phi^k}$ where $\Gamma'$ is a
refinement of $\Gamma$, i.e.  some type variable declarations that occur in $T$ and
in $S$ are constrained
s.t. $\ctxapp{\Gamma'}{\Phi} \vde T = S$. 


\begin{figure}[t]
  \begin{displaymath}
    \begin{array}{c}
      \infer{\Gamma ; \Gamma_0 \vde \cdot = \cdot \searrow \Gamma}{}
      \quad
      \infer{\Gamma ; \Gamma_0 \vde (\Psi, \alpha{:}(\Psi' \atl{n} *)) = (\Phi, \beta{:}(\Phi' \atl{n} *)) \searrow \Gamma''}
            {\Gamma ; \Gamma_0 \vde \Psi = \Phi \searrow \Gamma' &
             \Gamma' ; \ctxapp{\Gamma_0}{\ctxrestrict{\Psi}n} \vde \Psi' = \Phi' \searrow \Gamma''}
      \\[0.75em]
      \infer{\Gamma ; \Gamma_0 \vde (\Psi, x{:}(\Psi' \atl{n} T')) = (\Phi, y{:}(\Phi'
             \atl{n} S')) \searrow \Gamma'''}
            {\Gamma ; \Gamma_0 \vde \Psi = \Phi \searrow \Gamma' &
             \Gamma' ; \ctxapp{\Gamma_0}{\ctxrestrict{\Psi}n} \vde \Psi' = \Phi' \searrow \Gamma'' &
             \Gamma'' ; \ctxapp{\Gamma_0}{\ctxapp{\ctxrestrict{\Psi}n}{\Psi'}} \vde T' = S'
             \searrow \Gamma'''}
      \\[0.75em]
      \mbox{[All other cases yield $\Gamma,\#$]}
    \end{array}
  \end{displaymath}
  \caption{Context Unification: \fbox{$\Gamma ; \Gamma_0 \vde \Psi = \Phi \searrow \Gamma'$}}
  \label{fig:contextunify}
\end{figure}

\newcommand{\ext}{\succeq}
We write $\Gamma' \ext \Gamma$ to describes context refinement 
where a declaration $\alpha{:}(\Phi \atl{n} *)$ in $\Gamma$ has been
updated to $\alpha{:=}(\Phihat^n.T) : (\Phi \atl{n} *)$ and all
variable declarations present in $\Gamma$ are also present in $\Gamma'$.

Unification is defined recursively based on the type $T$ and $S$. In
particular, if $T= T_1 \arrow T_2$ and $S = S_1 \arrow S_2$, we first
match $T_1$ against $S_1$ which returns a refined context $\Gamma_1$
and we subsequently match $T_2$ against $S_2$ yielding a further
refinement $\Gamma_2$. We similarly proceed to match 
$\cbox{\Psi \atl{n} T}$ against $\cbox{\Psi' \atl{n} S} $ given the
pattern variable context $\Gamma$ and the bound variable context
$\Phi$. Since $\cbox{\Psi \atl{n} T}$ is well-kinded type pattern in
$\Gamma ; \Phi$, we know that $T$ is a well-kinded type pattern in 
$\Gamma ;\ctxapp{\ctxrestrict{\Phi}{n}}{\Psi}$. We therefore first
match $\Psi$ against $\Psi'$ in $\Gamma ; \ctxrestrict{\Phi}{n}$. This
returns an updated context $\Gamma'$. Next, we match $T$ against $S$
in the context $\Gamma'; \ctxapp{\ctxrestrict{\Phi}{n}}{\Psi}$.

To unify two polymorphic types, $(\alpha{:}(\Psi\atl{n} *)) \arrow T $ against 
$(\alpha{:}(\Psi'\atl{n} *)) \arrow S$, we first match $\Psi$ against
$\Psi'$ and subsequently match $T$ against $S$. The first context
match is not strictly necessary in our setting, as $\Psi$ and $\Psi'$
are type variable contexts and hence do not themselves contain any
pattern variables themselves. Nevertheless, we keep the general form,
as it highlights how one would extend it to a fully dependently typed
system where types could also depend on term variables.  

For pattern variables, where we match $\alpha$ against a
type $T$, there are three cases we need to consider. Note that we
again omit writing the identity substitution that is associated with
the pattern type variable $\alpha$. 

 1) if $T$
contains $\alpha$ (occur's check), then we return a context where we
add a contradiction. The constraints are unsatisfiable, and hence, we
can conclude anything.

 2) if $\alpha$ is not yet constrained, we add
$(\Phihat^k.T)$ as a constraint, making sure that $T$ is well-typed in
$\ctxapp{\Gamma_1}{\Phi}$. This will ensure that the resulting
constrained context is well-formed. 

3) if $\alpha$ was associated with a constraint $\alpha := (\Phihat^k.S)$, then we continue to match $S$ against $T$.

We omit here the symmetric variable cases where $T = \alpha$ for
compactness and lack of space, but they follow the same principle

Further, we mostly concentrate here on the cases where unification
succeeds; the omitted cases, where $T$ is not a pattern type variable
and $T$ and $S$ do not share the top-level type constructor, lead to
a contradiction. This is reflected in returning the context $\Gamma, \#$.

In practice, we use the given unification algorithm for
bi-directional matching where both sides $T$ and $T_i$ have distinct
unification variables. During type-checking, we match $T_i$ (type of
the pattern) against $T$ (the type of the scrutinee),
and we instantiate meta-variables in $T$. During runtime, the type
$T$ is concrete, and we match $T$ against $T_i$ to select the
appropriate branch in the case-expression.

As a consequence, our unification algorithm is biased. 
For example, if $\beta$ occurs before $\alpha$ in $\Gamma$, and
we encounter a unification problem $\alpha = \beta$, we will
instantiate $\alpha$. If $\alpha \arrow \alpha$ is the type of the
pattern and $\beta$ is the type of the scrutinee, the algorithm will
fail. This is in line with our intuition that the pattern refines the
type of the scrutinee, but not vice versa.

Unification on contexts proceeds recursively using the judgment $\Gamma ; \Gamma_0 \vde \Psi = \Phi \searrow \Gamma'$. Here $\Gamma$ denote the pattern variables. The context $\ctxapp{\Gamma_0}{\Psi}$ and $\ctxapp{\Gamma_0}{\Phi}$ are well-formed, i.e., all declarations in $\Gamma$ are at higher levels than declarations in $\Gamma_0$, and all declarations in $\Gamma_0$ are at higher levels than declarations in $\Psi$ and $\Phi$. We maintain these criteria during unification, which helps us to cleanly manage variable dependencies. 

\begin{lemma}
  If $~\Gamma' \ext \Gamma$ and $\Gamma \vde T = T'$ then
$\Gamma' \vde  T = T'$.
\end{lemma}

We next state and prove some properties of unification.

\begin{lemma}[Well-defined Unification]\label{lem:welldefunify}\quad
Assume     $\vde \ctxapp{\Gamma}{\ctxapp{\Gamma_0}{\Psi}}$.
  \begin{enumerate}
  \item\label{itm:welldefunify-context}
    If
    $\vde \ctxapp{\Gamma}{\ctxapp{\Gamma_0}{\Phi}}$
    then
    there is a $\Gamma'$ s.t. $\Gamma ; \Gamma_0 \vde \Psi = \Phi \searrow \Gamma'$.
  \item\label{itm:welldefunify-type}
    If
    $\ctxapp{\Gamma}{\Psi} \vde T$ 
    $\ctxapp{\Gamma}{\Psi} \vde S$    then
    there is a $\Gamma'$ s.t. $\Gamma ; \Psi \vde T = S \searrow \Gamma'$
  \end{enumerate}
\end{lemma}

\begin{lemma}[Soundness of Unification]\label{lem:unifysound}\quad 
  \begin{enumerate}
  \item\label{itm:unifysound-context}
    If
    $\Gamma ; \Gamma_0 \vde \Phi = \Psi \searrow \Gamma'$
    then
    $\ctxapp{\Gamma'}{\Gamma_0} \vde \Phi = \Psi$ and
    $\Gamma' \ext \Gamma$.
  \item\label{itm:unifysound-type} 
    If
    $\vde \ctxapp{\Gamma}{\Phi}$ and
    $\Gamma ; \Phi \vde T = S \searrow \Gamma'$
    then
    $\vde \Gamma'$ and
    $\ctxapp{\Gamma'}{\Phi} \vde T = S$ and
    $\Gamma' \ext \Gamma$.  
  \end{enumerate}
\end{lemma}



%
%


\begin{lemma}[Unification is stable under substitution]\label{lem:stableunify}\mbox{}

  If
  $\ctxapp{\Gamma}{\Psi} \vde T$ and
  $\ctxapp{\Gamma}{\Psi} \vde S$ and
  $\Gamma ; \Phi \vde T = S \searrow \Gamma'$ and
  $\Gamma_1 \vde \sigma : \Gamma$ then
  there is $\Gamma'_1$ such that
  $\Gamma_1 ; [\sigma/\Gammahat]\Phi \vde [\ctxapp{(\sigma/\Gammahat)}{(\m{id}(\hat\Phi)/\hat\Phi)}]T = [\ctxapp{(\sigma/\Gammahat)}{(\m{id}(\hat\Phi)/\hat\Phi)}]S \searrow \Gamma'_1$ and
  $\Gamma'_1 \vde \sigma : \Gamma'$
\end{lemma}

\begin{proof}
  By induction on the structure of $\Gamma ; \Phi \vde T = S \searrow \Gamma'$. Note that if \(\Gamma'_1\) contains \(\#\), \(\Gamma'_1 \vde \sigma : \Gamma'\) becomes trivial with \(\Gamma_1 \vde \sigma : \Gamma\).
\end{proof}

\begin{lemma}[Unification is Compatible with Structural Equality]\label{lem:unifycomp}\mbox{}

  For pure context \(\Psi\),
  if
  $\Psi \vde T$ and
  $\Psi \vde S$ and
  $\mathcal{D}:\Phi \vde T = S$ and
  $\mathcal{E}:\cdot ; \Phi \vde T = S \searrow \Gamma'$
  then $\Gamma' = \cdot$
\end{lemma}

\begin{proof}
  By induction on the structure of \(\mathcal{D}\) and case analysis on the structure of \(\mathcal{E}\).
\end{proof}



\subsection{Operational Semantics for Case-expressions} 
We define here the operational semantics of case-expressions and show
that types are preserved. Together with the reduction rules for our
core multi-level lambda-calculus that
emerge from Sec.~\ref{sec:locsound} this provides an operational
semantics for \Moebius. 

\begin{figure}[htb]
\[
  \begin{array}{lcl}
(\fn{x}{e})~v & \Longrightarrow & [(~ .v)/x^0]e\\
(\tfn{\alpha^n}{e})~(\Phihat^n.T) & \Longrightarrow & [(\Phihat^n.T)/\alpha^n]e\\
\letboxm{\Phihat^n.u}{\boxm{\Phihat^n}{e}}{e'} & \Longrightarrow & [(\Phihat^n.e)/u^n]e'
\\[0.75em]
\multicolumn{3}{c}{
\infer{\caseof{(\boxm{\Phihat^k}{e}) }{\cbox{\Phi \atl{k} T}}
              {~\overrightarrow{(\Psi_i.(\Phihat_i.q_i) : (\Phi_i \atl{k} T_i) \tto e'_i)} }
       \Longrightarrow
        [\sigma_i/\Psihat_i]e'_i}
      {
\cdot \vde \sigma_i : \Psi_i \qquad
      \cdot \vde [\sigma_i/\Psihat_i]\cbox{\Phi_i \atl{k} T_i}  = \cbox{\Phi \atl{k} T}\qquad
      \cdot \vde [\sigma_i/\Psihat_i]\boxm{\Phihat_i}{q_i} = \boxm{\Phihat^k}{e}
}
}
  \end{array}
\]
  \caption{Operational Semantics: \fbox{$e \Longrightarrow e'$} }
  \label{fig:opsem}
\end{figure}

Last, we prove type preservation. In the proof below, we only
concentrate on case-expressions; the remaining cases follow the
arguments for local soundness given in Sec.~\ref{sec:locsound}.

\begin{theorem}[Type preservation]
If $\vde e : T$ and $e \Longrightarrow e'$ then $\vde e' : T$. 
\end{theorem}
\begin{proof}
By case analysis on the second derivation. We show below only the case where
 $e$ is a $\m{case}$ expression.
\paragraph{Case} $\caseof{(\boxm{\Phihat^k}{e}) }{\cbox{\Phi \atl{k} T}}
              {~\overrightarrow{(\Psi_i.(\Phihat_i.q_i) : (\Phi_i \atl{k} T_i) \tto e'_i)} }
       \Longrightarrow
        [\sigma_i/\Psihat_i]e'_i$
\\[0.75em]
$\cdot \vde \caseof{(\boxm{\Phihat^k}{e}) }{\cbox{\Phi \atl{k} T}}
              {~\overrightarrow{(\Psi_i.(\Phihat_i.q_i) : (\Phi_i \atl{k} T_i) \tto e'_i)}} : S$ \hfill by assumption\\
$\Psi_i \vde \cbox{\Phi_i \atl{k} T_i} = \cbox{\Phi \atl{k} T} \searrow \Psi'_i$ \hfill by inversion\\
$\Psi'_i \vde e_i : S  $ \hfill by inversion\\
$\mbox{there exists}~\vde \sigma_i : \Psi_i $ \hfill by reduction rule for $\m{case}$\\
$\cdot \vde [\sigma_i/\Psihat_i]\cbox{\Phi_i \atl{k} T_i}  = \cbox{\Phi \atl{k} T} \searrow \Gamma$ \hfill \\
$\Gamma \vde \sigma_i : \Psi'_i$ \hfill by unification stability under substitution lemma~\ref{lem:stableunify}\\
$\cdot \vde [\sigma_i/\Psihat_i]\cbox{\Phi_i \atl{k} T_i}  = \cbox{\Phi \atl{k} T}$ \hfill by reduction rule for $\m{case}$\\
$\Gamma = \cdot$ \hfill by unification compatibility with structural equality~\ref{lem:unifycomp}\\
$\Psi'_i \ext \Psi_i$ and therefore $\Psihat'_i = \Psihat_i$ \hfill by soundness of unification lemma~\ref{lem:unifysound} and inversion on \(\vde \sigma_i : \Psi'_i\)\\
$\cdot \vde [\sigma_i/\Psihat'_i]e_i: S$ \hfill by simultaneous
substitution lemma~\ref{lem:simsubst} and the fact that $\cdot \vde S$\\
$\cdot \vde [\sigma_i/\Psihat_i]e_i: S$ \hfill since \(\Psihat'_i = \Psihat_i\)\\
\end{proof}

\section{Related Work}

\paragraph{Early staged computation systems}

\citet{Davies:ACM01} first observed that the necessity modality of the
S4 modal logic is ideally suited to distinguish 
between closed code which has type ${\square} \tau$ and programs, which have
type $\tau$ in multi-staged programming setting. To characterize open
code fragments, \citet{Davies:LICS96} proposed \(\lambda^{\bigcirc}\)
which corresponds to linear temporal logic. Open code is characterized
by the type \({\bigcirc} \tau\). While this logical foundation
provides an explanation for generating and splicing in open
code, we cannot distinguish open from closed code and the type system
of  \(\lambda^{\bigcirc}\) does not provide guarantees for its safe
evaluation. The reason is exactly in the openness: it is not sound to
evaluate an open expression before all of its free variables are
bound. This openness also makes it difficult to support pattern
matching on (open) code. 

The desire to combine the advantages of $\lambda^{\square}$  and
  \(\lambda^{\bigcirc}\) has inspired a long line of research
on type systems for staged computation, most notable
being MetaML \cite{Taha:TCS00} and in
particular, the work of \citet{Taha:POPL03} on environment classifiers which
allows manipulation of open code and supports type inference. Environment classifiers
give names for typing environments to specify when exactly its \(\m{eval}\)
function can correctly insert code, but unlike contextual types does not list
the variables that may occur in an environment (or context)
concretely. Instead, we reason about environments and their extensions
abstractly. 
Inherent in the nature of that work is that we do not reason about the concrete
variables that occur in a given piece of code. It seems hence less
expressive than using contextual types, and, in particular, it seems
difficult to extend to support pattern matching on code. As
\citet{Taha:PEPM00} observed, adding intensional analysis to MetaML
would make many optimizations unsound. Environment classifiers
also seem too coarse to characterize the dependencies that exist
within the context of assumptions. This however seems critical when
adding System F style polymorphism, as we do in this work, or
extending such a system to dependent types.
\\[-2em]

\paragraph{Reflecting contexts in types}

Subsequent work by \citet{Kim:POPL06} characterizes 
open code using typing environments (using extensible records). Since
their goal is to extend ML with Lisp's ability to write both hygienic
(via capture-avoiding substitution) and unhygienic (via capturing
substitution) templates, variables are treated symbolically, unlike with
contextual types. As a result, code templates are not guaranteed to be
lexically well-scoped. 

\citet{Kiselyov:APLAS16} describes a staged calculus \lstinline!<NJ>! that safely permits open
code movements across different binding environments. The key to
ensuring hygiene and type safety when manipulating open code is
reflecting free variables of a code fragment in its type, which evokes
contextual modal type theory. However, unlike CMTT, that work uses
environment classifiers with the argument that a type does not need to
know exactly in which order free variables are bound. We would
disagree. Especially when moving to a richer type system where we
track dependencies among assumptions, we do need to know the order. 
Contextual types provide a more exact and general approach to model
and reason about variable dependencies. 

\citet{Rhiger:ESOP12} shows a typed Kripke-style calculus for
staged computation \(\lambda^{[]}\), supporting evaluation under binders in future
stages, manipulation of open code, and mutable state. It uses its own
notion of contextual modal types, which models linear time (the typing
judgment tracks all future stages), rather than the branching time of
our contextual modal types. As in the work on MetaML, the context is
still represented abstractly, and we hence lack the ability to express
dependencies among assumptions.\\[-2em]


\paragraph{Inspecting code}

Closely related to using contextual types to model typed code
fragments, \citet{Chen:ICFP03,Chen:JFP05} characterize typeful code
representations pairing the list of value types that may appear
in a given code together with the type of the code itself. This is
reminiscent of what contextual types accomplish, but fundamentally
uses a de Bruijn representation to model bound variables occurring in
code. This is also a popular approach when modeling object languages
in proof assistants and supporting mechanizing meta-theory (see for
example \cite{Benton:JAR12}). In \citet{Chen:JFP05}, the authors sketch
how to add pattern matching on code, but it remains unclear how to
exploit the full potential of pattern matching. In particular, the
refinement of types, which happens when pattern matching on
polymorphic code, is omitted.


One of the few works that have investigated code inspection via
pattern matching in a typed
multi-staged programming setting is by \citet{Viera:GPCE06}. Similar
to contextual types, the type of code fragments, \(\m{cod}^{\Gamma}\),
are annotated with a typing environment of the free
variables. However, unlike \Moebius, the type of the code fragment
is not tracked, and hence we cannot fully reason about the type of
open code. Instead, type-checking of evaluated code fragments is
deferred to runtime. The complexity that arises from pattern matching
on typed code fragments is hence sidestepped.


\citet{Parreaux:POPL18} describe Squid, a Scala macro library which
supports code generation and code inspection using a rewrite primitive.
This allows programmers to define code transformations
without extra code for the traversal through the program itself.
This is very convenient in practice. Yet, from a theoretical point of view,
it remains unclear how these rewrite rules are applied (i.e. in what context
is a rule applied) and what conditions the rewrite rules must satisfy (i.e.
should these rules be non-overlapping, deterministic in how they are applied,
etc.).\\[-2em]

\paragraph{Beside multi-staged programming: Typed Self-Interpreter}
Many popular languages have a self-interpreter, that is, an
interpreter for the language written in itself. In this line of work
we do not reason about different stages of computation, however we can
respresent programs within the language and recover a program from its
representation. Most recently, \citet{Brown:POPL16} show that System
F$_\omega$ is capable of representing a self-interpreter, a program that recovers a program
from its representation and is implemented in the language itself.
To circumvent the limits of the type language, they encode types and
type-level functions extensionally as an instantiation
function. Our work, which sits between System F$_\omega$ and System F, faced
a similar problem and choses an intensional encoding instead. While
we do not literally support type-level computation as in System F$_\omega$, we
do support contextual kinds and view contextual type variables as a
closure. This simplifies the equational theory compared to System
F$_\omega$. Further, our notion of levels gives us the appropriate
structure to describe type and code pattern typing in
case-expressions enabling an intensional analysis of code. \\[-2em]

\paragraph{Modal Dependent Type Theory}
Recently, there has been also interest in developing a dependently
typed modal type theory that includes the box-modality.
In particular, 
\citet{Gratzer:ICFP19} describe such a type theory which, in
principle, allows us to describe the generation of dependently typed
and polymorphic code using \lstinline!box! and
\lstinline!unbox!. Their modal dependent type theory exploits a
Fitch-style system where locks are added to the context of
assumptions. Those locks manage access to variables in a similar
fashion as context stacks in Kripke-style modal
type systems. Intuitively existing assumptions 
are locked when we enter a new stage and go under a
box expression. But in that work, \emph{all} prior
assumptions are unlocked. This essentially removes the ability to
reason about different stages, as all prior stages are treated the
same; the system also has no ability to run code. Hence, this work lacks the ability to reason about
multi-staged code and run it. It also does not support reasoning about
open code nor pattern matching on code, which raise independent issues.

\citet{Pientka:LICS19} present a Martin L{\"o}f type theory that uses
the box modality. This work allows the generation and analysis of
code using pattern matching. However code is represented in the
logical framework LF  and only a subset of the computation language can be directly embedded
into LF. As such, the system is not a homogeneous programming
language and not rich enough to generally model code that contains
nested boxed and let box expressions. It also does not allow us to model
multi-staged metaprogramming, as the system is fundamentally
restricted to two stages. 

\citet{Kawata:APLAS19} develop a dependently typed multi-stage
calculus based on the logical framework LF that supports execution of
code and cross-stage persistence, but does not handle pattern matching
or polymorphism. Their approach based on environment
classifiers and cross-stage persistence annotations provides less
fine-grained control compared to the use of contextual types and level
annotations.

\section{Conclusion}
This paper describes a metaprogramming foundation that, for the first
time, supports the generation of typed polymorphic code and its analysis
via pattern matching. This is accomplished by generalizing contextual
types to a multi-level contextual type system. This allows us to
disentangle the notions of levels and stages. While a stage refers to
when at run-time code is generated vs executed, levels are annotations on
code and the variables that appear inside the code. These level
annotations on variables allow us to cleanly distinguish between holes whose
value is supplied at runtime (global variables) and holes (local
variables) that will be instantiated when the code is spliced into
another code fragment. Hence levels enforce a stratification of
variables based on how they are used. 

The multi-level contextual types are the key to defining pattern matching on
code and type fragments and unlocking the full potential of typed
meta-programming.  More specifically, the multiple levels are crucial to
support generating and inspecting code which itself generates 
code. 

An important benefit of multi-level contextual types is that they allow us
to work with clean and high-level type and code pattern abstractions where
we neither expose nor commit to a concrete representation of variable
bindings and contexts. We see our work as a step towards building a general
type-theoretic foundation for multi-staged metaprogramming that, on the one
hand, enforces strong type guarantees and, on the other hand, makes it easy
to generate and manipulate code. This will allow us to exploit the full
potential of metaprogramming without sacrificing the reliability of and
trust in the code we are producing and running.

In the immediate future, we aim to add context abstraction to
\Moebius{} following the approach taken in \Beluga. Abstracting over
contexts is an orthogonal issue, which, we believe, is compatible with
the foundation of \Moebius. 
In the longer term, we aim to generalize the foundations to dependent
types. This would provide the first dependently-type
meta-programming language.

\section{Acknowledgments}
This work was funded by the Natural Sciences and Engineering Research
Council of Canada (grant number 206263), Fonds de recherche
du Qu{\'e}bec - Nature et technologies (grant number 253521), and a
graduate fellowship from Fonds de recherche
du Qu{\'e}bec - Nature et technologies (grant number 304215) awarded to the first author.


\bibliography{bibi-extract}

\LONGVERSION{
\clearpage
\appendix
\section{Context Operation}\label{sec:ctxop}
In addition to restricting a context by dropping all variable
assumptions below a given level, it is also sometimes convenient to define a
context operation where we ignore all assumptions above a level. 


\[
  \begin{array}{lcll}
\multicolumn{3}{l}{\mbox{Chopping upper context:} \quad\ctxignore{\Psi}{n}~ =~  \Phi
}\\
    \ctxignore{(\cdot)}n &= & \cdot \\ \relax
    \ctxignore{(\Psi, x {:} (\Phi \atl{k}K)}n &= & \cdot
    &\text{if $k \ge n$}
    \\
    \ctxignore{(\Psi, x {:} (\Phi \atl{k}K))}n &= & \ctxignore{\Psi}n, x {:} (\Phi \atl{k} K)
    &\text{otherwise}
    
  \end{array}
\]

We include here the definition for appending two sorted contexts
$\Psi$ and $\Phi$ where for all declarations $x{:}(\Gamma \atl{n} \_)$
in \(\Psi\), we have \(\m{level}(\Phi) \le n + 1\).

We, therefore, define the operation of appending two sorted contexts
\(\ctxappend{\Psi}{\Phi}\) which we use instead of
merging for conciseness:

\[
  \begin{array}{lcll}
    \multicolumn{3}{l}{\mbox{Appending contexts:} \ctxappend{\Psi}{\Phi} = \Gamma }\\
    \ctxappend{\Psi}{\cdot}                         &= & \Psi\\
    \ctxappend{\Psi}{(\Phi, x{:}(\Gamma \atl{n} T))} &= & (\ctxappend{\Psi}{\Phi}), x{:}(\Gamma \atl{n} T)\\
  \end{array}
\]

{
\subsection{Structural properties of context operations}

With these operations, the following properties hold:

\begin{lemma}[Properties of Context Operations]\label{lem:contextprop}Assume that the context
  $\Psi$ and $\Phi$ are ordered, and $\ctxmerge{\Psi}{\Phi}=\Gamma$.
  \begin{enumerate}
  \item $\Gamma$ is ordered.
  \item For any level $n$, $\Gamma =
    \ctxappend{\ctxrestrict{\Gamma}{n}}{\ctxignore{\Gamma}{n}} = \ctxmerge{\ctxrestrict{\Gamma}{n}}{\ctxignore{\Gamma}{n}}$.
  \item If $\Psi$ only contains declarations greater or equal to $n$ and $\Phi$
    contains declarations below $n$,
    then $\ctxmerge{\Psi}{\Phi} = \ctxappend{\Psi}{\Phi}$.
  \end{enumerate}
\end{lemma}

\begin{lemma}[Chopping Distribution Over Merge]\label{lem:chopdistrib}\mbox{}\quad
  \begin{enumerate}
  \item \(\ctxrestrict{(\ctxinsert{\Gamma_1}{\Gamma_0})}{n} = \ctxinsert{\ctxrestrict{\Gamma_1}{n}}{\ctxrestrict{\Gamma_0}{n}}\)
  \item \(\ctxrestrict{(\ctxappend{\Gamma_1}{\Gamma_0})}{n} = \ctxappend{\ctxrestrict{\Gamma_1}{n}}{\ctxrestrict{\Gamma_0}{n}}\)
  \item \(\ctxignore{(\ctxinsert{\Gamma_1}{\Gamma_0})}{n} = \ctxinsert{\ctxignore{\Gamma_1}{n}}{\ctxignore{\Gamma_0}{n}}\)
  \item \(\ctxignore{(\ctxappend{\Gamma_1}{\Gamma_0})}{n} = \ctxappend{\ctxignore{\Gamma_1}{n}}{\ctxignore{\Gamma_0}{n}}\)
  \end{enumerate}
\end{lemma}

\begin{lemma}[Merge Reordering]\label{lem:mergereorder}\mbox{}\quad
  \begin{enumerate}
  \item If \(n < \m{level}(\Gamma_0)\) then \(\ctxinsert{(\ctxappend{\Gamma_1}{\Gamma_0})}{\alpha{:}(\Phi \atl{n} *)} = \ctxappend{\Gamma_1}{(\ctxinsert{\Gamma_0}{\alpha{:}(\Phi \atl{n} *)})}\)
  \item If \(n \ge \m{level}(\Gamma_0)\) then \(\ctxinsert{(\ctxappend{\Gamma_1}{\Gamma_0})}{\alpha{:}(\Phi \atl{n} *)} = \ctxappend{(\ctxinsert{\Gamma_1}{\alpha{:}(\Phi \atl{n} *)})}{\Gamma_0}\)
  \end{enumerate}
\end{lemma}

We note that when \(\m{level}(\Gamma_0) \le n\), we can derive the
following corollary using Lemma \ref{lem:chopdistrib}:

\[
  \ctxrestrict{(\ctxinsert{\Gamma_1}{\Gamma_0})}{n} = \ctxrestrict{(\ctxappend{\Gamma_1}{\Gamma_0})}{n} = \ctxrestrict{\Gamma_1}{n}
\]

The following lemmas show interactions between context operations and substitutions:

\begin{lemma}[Chopping\--Substitution Commutativity]\label{lem:chopsubstcommute}\mbox{}\quad
  \begin{enumerate}
  \item \([(\hat\Phi^n.T)/\alpha^n](\ctxrestrict{\Gamma}{n}) = \ctxrestrict{([(\hat\Phi^n.T)/\alpha^n]\Gamma)}{n}\)
  \item \([(\hat\Phi^n.e)/u](\ctxrestrict{\Gamma}{n}) = \ctxrestrict{([(\hat\Phi^n.e)/y]\Gamma)}{n}\)
  \item \([\sigma/\hat\Phi](\ctxrestrict{\Gamma}{n}) = \ctxrestrict{([\sigma/\hat\Phi]\Gamma)}{n}\)
  \end{enumerate}
\end{lemma}

\begin{lemma}[Substitution Distribution Over Merge]\label{lem:substdistrib}\mbox{}\quad
  \begin{enumerate}
  \item \([(\hat\Phi^n.T)/\alpha^n](\ctxinsert{\Gamma_1}{\Gamma_0}) = \ctxinsert{[(\hat\Phi^n.T)/\alpha^n]\Gamma_1}{[(\hat\Phi^n.T)/\alpha^n]\Gamma_0}\)
  \item \([(\hat\Phi^n.T)/\alpha^n](\ctxappend{\Gamma_1}{\Gamma_0}) = \ctxappend{[(\hat\Phi^n.T)/\alpha^n]\Gamma_1}{[(\hat\Phi^n.T)/\alpha^n]\Gamma_0}\)
  \item \([(\hat\Phi^n.e)/u](\ctxinsert{\Gamma_1}{\Gamma_0}) = \ctxinsert{[(\hat\Phi^n.e)/u]\Gamma_1}{[(\hat\Phi^n.e)/u]\Gamma_0}\)
  \item \([(\hat\Phi^n.e)/u](\ctxappend{\Gamma_1}{\Gamma_0}) = \ctxappend{[(\hat\Phi^n.e)/u]\Gamma_1}{[(\hat\Phi^n.e)/u]\Gamma_0}\)
  \item \([\sigma/\hat\Phi](\ctxinsert{\Gamma_1}{\Gamma_0}) = \ctxinsert{[\sigma/\hat\Phi]\Gamma_1}{[\sigma/\hat\Phi]\Gamma_0}\)
  \item \([\sigma/\hat\Phi](\ctxappend{\Gamma_1}{\Gamma_0}) = \ctxappend{[\sigma/\hat\Phi]\Gamma_1}{[\sigma/\hat\Phi]\Gamma_0}\)
  \end{enumerate}
\end{lemma}
}

\section{Substitution Operations}\label{sec:substop}
We define here several substitution operations.

\subsection{Chopping, Merging, and Appending Simultaneous
  Substitutions}

Dual to the operations for contexts, we rely on the corresponding
operations for simultaneous substitutions, which we define below.

\[
  \begin{array}{rcll}
\multicolumn{3}{l}{\mbox{Chopping lower part of the substitution:} \quad\ctxrestrict{\sigma/\Psihat}{n}~ =~  \sigma'/\Psihat'
}\\
    \ctxrestrict{(\cdot/\cdot)}n &= & \cdot/\cdot \\ \relax
    \ctxrestrict{(\sigma, \Phihat^k.e) / (\Psihat, x^k)}n &= & \ctxrestrict{(\sigma/\Psihat)}n
    &\text{if $k < n$}
    \\
    \ctxrestrict{(\sigma, \Phihat^k.e) / (\Psihat, x^k)}n &= & (\sigma, \Phihat^k.e)/ (\Psihat, x^k)
    &\text{otherwise}
    \\
    \ctxrestrict{(\sigma; x^k) / (\Psihat, x^k)}n &= & \ctxrestrict{(\sigma/\Psihat)}n
    &\text{if $k < n$}
    \\
    \ctxrestrict{(\sigma; x^k) / (\Psihat, x^k)}n &= & (\sigma;x^k)/ (\Psihat, x^k)
    &\text{otherwise}
    \\
    \ctxrestrict{(\sigma, \Phihat^k.T) / (\Psihat, \alpha^k)}n &= & \ctxrestrict{(\sigma/\Psihat)}n
    &\text{if $k < n$}
    \\
    \ctxrestrict{(\sigma, \Phihat^k.T) / (\Psihat, \alpha^k)}n &= & (\sigma, \Phihat^k.T)/ (\Psihat, \alpha^k)
    &\text{otherwise}
    \\
    \ctxrestrict{(\sigma; \alpha^k) / (\Psihat, \alpha^k)}n &= & \ctxrestrict{(\sigma/\Psihat)}n
    &\text{if $k < n$}
    \\
    \ctxrestrict{(\sigma; \alpha^k) / (\Psihat, \alpha^k)}n &= & (\sigma;\alpha^k)/ (\Psihat, \alpha^k)
    &\text{otherwise}
  \end{array}
\]

We define appending of two substitutions below. Note that we will only
use append, when the lowest level of the variables in $\Psihat$ is
greater than or equal to $\m{level}(\Phihat) - 1$.

\[
  \begin{array}{lcll}
    \multicolumn{3}{l}{\mbox{Appending substitutions:} \ctxappend{\sigma/\Psihat}{\rho/\Phihat} = \sigma'/\Gammahat }\\
    \ctxappend{\sigma/\Psihat}{\cdot}                         &= & \sigma/\Psi\\
    \ctxappend{\sigma/\Psihat}{(\rho, \Gammahat^n.e)/(\Phihat, x^n)} &= & (\ctxappend{\sigma/\Psihat}{\rho/\Phihat}), \Gammahat^n.e\\
    \ctxappend{\sigma/\Psihat}{(\rho, \Gammahat^n.T)/(\Phihat, \alpha^n)} &= & (\ctxappend{\sigma/\Psihat}{\rho/\Phihat}), \Gammahat^n.T\\
    \ctxappend{\sigma/\Psihat}{(\rho; x^n)/(\Phihat, x^n)} &= & (\ctxappend{\sigma/\Psihat}{\rho/\Phihat}); x^n\\
    \ctxappend{\sigma/\Psihat}{(\rho; \alpha^n)/(\Phihat, \alpha^n)} &= & (\ctxappend{\sigma/\Psihat}{\rho/\Phihat}); \alpha^n
  \end{array}
\]

\subsection{Variable Lookup} 
In our definition of the substitution operation, we rely on looking up
a mapping for a variable $x$ in the substitution
$\sigma/\Psihat$. This operation is defined recursively on the
structure of $\sigma/\Psihat$. 

  \begin{displaymath}
    \renewcommand{\arraystretch}{1.2}
    \begin{array}{l@{~=~}l@{\qquad\qquad}r}
\multicolumn{3}{l}{\mbox{Variable Lookup \(\LKP{\sigma/\Psihat}x \)}}\\[1em]
      \LKP{(\sigma,\hat\Phi^n.e)/(\Psihat,x^n)}x & (\hat\Phi^n.e) &\\
      \LKP{(\sigma,\hat\Phi^n.e)/(\Psihat,y^n)}x & \LKP{\sigma/\Psihat}x & x \neq y~\mbox{or}~k \neq n\\
      \LKP{(\sigma;y^n)/(\Psihat,x^n)}x & y^n &\\
      \LKP{(\sigma;w^k)/(\Psihat,y^n)}x & \LKP{\sigma/\Psihat}x & x \neq y ~\mbox{or}~k \neq n\\[0.5em]
      \LKP{(\sigma,\hat\Phi^n.T)/(\Psihat,\alpha^n)}\alpha & (\hat\Phi^n.T) &\\
      \LKP{(\sigma,\hat\Phi^n.T)/(\Psihat,\beta^n)}\alpha & \LKP{\sigma/\Psihat}\alpha & \alpha \neq \beta\\
      \LKP{(\sigma;\beta^n)/(\Psihat,\alpha^n)}\alpha & \beta^n &\\
      \LKP{(\sigma;\gamma^k)/(\Psihat,\beta^n)}\alpha^n &
                                                          \LKP{\sigma/\Psihat}\alpha & \alpha \neq \beta ~\mbox{or}~k\neq n\\
    \end{array}
  \end{displaymath}

\subsection{Simultaneous Substitution Operation}\label{sec:substerm}

We define substitution operation for terms and substitutions
in a similar fashion to the substitution operation for types.
For variables, we again distinguish between three cases:
If $x$ is in $\Psihat$, we again consider the case where $x$ is mapped
to $\Phihat^n.e$ and where $x$ is simply renamed and mapped to $y$.

When we push a substitution inside the body of a function $\fn{x}{e}$
or $\tfn{\alpha^n}{e}$, we extend the substitution with an identity
mapping for the variable $x$ or $\alpha$ resp if its level is below
\(\mbox{level}(\Psihat)\).
For applications $e_1~e_2$, we simply apply the substitution
recursively. For type applications $(e~(\Phihat^n.T)$, we apply the
substitution recursively to $e$. To apply $\sigma/\Psihat$ to
$(\Phihat^n.T)$ we have to be a little careful: we first drop all the
mappings from $\sigma/\Psihat$ which are below level $n$
($\ctxrestrict{(\sigma/\Psihat)}n = \sigma'/\Psihat'$) and replace
those mappings with identity mappings creating the substitution
$\ctxappend{(\sigma'/\Psihat')}{(\id(\Phihat) / \Phihat)}$. Again,
if \(n > \mbox{level}(\Psihat)\), we do not push the substitution
further, as the substitution $(\sigma/\Psihat)$ has no effect on
$T$. The same principle is also used to push the substitution inside
$\boxm{\Phihat^n}{e}$.  

Applying the substitution to $\letboxm{(\Phihat^n.u)}{e_1}{e_2}$ is
another case where we need
a careful definition. If  $\m{level}(\Psihat) > n$, then we extend
the substitution with the identity mapping which is inserted at the
appropriate position (written here as
$\ctxinsert{(\sigma/\Psihat)}{(u/u^n)}$). Otherwise, $u$ is a
``global'' variable where $\m{level}(\Psihat) \leq n$, and we apply
the substitution as-is to $e_2$  preserving the level of \(\Psihat\).

  \begin{displaymath}
    \renewcommand{\arraystretch}{1.2}
    \begin{array}{lcl}
      \SSubst{\sigma/\Psihat}{(x[\sigma'])} & = &
          x[\sigma'']\hfill
          \qquad\quad x \not\in \hat\Psi~\mbox{and}~\SSubst{\sigma/\Psihat}{\sigma'} = \sigma''
          \\
      \SSubst{\sigma/\Psihat}{(x[\sigma'])} & = &
          e'
          \qquad\quad \hfill
          \LKP{\sigma /\Psihat}{x} = (\Phihat^n.e)~\mbox{and}~
          \SSubst{\sigma/\Psihat}{\sigma'} = \sigma''~\mbox{and}~\\
      & & \quad \hfill
           \m{level}(\Psihat) > n~\mbox{and}~
            \SSubst{\sigma''/\Phihat}{e} = e'
            \\
      \SSubst{\sigma/\Psihat}{(x[\sigma'])} & = &
          y[\sigma'']
          \qquad \hfill
          \quad\LKP{\sigma /\Psihat}{x} = y~
          \mbox{and}~
          \SSubst{\sigma/\Psihat}{\sigma'} = \sigma''
          \\
      \SSubst{\sigma/\Psihat}{(\fn{x}{e})} & = &
          \fn{x}{e'} \hfill
          \SSubst{\ctxinsert{\sigma/\Psihat}{x/x^{0}}}{e} = e'
          \\
      \SSubst{\sigma/\Psihat}{(\tfn{\alpha^n}{e})} & = &
          \tfn{\alpha^n}{e'} \hfill
          \mbox{level}(\Psihat) > n \mbox{ and } \SSubst{\ctxinsert{\sigma/\Psihat}{\alpha/\alpha^{n}}}{e} = e'
          \\
      \SSubst{\sigma/\Psihat}{(\tfn{\alpha^n}{e})} & = &
          \tfn{\alpha^n}{e'} \hfill
          \mbox{level}(\Psihat) \le n \mbox{ and } \SSubst{\sigma/\Psihat}{e} = e'
          \\
      \SSubst{\sigma/\Psihat}{(e_1\;e_2)} & = &
          e'_1~e'_2 \hfill
          \SSubst{\sigma/\Psihat}{e_1} = e'_1
          \mbox{ and } \SSubst{\sigma/\Psihat}{e_2} = e'_2
          \\
      \SSubst{\sigma/\Psihat}{(e~(\Phihat^n.T))} & = &
          e'~(\Phihat^n.T') \hfill
          \mbox{level}(\Psihat) \ge n \mbox{ and } \SSubst{\sigma/\Psihat}{e} = e' \mbox{ and } \\
      & & \quad \hfill
          \ctxrestrict{(\sigma/\Psihat)}n = \sigma'/\Psihat' \mbox{ and }
          \SSubst{\ctxappend{(\sigma'/\Psihat')}{(\id(\Phihat)/\Phihat)}}{T} = T'
          \\
      \SSubst{\sigma/\Psihat}{(e~(\Phihat^n.T))} & = &
          e'~(\Phihat^n.T) \hfill
          \mbox{level}(\Psihat) < n \mbox{ and } \SSubst{\sigma/\Psihat}{e} = e'
          \\
      \SSubst{\sigma/\Psihat}{(\boxm{\Phihat^n}{e})} & = &
          \boxm{\Phihat^n}{e'} \qquad\quad \hfill
          \mbox{level}(\Psihat) \ge n \mbox{ and } \ctxrestrict{(\sigma/\Psihat)}n = (\sigma'/\Psihat') \mbox{ and } \\
      & & \quad \hfill
          \SSubst{\ctxappend{(\sigma'/\Psihat')}{(\id(\Phihat)/\Phihat)}}{e} = e'
          \\
      \SSubst{\sigma/\Psihat}{(\boxm{\Phihat^n}{e})} & = &
          \boxm{\Phihat^n}{e} \qquad\quad \hfill
          \mbox{level}(\Psihat) < n
          \\
      \SSubst{\sigma/\Psihat}{(\letboxm{(\Phihat^n.u)}{e_1}{e_2})} & = &
          \letboxm{(\Phihat^n.u)}{e'_1}{e'_2}
          \hfill
          \mbox{level}(\Psihat) > n \mbox{ and } \\
      & & \quad \hfill
          \SSubst{\sigma/\Psihat}{e_1}
          = e'_1~\mbox{and}~
          \SSubst{\ctxinsert{(\sigma/\Psihat)}{(u/u^n)}}{e_2} = e_2'
          \\
      \SSubst{\sigma/\Psihat}{(\letboxm{(\Phihat^n.u)}{e_1}{e_2})} & = &
          \letboxm{(\Phihat^n.u)}{e'_1}{e'_2}
          \hfill
          \mbox{level}(\Psihat) \le n \mbox{ and }\\
      & & \quad \hfill
            \SSubst{\sigma/\Psihat}{e_1} = e'_1 \mbox{ and }
            \SSubst{\sigma/\Psihat}{e_2} = e_2'
    \end{array}
  \end{displaymath}

We give here the simultaneous substitution operation for composing two
substitutions $\sigma$ and $\rho$. We only consider extensions of
substitutions with the terms below, but extensions with types follow a
similar principle.

\[
  \begin{array}{lcl}
\SSubst{\sigma/\Psihat}{(~\cdot~)} & = & \cdot
\\[1em]
\SSubst{\sigma/\Psihat}{(\rho; x^n)} & = &
\rho'; x^n
\hfill \quad x \not\in \Psihat \text{ and } 
\SSubst{\sigma/\Psihat}\rho = \rho'
\\
\SSubst{\sigma/\Psihat}{(\rho; x^n)}  & = &
\rho', (\Gammahat^n.e) \hfill \quad
\m{level}(\Psihat) > n~\mbox{and}~
\LKP{\sigma /\Psihat}{x} = (\Gammahat^n.e)~ ~\mbox{and}~ 
\SSubst{\sigma/\Psihat}\rho = \rho'
\\
\SSubst{\sigma/\Psihat}{(\rho; x^n)} & = &
\rho'; y^n \hfill \quad
\m{level}(\Psihat) > n~\mbox{and}~
\LKP{\sigma /\Psihat}{x} = y^n~ ~\mbox{and}~ 
\SSubst{\sigma/\Psihat}\rho = \rho'
\\[1em]
\SSubst{\sigma/\Psihat}{(\rho, (\Gammahat^n.e))}  & = &
\rho', (\Gammahat^n.e)
\hfill x \not\in \Psihat \text{ and } 
\SSubst{\sigma/\Psihat}\rho = \rho'
 \text{ and } 
\SSubst{\sigma/\Psihat}\rho = \rho'
\\
\SSubst{\sigma/\Psihat}{(\rho, (\Gammahat^n.e))}  & = &
\rho', (\Gammahat^n.e') \hfill \quad
\m{level}(\Psihat) > n~\mbox{and}~
\SSubst{\sigma/\Psihat}\rho = \rho' ~\mbox{and}~ \\
& & \quad \hfill 
\ctxrestrict{(\sigma/\Psihat)}n = \sigma'/\Psihat' \mbox{ and }
\SSubst{\ctxappend{(\sigma'/\Psihat')}{(\id(\Gammahat)/\Gammahat)}}{e} = e'
\\[1em]
  \end{array}
\]

Last, we give the simultaneous substitution operation on contexts. The tricky case is how to push the substitution $(\sigma/\Psihat)$ into a context $(\Gamma, x{:}(\Phi \atl{n} T))$. Here we first apply the substitution to $\Gamma$; then we intuitively want to apply the substitution $\ctxappend{(\sigma/\Psihat)}{(\id(\Gammahat)/\Gammahat)}$ to $(\Phi \atl{n} T)$. Since $\Phi$ will replace variables at level below $n$ in $\ctxappend{\Psihat}{\Gammahat}$, we first drop those variables from the extended substitution (i.e. $\ctxrestrict{(\ctxappend{(\sigma/\Psihat)}{(\id(\Gammahat)/\Gammahat)})}{n} = \rho/\Upsilon$) and then apply it to $\Phi$. To push the substitution int $T$, we extend $(\rho/\Upsilon)$ with the identity mapping for all the variables in $\Phi$.

\begin{figure}[ht]
  \centering
\[
\begin{array}{lcl}
\SSubst{\sigma/\Psihat}{(.)} &  = & . \\
\SSubst{\sigma/\Psihat}{(\Gamma, x{:}(\Phi \atl{n} *))} & = & 
\Gamma', x{:}(\Phi \atl{n} *)
\hfill\quad
\SSubst{\sigma/\Psihat}{\Gamma} = \Gamma'
\\
\SSubst{\sigma/\Psihat}{(\Gamma, x{:}(\Phi \atl{n} T))} & = & 
\Gamma', x{:}(\Phi' \atl{n} T')
\hfill\quad
\SSubst{\sigma/\Psihat}{\Gamma} = \Gamma' ~\mbox{and}
\SSubst{\ctxrestrict{(\ctxappend{(\sigma/\Psihat)}{(\id(\Gammahat)/\Gammahat)})}{n}}{\Phi}
                                                              = \Phi'
                                                              ~\mbox{and}~\\
& & \quad \hfill
\SSubst{\ctxappend{\ctxrestrict{(\ctxappend{(\sigma/\Psihat)}{(\id(\Gammahat)/\Gammahat)})}{n}}{\id(\Phihat)/\Phihat}}{T} = T'      
    \end{array}
\]
  \caption{Simultaneous Substitution Operation on Context:
    $\SSubst{\sigma/\Psihat}{\Phi} = \Phi'$}
  \label{fig:ctxsubst}
\end{figure}

\section{Typing Rules for Simultaneous Substitution}\label{sec:subtyp}

We give here in full the typing rules for simultaneous substitutions.

\[
\begin{array}{c}
\infer{\Gamma \vde \cdot : \cdot}{}
\quad
\infer{\Gamma \vde (\sigma, \Psihat^n.e) ~:~ (\Gamma', x{:}(\Psi \atl{n} T))}{
 \Gamma \vde \sigma : \Gamma' &
\ctxrestrict{(\sigma/\Gammahat')}{n} = (\sigma'/\Gammahat'') &
\ctxappend{\ctxrestrict{\Gamma}{n}}{[\sigma'/\Gammahat'']\Psi} \vde
   e : [\ctxappend{(\sigma'/\Gammahat')}{(\id(\Psihat)/\Psihat)}]T}
\\[1em]
\infer{\Gamma \vde (\sigma;x) ~:~ (\Gamma', x{:}(\Psi \atl{n} T))}
{\Gamma \vde \sigma : \Gamma' &
 \ctxrestrict{\sigma/\Gammahat'}{n} = (\sigma''/\Gammahat'') &
 \Gamma(x) = ([\sigma''/\Gammahat'']\Psi \atl{n} [\sigma''/\Gammahat'']T)
}
\\[1em]
\infer{\Gamma \vde (\sigma, \Psihat^n.T) ~:~ (\Gamma', \alpha{:}(\Psi \atl{n} *))}{
 \Gamma \vde \sigma : \Gamma' &
\ctxrestrict{(\sigma/\Gammahat')}{n} = (\sigma'/\Gammahat'') &
\ctxappend{\ctxrestrict{\Gamma}{n}}{[\sigma'/\Gammahat'']\Psi} \vde T
}
\\[1em]
\infer{\Gamma \vde (\sigma;\alpha) ~:~ (\Gamma', \alpha{:}(\Psi \atl{n} *))}
{\Gamma \vde \sigma : \Gamma' &
 \ctxrestrict{(\sigma/\Gammahat')}{n} = (\sigma''/\Gammahat'') &
 \Gamma(\alpha) = ([\sigma''/\Gammahat'']\Psi \atl{n} *)
}
\end{array}
\]

\section{Proof of Lemma~\ref{lem:typesubst},~\ref{lem:termsubst}~and~\ref{lem:simsubst}}\label{sec:proofsubsts}
\newtheoremstyle{TheoremNum}
        {\topsep}{\topsep}              
        {\itshape}                      
        {}                              
        {\bfseries}                     
        {.}                             
        { }                             
        {\thmname{#1}\thmnote{ \bfseries #3}}
    \theoremstyle{TheoremNum}
    \newtheorem{thmn}{Lemma}

Here we state some extra lemmas for the proofs:

\begin{lemma}[Chopping Substitution]\label{lem:substchop} If \(\Gamma \vde \sigma : \Phi\) then \(\ctxrestrict{\Gamma}{m} \vde \sigma' : \Phi'\) for \(\sigma'/\Phi' = \ctxrestrict{(\sigma/\Phi)}{m}\)
\end{lemma}

\begin{proof}
  By induction on the first derivation.
\end{proof}

\begin{lemma}[Substitution Extension]\label{lem:substext}\mbox{}
  \begin{enumerate}
  \item If \(\Gamma \vde \sigma : \Phi\)\\
    then \(\ctxinsert{\Gamma}{\alpha{:}(\Psi \atl{m} *)} \vde \ctxinsert{\sigma}{\alpha^m} : \ctxinsert{\Phi}{\alpha{:}(\Psi \atl{m} *)}\)
  \item If \(\Gamma \vde \sigma : \Phi\) and \(\ctxapp{\ctxrestrict{\Phi}{m}}{\Psi} \vde S\)\\
    then \(\ctxinsert{\Gamma}{u{:}([\ctxrestrict{(\sigma/\hat\Phi)}{m}]\Psi \atl{m} [\ctxappend{\ctxrestrict{(\sigma/\hat\Phi)}{m}}{(\id(\hat\Psi)/\hat\Psi)}]S)} \vde \ctxinsert{\sigma}{u^m} : \ctxinsert{\Phi}{u{:}(\Psi \atl{m} S)}\)
  \end{enumerate}
\end{lemma}

\begin{proof}
  By induction on the first derivation.
\end{proof}

\begin{lemma}[Type Ignores Term Substitution]\label{lem:typeignoresubst}
  \([\ctxinsertatl{(\sigma/\hat\Phi)}{u/u}{m}]T = [\sigma/\hat\Phi]T\)
\end{lemma}

\begin{proof}
  By induction on \(T\).
\end{proof}

\begin{lemma}[Property of Identity Substitution]\label{lem:idsubstprop}\mbox{}
  \begin{enumerate}
  \item \([\id(\hat\Phi)/\hat\Phi]T = T\)
  \item \([\id(\hat\Phi)/\hat\Phi]e = e\)
  \end{enumerate}
\end{lemma}

\begin{proof}
  By induction on \(T\) or \(e\).
\end{proof}

\begin{thmn}[\ref{lem:typesubst}]\mbox{\textnormal{\textsc{(Type Substitution Lemma).}}}\quad
Assuming
  $\vde \ctxapp{\ctxapp{\Gamma_1}{\alpha{:}(\Phi \atl{n}
      *)}}{\Gamma_0}$ and is ordered such that all assumptions
  in $\Gamma_1$ are at level $n$ or above and all assumptions in
  $\Gamma_0$ are at level $n$ or below.
  \begin{enumerate}
  \item If $\vde \ctxapp{\ctxapp{\Gamma_1}{\alpha{:}(\Phi \atl{n} *)}}{\Gamma_0}$
and $\ctxapp{\Gamma_1}{\Phi} \vde T$\\
then $\vde \ctxapp{\Gamma_1}{[(\Phihat^n.T)/\alpha^n]\Gamma_0}$.
  \item If $\ctxapp{\ctxapp{\Gamma_1}{\alpha{:}(\Phi \atl{n} *)}}{\Gamma_0} \vde S$
and $\ctxapp{\Gamma_1}{\Phi} \vde T$\\
then $\ctxapp{\Gamma_1}{[(\Phihat^n.T)/\alpha^n]\Gamma_0} \vde [(\Phihat^n.T)/\alpha^n]S$.
  \item If $\ctxapp{\ctxapp{\Gamma_1}{\alpha{:}(\Phi \atl{n} *)}}{\Gamma_0} \vde e_2 : S$
and $\ctxapp{\Gamma_1}{\Phi} \vde T$\\
then $\ctxapp{\Gamma_1}{[(\Phihat^n.T)/\alpha^n]\Gamma_0} \vde [(\Phihat^n.T)/\alpha^n]e_2 : [(\Phihat^n.T)/\alpha^n]S$.
  \item If $\ctxapp{\ctxapp{\Gamma_1}{\alpha{:}(\Phi \atl{n} *)}}{\Gamma_0} \vde \sigma : \Psi$
and $\ctxapp{\Gamma_1}{\Phi} \vde T$\\
then $\ctxapp{\Gamma_1}{[(\Phihat^n.T)/\alpha^n]\Gamma_0} \vde [(\Phihat^n.T)/\alpha^n]\sigma : [(\Phihat^n.T)/\alpha^n]\Psi$.
  \end{enumerate}
\end{thmn}

\begin{thmn}[\ref{lem:termsubst}]\mbox{\textnormal{\textsc{(Term Substitution Lemma).}}}\quad
  \begin{enumerate}
  \item If $\ctxapp{\ctxapp{\Gamma_1}{u{:}(\Phi \atl{n} T)}}{\Gamma_0} \vde e_2 : S$
and $\ctxapp{\Gamma_1}{\Phi} \vde e : T$
then $\ctxapp{\Gamma_1}{\Gamma_0} \vde [(\Phihat^n.e)/u]e_2 : S$.
  \item If $\ctxapp{\ctxapp{\Gamma_1}{u{:}(\Phi \atl{n} T)}}{\Gamma_0} \vde \sigma : \Psi$
and $\ctxapp{\Gamma_1}{\Phi} \vde e : T$
then $\ctxapp{\Gamma_1}{\Gamma_0} \vde [(\Phihat^n.e)/u]\sigma : \Psi$.
  \end{enumerate}
\end{thmn}

\begin{thmn}[\ref{lem:simsubst}]\mbox{\textnormal{\textsc{(Simultaneous Substitution Lemma).}}}\quad
  \begin{enumerate}
  \item If $\vde \ctxapp{\ctxapp{\Gamma_1}{\Gamma_0}}{\Psi}$ and $\ctxapp{\Gamma_1}{\Phi} \vde \sigma : \Gamma_0$
then $\vde \ctxapp{\ctxapp{\Gamma_1}{\Phi}}{([\sigma/\hat\Gamma_0]\Psi)}$.
  \item If $\ctxapp{\Gamma_1}{\Gamma_0} \vde S$ and $\ctxapp{\Gamma_1}{\Phi} \vde \sigma : \Gamma_0$
then $\ctxapp{\Gamma_1}{\Phi} \vde [\sigma/\hat\Gamma_0]S$.
  \item If $\ctxapp{\Gamma_1}{\Gamma_0} \vde e : S$ and $\ctxapp{\Gamma_1}{\Phi} \vde \sigma : \Gamma_0$
then $\ctxapp{\Gamma_1}{\Phi} \vde [\sigma/\hat\Gamma_0]e : [\sigma/\hat\Gamma_0]S$.
  \item If $\ctxapp{\Gamma_1}{\Gamma_0} \vde \sigma_2 : \Psi$ and $\ctxapp{\Gamma_1}{\Phi} \vde \sigma : \Gamma_0$
then $\ctxapp{\Gamma_1}{\Phi} \vde [\sigma/\hat\Gamma_0]\sigma_2 : [\sigma/\hat\Gamma_0]\Psi$.
  \end{enumerate}
\end{thmn}

\begin{proof}
  By lexicographic induction on the level of \(\Gamma_0\) and first
  derivation. Either the level decreases or the level stays the same
  and the first derivation decreases. Here, we covered only the
  interesting cases.

  \paragraph{Cases for type substitution~(\ref{itm:typesubst-type}):}
  \begin{itemize}
  \item Case:
    \[
      \infer%
      {\ctxapp{\ctxapp{\Gamma_1}{\alpha{:}(\Phi \atl{n} *)}}{\Gamma_0} \vde \beta[\sigma]}%
      {%
        \deduce%
        {(\ctxapp{\ctxapp{\Gamma_1}{\alpha{:}(\Phi \atl{n} *)}}{\Gamma_0})(\beta) = (\Psi \atl{m} *)}%
        {\mathcal{D}_1}%
        &%
        \deduce%
        {\ctxapp{\ctxapp{\Gamma_1}{\alpha{:}(\Phi \atl{n} *)}}{\Gamma_0} \vde \sigma : \Psi}%
        {\mathcal{D}_2}%
      }%
    \]
    \begin{itemize}
    \item Subcase: \(\beta = \alpha\)\\
      \(\ctxapp{\ctxapp{\Gamma_1}{\alpha{:}(\Phi \atl{n} *)}}{\Gamma_0} \vde \sigma : \Phi\) and \(\Psi = \Phi\) \hfill by \(\mathcal{D}_2\)\\
      \(\ctxapp{\Gamma_1}{[(\hat\Phi^n.T)/\alpha^n]\Gamma_0} \vde [(\hat\Phi^n.T)/\alpha^n]\sigma : [(\hat\Phi^n.T)/\alpha^n]\Phi\) \hfill by IH (type subst.~(\ref{itm:typesubst-subst}))\\
      \(\ctxapp{\Gamma_1}{[(\hat\Phi^n.T)/\alpha^n]\Gamma_0} \vde [(\hat\Phi^n.T)/\alpha^n]\sigma : \Phi\) \hfill \(\Phi\) contains type variables only\\
      \(\ctxapp{\Gamma_1}{\Phi} \vde T\) \hfill by assumption\\
      \(\ctxapp{\Gamma_1}{[(\hat\Phi^n.T)/\alpha^n]\Gamma_0} \vde [[(\hat\Phi^n.T)/\alpha^n]\sigma/\Phihat]T\) \hfill by IH (sim.~subst.~(\ref{itm:simsubst-type}))\\
      \(\ctxapp{\Gamma_1}{[(\hat\Phi^n.T)/\alpha^n]\Gamma_0} \vde [(\hat\Phi^n.T)/\alpha^n](\alpha[\sigma])\) \hfill by the def.~of substitution\\

    \item Subcase: \(\beta \not= \alpha\)\\
      \((\ctxapp{\ctxapp{\Gamma_1}{\alpha{:}(\Phi \atl{n} *)}}{\Gamma_0})(\beta) = (\Psi \atl{m} *)\) \hfill by \(\mathcal{D}_1\)\\
      \((\ctxapp{\Gamma_1}{[(\Phihat^n.T)/\alpha^n]\Gamma_0})(\beta) = (\Psi \atl{m} *)\)\\
      \mbox{} \hfill context subst.~distrib.~\ref{lem:substdistrib} and \(\Psi\) contains type variables only\\
      \(\ctxapp{\ctxapp{\Gamma_1}{\alpha{:}(\Phi \atl{n} *)}}{\Gamma_0} \vde \sigma : \Psi\) \hfill by \(\mathcal{D}_2\)\\
      \(\ctxapp{\Gamma_1}{[(\Phihat^n.T)/\alpha^n]\Gamma_0} \vde [(\hat\Phi^n.T)/\alpha^n]\sigma : [(\hat\Phi^n.T)/\alpha^n]\Psi\) \hfill by IH (type subst.~(\ref{itm:typesubst-subst}))\\
      \(\ctxapp{\Gamma_1}{[(\Phihat^n.T)/\alpha^n]\Gamma_0} \vde [(\hat\Phi^n.T)/\alpha^n]\sigma : \Psi\) \hfill \(\Psi\) contains type variables only\\
      \(\ctxapp{\Gamma_1}{[(\Phihat^n.T)/\alpha^n]\Gamma_0} \vde \beta[[(\hat\Phi^n.T)/\alpha^n]\sigma]\) \hfill by kinding rule\\
      \(\ctxapp{\Gamma_1}{[(\Phihat^n.T)/\alpha^n]\Gamma_0} \vde [(\hat\Phi^n.T)/\alpha^n](\beta[\sigma])\) \hfill by the def.~of substitution\\
    \end{itemize}


  \item Case:
    \[
      \infer%
      {\ctxapp{\ctxapp{\Gamma_1}{\alpha{:}(\Phi \atl{n} *)}}{\Gamma_0} \vde (\beta{:}(\Psi \atl{m} *)) \arrow S}%
      {%
        \deduce%
        {\vde \ctxapp{\ctxrestrict{(\ctxapp{\ctxapp{\Gamma_1}{\alpha{:}(\Phi \atl{n} *)}}{\Gamma_0})}{m}}{\Psi}}%
        {\mathcal{D}_1}%
        &%
        \deduce%
        {\ctxinsert{(\ctxapp{\ctxapp{\Gamma_1}{\alpha{:}(\Phi \atl{n} *)}}{\Gamma_0})}{\beta{:}(\Psi \atl{m} *)} \vde S}%
        {\mathcal{D}_2}%
      }%
    \]
    \begin{itemize}
    \item Subcase: \(m > n\)\\
      \(\vde \ctxapp{\ctxrestrict{(\ctxapp{\ctxapp{\Gamma_1}{\alpha{:}(\Phi \atl{n} *)}}{\Gamma_0})}{m}}{\Psi}\) \hfill by \(\mathcal{D}_1\)\\
      \(\vde \ctxapp{\ctxrestrict{\Gamma_1}{m}}{\Psi}\) \hfill by context chop distrib.~\ref{lem:chopdistrib}\\
      \(\vde \ctxapp{\ctxrestrict{(\ctxapp{\Gamma_1}{[(\Phihat^n.T)/\alpha^n]\Gamma_0})}{m}}{\Psi}\) \hfill by context chop distrib.~\ref{lem:chopdistrib}\\
      \(\ctxinsert{(\ctxapp{\ctxapp{\Gamma_1}{\alpha{:}(\Phi \atl{n} *)}}{\Gamma_0})}{\beta{:}(\Psi \atl{m} *)} \vde S\) \hfill by \(\mathcal{D}_2\)\\
      \(\ctxapp{\ctxapp{(\ctxinsert{\Gamma_1}{\beta{:}(\Psi \atl{m} *)})}{\alpha{:}(\Phi \atl{n} *)}}{\Gamma_0} \vde S\) \hfill by context merge reordering~\ref{lem:mergereorder} \(m > n\)\\
      \(\ctxappend{(\ctxinsert{\Gamma_1}{\beta{:}(\Psi \atl{m} *)})}{[(\Phihat^n.T)/\alpha^n]\Gamma_0} \vde [(\Phihat^n.T)/\alpha^n]S\) \hfill by IH (type subst.~(\ref{itm:typesubst-type}))\\
      \(\ctxinsert{(\ctxapp{\Gamma_1}{[(\Phihat^n.T)/\alpha^n]\Gamma_0})}{\beta{:}(\Psi \atl{m} *)} \vde [(\Phihat^n.T)/\alpha^n]S\) \hfill by context merge reordering~\ref{lem:mergereorder} \(m > n\)\\
      \(\ctxapp{\Gamma_1}{[(\Phihat^n.T)/\alpha^n]\Gamma_0} \vde (\beta{:}(\Psi \atl{m} *)) \arrow [(\Phihat^n.T)/\alpha^n]S\) \hfill by kinding rule\\
      \(\ctxapp{\Gamma_1}{[(\Phihat^n.T)/\alpha^n]\Gamma_0} \vde [(\Phihat^n.T)/\alpha^n]((\beta{:}(\Psi \atl{m} *)) \arrow S)\) \hfill by the def.~of substitution\\

    \item Subcase: \(m \le n\)\\
      \(\vde \ctxapp{\ctxrestrict{(\ctxapp{\ctxapp{\Gamma_1}{\alpha{:}(\Phi \atl{n} *)}}{\Gamma_0})}{m}}{\Psi}\) \hfill by \(\mathcal{D}_1\)\\
      \(\vde \ctxapp{\ctxapp{\ctxapp{\Gamma_1}{\alpha{:}(\Phi \atl{n} *)}}{\ctxrestrict{\Gamma_0}{m}}}{\Psi}\) \hfill by context chop distrib.~\ref{lem:chopdistrib}\\
      \(\vde \ctxapp{\Gamma_1}{[(\Phihat^n.T)/\alpha^n](\ctxapp{\ctxrestrict{\Gamma_0}{m}}{\Psi})}\) \hfill by IH (type subst.~(\ref{itm:typesubst-context}))\\
      \(\vde \ctxapp{\Gamma_1}{\ctxapp{[(\Phihat^n.T)/\alpha^n]\ctxrestrict{\Gamma_0}{m}}{[(\Phihat^n.T)/\alpha^n]\Psi}}\) \hfill by commuting chop subst.~\ref{lem:chopsubstcommute}\\
      \(\vde \ctxapp{\ctxrestrict{(\ctxapp{\Gamma_1}{[(\Phihat^n.T)/\alpha^n]\Gamma_0})}{m}}{[(\Phihat^n.T)/\alpha^n]\Psi}\) \hfill by context chop distrib.~\ref{lem:chopdistrib}\\
      \(\vde \ctxapp{\ctxrestrict{(\ctxapp{\Gamma_1}{[(\Phihat^n.T)/\alpha^n]\Gamma_0})}{m}}{\Psi}\) \hfill \(\Psi\) contains type variables only\\
      \(\ctxinsert{(\ctxapp{\ctxapp{\Gamma_1}{\alpha{:}(\Phi \atl{n} *)}}{\Gamma_0})}{\beta{:}(\Psi \atl{m} *)} \vde S\) \hfill by \(\mathcal{D}_2\)\\
      \(\ctxapp{\ctxapp{\Gamma_1}{\alpha{:}(\Phi \atl{n} *)}}{(\ctxinsert{\Gamma_0}{\beta{:}(\Psi \atl{m} *)})} \vde S\) \hfill by context merge reordering~\ref{lem:mergereorder} \(m \le n\)\\
      \(\ctxapp{\Gamma_1}{[(\Phihat^n.T)/\alpha^n](\ctxinsert{\Gamma_0}{\beta{:}(\Psi \atl{m} *)})} \vde [(\Phihat^n.T)/\alpha^n]S\) \hfill by IH (type subst.~(\ref{itm:typesubst-type}))\\
      \(\ctxapp{\Gamma_1}{(\ctxinsert{[(\Phihat^n.T)/\alpha^n]\Gamma_0}{[(\Phihat^n.T)/\alpha^n](\beta{:}(\Psi \atl{m} *))})} \vde [(\Phihat^n.T)/\alpha^n]S\)\\
      \mbox{} \hfill by context subst.~distrib.~\ref{lem:substdistrib}\\
      \(\ctxapp{\Gamma_1}{(\ctxinsert{[(\Phihat^n.T)/\alpha^n]\Gamma_0}{\beta{:}(\Psi \atl{m} *)})} \vde [(\Phihat^n.T)/\alpha^n]S\) \hfill \(\Psi\) contains type variables only\\
      \(\ctxinsert{(\ctxapp{\Gamma_1}{[(\Phihat^n.T)/\alpha^n]\Gamma_0})}{\beta{:}(\Psi \atl{m} *)} \vde [(\Phihat^n.T)/\alpha^n]S\) \hfill by context merge reordering~\ref{lem:mergereorder} \(m \le n\)\\
      \(\ctxapp{\Gamma_1}{[(\Phihat^n.T)/\alpha^n]\Gamma_0} \vde (\beta{:}(\Psi \atl{m} *)) \arrow [(\Phihat^n.T)/\alpha^n]S\) \hfill by kinding rule\\
      \(\ctxapp{\Gamma_1}{[(\Phihat^n.T)/\alpha^n]\Gamma_0} \vde [(\Phihat^n.T)/\alpha^n]((\beta{:}(\Psi \atl{m} *)) \arrow S)\) \hfill by the def.~of substitution\\
    \end{itemize}

  \item Case:
    \[
      \infer%
      {\ctxapp{\ctxapp{\Gamma_1}{\alpha{:}(\Phi \atl{n} *)}}{\Gamma_0} \vde \cbox{\Psi \atl{m} S}}%
      {%
        \deduce%
        {\vde \ctxapp{\ctxrestrict{(\ctxapp{\ctxapp{\Gamma_1}{\alpha{:}(\Phi \atl{n} *)}}{\Gamma_0})}{m}}{\Psi}}%
        {\mathcal{D}_1}%
        &%
        \deduce%
        {\ctxapp{\ctxrestrict{(\ctxapp{\ctxapp{\Gamma_1}{\alpha{:}(\Phi \atl{n} *)}}{\Gamma_0})}{m}}{\Psi} \vde S}%
        {\mathcal{D}_2}%
      }%
    \]
    \begin{itemize}
    \item Subcase: \(m > n\)\\
      \(\vde \ctxapp{\ctxrestrict{(\ctxapp{\ctxapp{\Gamma_1}{\alpha{:}(\Phi \atl{n} *)}}{\Gamma_0})}{m}}{\Psi}\) \hfill by \(\mathcal{D}_1\)\\
      \(\vde \ctxapp{\ctxrestrict{\Gamma_1}{m}}{\Psi}\) \hfill by context chop distrib.~\ref{lem:chopdistrib}\\
      \(\vde \ctxapp{\ctxrestrict{(\ctxapp{\Gamma_1}{[(\Phihat^n.T)/\alpha^n]\Gamma_0})}{m}}{\Psi}\) \hfill by context chop distrib.~\ref{lem:chopdistrib}\\
      \(\ctxapp{\ctxrestrict{(\ctxapp{\ctxapp{\Gamma_1}{\alpha{:}(\Phi \atl{n} *)}}{\Gamma_0})}{m}}{\Psi} \vde S\) \hfill by \(\mathcal{D}_2\)\\
      \(\ctxapp{\ctxrestrict{\Gamma_1}{m}}{\Psi} \vde S\) \hfill by context chop distrib.~\ref{lem:chopdistrib}\\
      \(\ctxapp{\ctxrestrict{(\ctxapp{\Gamma_1}{[(\Phihat^n.T)/\alpha^n]\Gamma_0})}{m}}{\Psi} \vde S\) \hfill by context chop distrib.~\ref{lem:chopdistrib}\\
      \(\ctxapp{\Gamma_1}{[(\Phihat^n.T)/\alpha^n]\Gamma_0} \vde \cbox{\Psi \atl{m} S}\) \hfill by kinding rule\\
      \(\ctxapp{\Gamma_1}{[(\Phihat^n.T)/\alpha^n]\Gamma_0} \vde [(\Phihat^n.T)/\alpha^n]\cbox{\Psi \atl{m} S}\) \hfill by the def.~of substitution\\

    \item Subcase: \(m \le n\)\\
      \(\vde \ctxapp{\ctxrestrict{(\ctxapp{\ctxapp{\Gamma_1}{\alpha{:}(\Phi \atl{n} *)}}{\Gamma_0})}{m}}{\Psi}\) \hfill by \(\mathcal{D}_1\)\\
      \(\vde \ctxapp{\ctxapp{\ctxapp{\Gamma_1}{\alpha{:}(\Phi \atl{n} *)}}{\ctxrestrict{\Gamma_0}{m}}}{\Psi}\) \hfill by context chop distrib.~\ref{lem:chopdistrib}\\
      \(\vde \ctxapp{\Gamma_1}{[(\Phihat^n.T)/\alpha^n](\ctxapp{\ctxrestrict{\Gamma_0}{m}}{\Psi})}\) \hfill by IH (type subst.~(\ref{itm:typesubst-context}))\\
      \(\vde \ctxapp{\ctxapp{\Gamma_1}{[(\Phihat^n.T)/\alpha^n]\ctxrestrict{\Gamma_0}{m}}}{[(\Phihat^n.T)/\alpha^n]\Psi}\) \hfill by context subst.~distrib.~\ref{lem:substdistrib}\\
      \(\vde \ctxapp{\ctxrestrict{(\ctxapp{\Gamma_1}{[(\Phihat^n.T)/\alpha^n]\Gamma_0})}{m}}{[(\Phihat^n.T)/\alpha^n]\Psi}\) \hfill by context chop distrib.~\ref{lem:chopdistrib}\\
      \(\ctxapp{\ctxrestrict{(\ctxapp{\ctxapp{\Gamma_1}{\alpha{:}(\Phi \atl{n} *)}}{\Gamma_0})}{m}}{\Psi} \vde S\) \hfill by \(\mathcal{D}_2\)\\
      \(\ctxapp{\ctxapp{\ctxapp{\Gamma_1}{\alpha{:}(\Phi \atl{n} *)}}{\ctxrestrict{\Gamma_0}{m}}}{\Psi} \vde S\) \hfill by context chop distrib.~\ref{lem:chopdistrib}\\
      \(\ctxapp{\Gamma_1}{[(\Phihat^n.T)/\alpha^n](\ctxapp{\ctxrestrict{\Gamma_0}{m}}{\Psi})} \vde [(\Phihat^n.T)/\alpha^n]S\) \hfill by IH (type subst.~(\ref{itm:typesubst-type}))\\
      \(\ctxapp{\ctxapp{\Gamma_1}{[(\Phihat^n.T)/\alpha^n]\ctxrestrict{\Gamma_0}{m}}}{[(\Phihat^n.T)/\alpha^n]\Psi} \vde [(\Phihat^n.T)/\alpha^n]S\) \hfill by context subst.~distrib.~\ref{lem:substdistrib}\\
      \(\ctxapp{\ctxrestrict{(\ctxapp{\Gamma_1}{[(\Phihat^n.T)/\alpha^n]\Gamma_0})}{m}}{[(\Phihat^n.T)/\alpha^n]\Psi} \vde [(\Phihat^n.T)/\alpha^n]S\) \hfill by context chop distrib.~\ref{lem:chopdistrib}\\
      \(\ctxapp{\Gamma_1}{[(\Phihat^n.T)/\alpha^n]\Gamma_0} \vde \cbox{[(\Phihat^n.T)/\alpha^n]\Psi \atl{m} [(\Phihat^n.T)/\alpha^n]S}\) \hfill by kinding rule\\
      \(\ctxapp{\Gamma_1}{[(\Phihat^n.T)/\alpha^n]\Gamma_0} \vde [(\Phihat^n.T)/\alpha^n]\cbox{\Psi \atl{m} S}\) \hfill by the def.~of substitution\\
    \end{itemize}
  \end{itemize}

  \paragraph{Cases for type substitution~(\ref{itm:typesubst-term}):}
  \begin{itemize}
  \item Case:
    \[
      \infer%
      {\ctxapp{\ctxapp{\Gamma_1}{\alpha{:}(\Phi \atl{n} *)}}{\Gamma_0} \vde e_2\;(\hat\Psi^m . S_2) : [(\hat\Psi^m.S_2)/\beta^m]S}%
      {%
        \deduce%
        {\ctxapp{\ctxapp{\Gamma_1}{\alpha{:}(\Phi \atl{n} *)}}{\Gamma_0} \vde e_2 : (\beta{:}(\Psi \atl{m} *)) \arrow S}%
        {\mathcal{D}_1}%
        &%
        \deduce%
        {\ctxapp{\ctxrestrict{(\ctxapp{\ctxapp{\Gamma_1}{\alpha{:}(\Phi \atl{n} *)}}{\Gamma_0})}{m}}{\Psi} \vde S_2}%
        {\mathcal{D}_2}%
      }%
    \]
    \begin{itemize}
    \item Subcase: \(m > n\)\\
      \(\ctxapp{\ctxapp{\Gamma_1}{\alpha{:}(\Phi \atl{n} *)}}{\Gamma_0} \vde e_2 : (\beta{:}(\Psi \atl{m} *)) \arrow S\) \hfill by \(\mathcal{D}_1\)\\
      \(\ctxapp{\Gamma_1}{[(\Phihat^n.T)/\alpha^n]\Gamma_0} \vde [(\Phihat^n.T)/\alpha^n]e_2 : [(\Phihat^n.T)/\alpha^n]((\beta{:}(\Psi \atl{m} *)) \arrow S)\)\\
      \mbox{} \hfill by IH (type subst.~(\ref{itm:typesubst-term}))\\
      \(\ctxapp{\Gamma_1}{[(\Phihat^n.T)/\alpha^n]\Gamma_0} \vde [(\Phihat^n.T)/\alpha^n]e_2 : (\beta{:}(\Psi \atl{m} *)) \arrow [(\Phihat^n.T)/\alpha^n]S\)\\
      \mbox{} \hfill by the def.~of substitution\\
      \(\ctxapp{\ctxrestrict{(\ctxapp{\ctxapp{\Gamma_1}{\alpha{:}(\Phi \atl{n} *)}}{\Gamma_0})}{m}}{\Psi} \vde S_2\) \hfill by \(\mathcal{D}_2\)\\
      \(\ctxapp{\ctxrestrict{\Gamma_1}{m}}{\Psi} \vde S_2\) \hfill by context chop distrib.~\ref{lem:chopdistrib}\\
      \(\ctxapp{\ctxrestrict{(\ctxapp{\Gamma_1}{[(\Phihat^n.T)/\alpha^n]\Gamma_0})}{m}}{\Psi} \vde S_2\) \hfill by context chop distrib.~\ref{lem:chopdistrib}\\
      \(\ctxapp{\Gamma_1}{[(\Phihat^n.T)/\alpha^n]\Gamma_0} \vde [(\Phihat^n.T)/\alpha^n]e_2\;(\hat\Psi^m . S_2) : [(\hat\Psi^m.S_2)/\beta^m][(\Phihat^n.T)/\alpha^n]S\)\\
      \mbox{} \hfill by typing rule\\
      \(\ctxapp{\Gamma_1}{[(\Phihat^n.T)/\alpha^n]\Gamma_0} \vde [(\Phihat^n.T)/\alpha^n]e_2\;(\hat\Psi^m . S_2) : [(\Phihat^n.T)/\alpha^n][(\hat\Psi^m.S_2)/\beta^m]S\)\\
      \mbox{} \hfill \(\beta\) does not appear in \(\Phi\) or \(T\) and \(\alpha\) does not appear in \(\Psi\) and \(S_2\)\\
      \(\ctxapp{\Gamma_1}{[(\Phihat^n.T)/\alpha^n]\Gamma_0} \vde [(\Phihat^n.T)/\alpha^n](e_2\;(\hat\Psi^m . S_2)) : [(\Phihat^n.T)/\alpha^n][(\hat\Psi^m.S_2)/\beta^m]S\)\\
      \mbox{} \hfill by the def.~of substitution\\

    \item Subcase: \(m \le n\)\\
      \(\ctxapp{\ctxapp{\Gamma_1}{\alpha{:}(\Phi \atl{n} *)}}{\Gamma_0} \vde e_2 : (\beta{:}(\Psi \atl{m} *)) \arrow S\) \hfill by \(\mathcal{D}_1\)\\
      \(\ctxapp{\Gamma_1}{[(\Phihat^n.T)/\alpha^n]\Gamma_0} \vde [(\Phihat^n.T)/\alpha^n]e_2 : [(\Phihat^n.T)/\alpha^n]((\beta{:}(\Psi \atl{m} *)) \arrow S)\)\\
      \mbox{} \hfill by IH (type subst.~(\ref{itm:typesubst-term}))\\
      \(\ctxapp{\Gamma_1}{[(\Phihat^n.T)/\alpha^n]\Gamma_0} \vde [(\Phihat^n.T)/\alpha^n]e_2 : (\beta{:}([(\Phihat^n.T)/\alpha^n]\Psi \atl{m} *)) \arrow [(\Phihat^n.T)/\alpha^n]S\)\\
      \mbox{} \hfill by the def.~of substitution\\
      \(\ctxapp{\Gamma_1}{[(\Phihat^n.T)/\alpha^n]\Gamma_0} \vde [(\Phihat^n.T)/\alpha^n]e_2 : (\beta{:}(\Psi \atl{m} *)) \arrow [(\Phihat^n.T)/\alpha^n]S\)\\
      \mbox{} \hfill \(\Psi\) contains type variables only\\
      \(\ctxapp{\ctxrestrict{(\ctxapp{\ctxapp{\Gamma_1}{\alpha{:}(\Phi \atl{n} *)}}{\Gamma_0})}{m}}{\Psi} \vde S_2\) \hfill by \(\mathcal{D}_2\)\\
      \(\ctxapp{\ctxapp{\ctxapp{\Gamma_1}{\alpha{:}(\Phi \atl{n} *)}}{\ctxrestrict{\Gamma_0}{m}}}{\Psi} \vde S_2\) \hfill by context chop distrib.~\ref{lem:chopdistrib}\\
      \(\ctxapp{\Gamma_1}{[(\Phihat^n.T)/\alpha^n](\ctxapp{\ctxrestrict{\Gamma_0}{m}}{\Psi})} \vde [(\Phihat^n.T)/\alpha^n]S_2\) \hfill by IH (type subst.~(\ref{itm:typesubst-type}))\\
      \(\ctxapp{\ctxapp{\Gamma_1}{[(\Phihat^n.T)/\alpha^n]\ctxrestrict{\Gamma_0}{m}}}{[(\Phihat^n.T)/\alpha^n]\Psi} \vde [(\Phihat^n.T)/\alpha^n]S_2\) \hfill by context subst.~distrib.~\ref{lem:substdistrib}\\
      \(\ctxapp{\ctxapp{\Gamma_1}{[(\Phihat^n.T)/\alpha^n]\ctxrestrict{\Gamma_0}{m}}}{\Psi} \vde [(\Phihat^n.T)/\alpha^n]S_2\) \hfill \(\Psi\) contains type variables only\\
      \(\ctxapp{\ctxrestrict{(\ctxapp{\Gamma_1}{[(\Phihat^n.T)/\alpha^n]\Gamma_0})}{m}}{\Psi} \vde [(\Phihat^n.T)/\alpha^n]S_2\) \hfill by context chop distrib.~\ref{lem:chopdistrib}\\
      \(\ctxapp{\Gamma_1}{[(\Phihat^n.T)/\alpha^n]\Gamma_0}\)\\
      \mbox{} \qquad \(\vde [(\Phihat^n.T)/\alpha^n]e_2\;(\hat\Psi^m . [(\Phihat^n.T)/\alpha^n]S_2) : [(\hat\Psi^m.[(\Phihat^n.T)/\alpha^n]S_2)/\beta^m][(\Phihat^n.T)/\alpha^n]S\)\\
      \mbox{} \hfill by typing rule\\
      \(\ctxapp{\Gamma_1}{[(\Phihat^n.T)/\alpha^n]\Gamma_0} \vde [(\Phihat^n.T)/\alpha^n]e_2\;(\hat\Psi^m . [(\Phihat^n.T)/\alpha^n]S_2) : [(\Phihat^n.T)/\alpha^n][(\hat\Psi^m.S_2)/\beta^m]S\)\\
      \mbox{} \hfill by commuting substitutions~\ref{lem:commsubst}\\
      \(\ctxapp{\Gamma_1}{[(\Phihat^n.T)/\alpha^n]\Gamma_0} \vde [(\Phihat^n.T)/\alpha^n](e_2\;(\hat\Psi^m . S_2)) : [(\Phihat^n.T)/\alpha^n][(\hat\Psi^m.S_2)/\beta^m]S\)\\
      \mbox{} \hfill by the def.~of substitution\\
    \end{itemize}

  \item Case:
    \[
      \infer[m > 0]%
      {\ctxapp{\ctxapp{\Gamma_1}{\alpha{:}(\Phi \atl{n} *)}}{\Gamma_0} \vde \letboxm{(\hat\Psi^m.u)}{e_1}{e_2} : S_2}%
      {%
        \deduce%
        {\ctxapp{\ctxapp{\Gamma_1}{\alpha{:}(\Phi \atl{n} *)}}{\Gamma_0} \vde e_1 : \cbox{\Psi \atl{m} S_1}}%
        {\mathcal{D}_1}%
        &%
        \deduce%
        {\ctxinsert{(\ctxapp{\ctxapp{\Gamma_1}{\alpha{:}(\Phi \atl{n} *)}}{\Gamma_0})}{u{:}(\Psi \atl{m} S_1)} \vde e_2 : S_2}%
        {\mathcal{D}_2}%
      }%
    \]
    \begin{itemize}
    \item Subcase: \(m > n\)\\
      \(\ctxapp{\ctxapp{\Gamma_1}{\alpha{:}(\Phi \atl{n} *)}}{\Gamma_0} \vde e_1 : \cbox{\Psi \atl{m} S_1}\) \hfill by \(\mathcal{D}_1\)\\
      \(\ctxapp{\Gamma_1}{[(\Phihat^n.T)/\alpha^n]\Gamma_0} \vde [(\Phihat^n.T)/\alpha^n]e_1 : [(\Phihat^n.T)/\alpha^n]\cbox{\Psi \atl{m} S_1}\) \hfill by IH (type subst.~(\ref{itm:typesubst-term}))\\
      \(\ctxapp{\Gamma_1}{[(\Phihat^n.T)/\alpha^n]\Gamma_0} \vde [(\Phihat^n.T)/\alpha^n]e_1 : \cbox{\Psi \atl{m} S_1}\) \hfill by the def.~of substitution\\
      \(\ctxinsert{(\ctxapp{\ctxapp{\Gamma_1}{\alpha{:}(\Phi \atl{n} *)}}{\Gamma_0})}{u{:}(\Psi \atl{m} S_1)} \vde e_2 : S_2\) \hfill by \(\mathcal{D}_2\)\\
      \(\ctxapp{\ctxapp{(\ctxinsert{\Gamma_1}{u{:}(\Psi \atl{m} S_1)})}{\alpha{:}(\Phi \atl{n} *)}}{\Gamma_0} \vde e_2 : S_2\) \hfill by context merge reordering~\ref{lem:mergereorder} \(m > n\)\\
      \(\ctxapp{(\ctxinsert{\Gamma_1}{u{:}(\Psi \atl{m} S_1)})}{[(\Phihat^n.T)/\alpha^n]\Gamma_0} \vde [(\Phihat^n.T)/\alpha^n]e_2 : [(\Phihat^n.T)/\alpha^n]S_2\)\\
      \mbox{} \hfill by IH (type subst.~(\ref{itm:typesubst-term}))\\
      \(\ctxinsert{(\ctxapp{\Gamma_1}{[(\Phihat^n.T)/\alpha^n]\Gamma_0})}{u{:}(\Psi \atl{m} S_1)} \vde [(\Phihat^n.T)/\alpha^n]e_2 : [(\Phihat^n.T)/\alpha^n]S_2\)\\
      \mbox{} \hfill by context merge reordering~\ref{lem:mergereorder} \(m > n\)\\
      \(\ctxapp{\Gamma_1}{[(\Phihat^n.T)/\alpha^n]\Gamma_0}\)\\
      \mbox{} \qquad \(\vde \letboxm{(\hat\Psi^m.u)}{[(\Phihat^n.T)/\alpha^n]e_1}{[(\Phihat^n.T)/\alpha^n]e_2} : [(\Phihat^n.T)/\alpha^n]S_2\)\\
      \mbox{} \hfill by typing rule\\
      \(\ctxapp{\Gamma_1}{[(\Phihat^n.T)/\alpha^n]\Gamma_0} \vde [(\Phihat^n.T)/\alpha^n](\letboxm{(\hat\Psi^m.u)}{e_1}{e_2}) : [(\Phihat^n.T)/\alpha^n]S_2\)\\
      \mbox{} \hfill by the def.~of substitution\\

    \item Subcase: \(m \le n\)\\
      \(\ctxapp{\ctxapp{\Gamma_1}{\alpha{:}(\Phi \atl{n} *)}}{\Gamma_0} \vde e_1 : \cbox{\Psi \atl{m} S_1}\) \hfill by \(\mathcal{D}_1\)\\
      \(\ctxapp{\Gamma_1}{[(\Phihat^n.T)/\alpha^n]\Gamma_0} \vde [(\Phihat^n.T)/\alpha^n]e_1 : [(\Phihat^n.T)/\alpha^n]\cbox{\Psi \atl{m} S_1}\) \hfill by IH (type subst.~(\ref{itm:typesubst-term}))\\
      \(\ctxapp{\Gamma_1}{[(\Phihat^n.T)/\alpha^n]\Gamma_0} \vde [(\Phihat^n.T)/\alpha^n]e_1 : \cbox{[(\Phihat^n.T)/\alpha^n]\Psi \atl{m} [(\Phihat^n.T)/\alpha^n]S_1}\)\\
      \mbox{} \hfill by the def.~of substitution\\
      \(\ctxinsert{(\ctxapp{\ctxapp{\Gamma_1}{\alpha{:}(\Phi \atl{n} *)}}{\Gamma_0})}{u{:}(\Psi \atl{m} S_1)} \vde e_2 : S_2\) \hfill by \(\mathcal{D}_2\)\\
      \(\ctxapp{\ctxapp{\Gamma_1}{\alpha{:}(\Phi \atl{n} *)}}{(\ctxinsert{\Gamma_0}{u{:}(\Psi \atl{m} S_1)})} \vde e_2 : S_2\) \hfill by context merge reordering~\ref{lem:mergereorder} \(m \le n\)\\
      \(\ctxapp{\Gamma_1}{[(\Phihat^n.T)/\alpha^n](\ctxinsert{\Gamma_0}{u{:}(\Psi \atl{m} S_1)})} \vde [(\Phihat^n.T)/\alpha^n]e_2 : [(\Phihat^n.T)/\alpha^n]S_2\)\\
      \mbox{} \hfill by IH (type subst.~(\ref{itm:typesubst-term}))\\
      \(\ctxapp{\Gamma_1}{(\ctxinsert{[(\Phihat^n.T)/\alpha^n]\Gamma_0}{[(\Phihat^n.T)/\alpha^n](u{:}(\Psi \atl{m} S_1))})} \vde [(\Phihat^n.T)/\alpha^n]e_2 : [(\Phihat^n.T)/\alpha^n]S_2\)\\
      \mbox{} \hfill by context subst.~distrib.~\ref{lem:substdistrib}\\
      \(\ctxinsert{(\ctxapp{\Gamma_1}{[(\Phihat^n.T)/\alpha^n]\Gamma_0})}{[(\Phihat^n.T)/\alpha^n](u{:}(\Psi \atl{m} S_1))} \vde [(\Phihat^n.T)/\alpha^n]e_2 : [(\Phihat^n.T)/\alpha^n]S_2\)\\
      \mbox{} \hfill by context merge reordering~\ref{lem:mergereorder} \(m \le n\)\\
      \(\ctxinsert{(\ctxapp{\Gamma_1}{[(\Phihat^n.T)/\alpha^n]\Gamma_0})}{u{:}([(\Phihat^n.T)/\alpha^n]\Psi \atl{m} [(\Phihat^n.T)/\alpha^n]S_1)}\)\\
      \mbox{} \qquad \(\vde [(\Phihat^n.T)/\alpha^n]e_2 : [(\Phihat^n.T)/\alpha^n]S_2\) \hfill by the def.~of substitution\\
      \(\ctxapp{\Gamma_1}{[(\Phihat^n.T)/\alpha^n]\Gamma_0}\)\\
      \mbox{} \qquad \(\vde \letboxm{(\hat\Psi^m.u)}{[(\Phihat^n.T)/\alpha^n]e_1}{[(\Phihat^n.T)/\alpha^n]e_2} : [(\Phihat^n.T)/\alpha^n]S_2\)\\
      \mbox{} \hfill by typing rule\\
      \(\ctxapp{\Gamma_1}{[(\Phihat^n.T)/\alpha^n]\Gamma_0} \vde [(\Phihat^n.T)/\alpha^n](\letboxm{(\hat\Psi^m.u)}{e_1}{e_2}) : [(\Phihat^n.T)/\alpha^n]S_2\)\\
      \mbox{} \hfill by the def.~of substitution\\
    \end{itemize}
  \end{itemize}

  \paragraph{Cases for simultaneous substitution~(\ref{itm:simsubst-type}):}
  \begin{itemize}
  \item Case:
    \[
      \infer%
      {\ctxapp{\Gamma_1}{\Gamma_0} \vde \alpha[\sigma_2]}%
      {%
        \deduce%
        {(\ctxapp{\Gamma_1}{\Gamma_0})(\alpha) = (\Psi \atl{m} *)}%
        {\mathcal{D}_1}%
        &%
        \deduce%
        {\ctxapp{\Gamma_1}{\Gamma_0} \vde \sigma_2 : \Psi}%
        {\mathcal{D}_2}%
      }%
    \]
    \begin{itemize}
    \item Subcase: \(\alpha \in \Gamma_1\)\\
      \((\ctxapp{\Gamma_1}{\Gamma_0})(\alpha) = (\Psi \atl{m} *)\) \hfill by \(\mathcal{D}_1\)\\
      \((\ctxapp{\Gamma_1}{\Phi})(\alpha) = (\Psi \atl{m} *)\) \hfill by \(\alpha \in \Gamma_1\)\\
      \(\ctxapp{\Gamma_1}{\Gamma_0} \vde \sigma_2 : \Psi\) \hfill by \(\mathcal{D}_2\)\\
      \(\ctxapp{\Gamma_1}{\Phi} \vde [\sigma/\hat\Gamma_0]\sigma_2 : [\sigma/\hat\Gamma_0]\Psi\) \hfill by IH (sim.~subst.~(\ref{itm:simsubst-subst}))\\
      \(\ctxapp{\Gamma_1}{\Phi} \vde [\sigma/\hat\Gamma_0]\sigma_2 : \Psi\) \hfill \(\Psi\) contains type variables only\\
      \(\ctxapp{\Gamma_1}{\Phi} \vde \alpha[[\sigma/\hat\Gamma_0]\sigma_2]\) \hfill by kinding rule\\
      \(\ctxapp{\Gamma_1}{\Phi} \vde [\sigma/\hat\Gamma_0](\alpha[\sigma_2])\) \hfill by the def.~of sim.~subst.\\

    \item Subcase: \(\alpha \in \Gamma_0\) and \(\mbox{lkp}(\sigma/\Gamma_0)\;\alpha = (\hat\Psi^m.T)\)\\
      \(\ctxapp{\Gamma_1}{\Phi} \vde \sigma : \Gamma_0\) \hfill by assumption\\
      \(\ctxapp{\ctxrestrict{(\ctxapp{\Gamma_1}{\Phi})}{m}}{\Psi} \vde T\) \hfill by inversion on \(\ctxapp{\Gamma_1}{\Phi} \vde \sigma : \Gamma_0\)\\
      \(\ctxapp{\Gamma_1}{\Gamma_0} \vde \sigma_2 : \Psi\) \hfill by \(\mathcal{D}_2\)\\
      \(\ctxapp{\Gamma_1}{\Phi} \vde [\sigma/\hat\Gamma_0]\sigma_2 : [\sigma/\hat\Gamma_0]\Psi\) \hfill by IH (sim.~subst.~(\ref{itm:simsubst-subst}))\\
      \(\ctxapp{\Gamma_1}{\Phi} \vde [\sigma/\hat\Gamma_0]\sigma_2 : \Psi\) \hfill \(\Psi\) contains type variables only\\
      \(\ctxapp{\ctxrestrict{(\ctxapp{\Gamma_1}{\Phi})}{m}}{\ctxignore{(\ctxapp{\Gamma_1}{\Phi})}{m}} \vde [\sigma/\hat\Gamma_0]\sigma_2 : \Psi\) \hfill by property of context~\ref{lem:contextprop}\\
      \(\ctxapp{\ctxrestrict{(\ctxapp{\Gamma_1}{\Phi})}{m}}{\ctxignore{\Phi}{m}} \vde [\sigma/\hat\Gamma_0]\sigma_2 : \Psi\) \hfill by context chop distrib.~\ref{lem:chopdistrib}\\
      \(\ctxapp{\ctxrestrict{(\ctxapp{\Gamma_1}{\Phi})}{m}}{\ctxignore{\Phi}{m}} \vde [[\sigma/\hat\Gamma_0]\sigma_2/\hat\Psi^m]T\)\\
      \mbox{} \hfill by IH (sim.~subst.~(\ref{itm:simsubst-type}))\\
      \(\ctxapp{\ctxapp{\Gamma_1}{\ctxrestrict{\Phi}{m}}}{\ctxignore{\Phi}{m}} \vde [[\sigma/\hat\Gamma_0]\sigma_2/\hat\Psi^m]T\) \hfill by context chop distrib.~\ref{lem:chopdistrib}\\
      \(\ctxapp{\Gamma_1}{\Phi} \vde [[\sigma/\hat\Gamma_0]\sigma_2/\hat\Psi^m]T\) \hfill by \hfill by property of context~\ref{lem:contextprop}\\
      \(\ctxapp{\Gamma_1}{\Phi} \vde [\sigma/\hat\Gamma_0](\alpha[\sigma_2])\) \hfill by the def.~of sim.~subst.\\

    \item Subcase: \(\alpha \in \Gamma_0\) and \(\mbox{lkp}(\sigma/\Gamma_0)\;\alpha = \beta\)\\
      \((\ctxapp{\Gamma_1}{\Gamma_0})(\alpha) = (\Psi \atl{m} *)\) \hfill by \(\mathcal{D}_1\)\\
      \(\ctxapp{\Gamma_1}{\Phi} \vdash \sigma : \Gamma_0\) \hfill by assumption\\
      \((\ctxapp{\Gamma_1}{\Phi})(\beta) = ([\sigma/\hat\Gamma_0]\Psi \atl{m} *)\) \hfill by inversion on \(\ctxapp{\Gamma_1}{\Phi} \vdash \sigma : \Gamma_0\)\\
      \((\ctxapp{\Gamma_1}{\Phi})(\beta) = (\Psi \atl{m} *)\) \hfill \(\Psi\) contains type variables only\\
      \(\ctxapp{\Gamma_1}{\Gamma_0} \vde \sigma_2 : \Psi\) \hfill by \(\mathcal{D}_2\)\\
      \(\ctxapp{\Gamma_1}{\Phi} \vde [\sigma/\hat\Gamma_0]\sigma_2 : [\sigma/\hat\Gamma_0]\Psi\) \hfill by IH (sim.~subst.~(\ref{itm:simsubst-subst}))\\
      \(\ctxapp{\Gamma_1}{\Phi} \vde [\sigma/\hat\Gamma_0]\sigma_2 : \Psi\) \hfill \(\Psi\) contains type variables only\\
      \(\ctxapp{\Gamma_1}{\Phi} \vde \beta[[\sigma/\hat\Gamma_0]\sigma_2]\) \hfill by kinding rule\\
      \(\ctxapp{\Gamma_1}{\Phi} \vde [\sigma/\hat\Gamma_0](\alpha[\sigma_2])\) \hfill by the def.~of sim.~subst.\\
    \end{itemize}

  \item Case:
    \[
      \infer%
      {\ctxapp{\Gamma_1}{\Gamma_0} \vde (\alpha{:}(\Psi \atl{m} *)) \arrow S}%
      {%
        \deduce%
        {\vde \ctxapp{\ctxrestrict{(\ctxapp{\Gamma_1}{\Gamma_0})}{m}}{\Psi}}%
        {\mathcal{D}_1}%
        &%
        \deduce%
        {\ctxinsert{(\ctxapp{\Gamma_1}{\Gamma_0})}{\alpha{:}(\Psi \atl{m} *)} \vde S}%
        {\mathcal{D}_2}%
      }%
    \]
    Let \(n\) be the minimum level of variables in \(\Gamma_1\) and \(\sigma'/\Gamma'_0 = \ctxrestrict{(\sigma/\hat\Gamma_0)}{m}\)
    \begin{itemize}
    \item Subcase: \(m \ge n\)\\
      \(\vde \ctxapp{\ctxrestrict{(\ctxapp{\Gamma_1}{\Gamma_0})}{m}}{\Psi}\) \hfill by \(\mathcal{D}_1\)\\
      \(\vde \ctxapp{\ctxrestrict{\Gamma_1}{m}}{\Psi}\) \hfill by context chop distrib.~\ref{lem:chopdistrib}\\
      \(\vde \ctxapp{\ctxrestrict{(\ctxapp{\Gamma_1}{\Psi})}{m}}{\Psi}\) \hfill by context chop distrib.~\ref{lem:chopdistrib}\\
      \(\ctxinsert{(\ctxapp{\Gamma_1}{\Gamma_0})}{\alpha{:}(\Psi \atl{m} *)} \vde S\) \hfill by \(\mathcal{D}_2\)\\
      \(\ctxapp{(\ctxinsert{\Gamma_1}{\alpha{:}(\Psi \atl{m} *)})}{\Gamma_0} \vde S\) \hfill by context merge reordering~\ref{lem:mergereorder}\\
      \(\ctxapp{(\ctxinsert{\Gamma_1}{\alpha{:}(\Psi \atl{m} *)})}{\Phi} \vde [\sigma/\hat\Gamma_0]S\) \hfill by IH (sim.~subst.~(\ref{itm:simsubst-type}))\\
      \(\ctxinsert{(\ctxapp{\Gamma_1}{\Phi})}{\alpha{:}(\Psi \atl{m} *)} \vde [\sigma/\hat\Gamma_0]S\) \hfill by context merge reordering~\ref{lem:mergereorder}\\
      \(\ctxapp{\Gamma_1}{\Phi} \vde (\alpha{:}(\Psi \atl{m} *)) \arrow [\sigma/\hat\Gamma_0]S\) \hfill by typing rule\\
      \(\ctxapp{\Gamma_1}{\Phi} \vde [\sigma/\hat\Gamma_0]((\alpha{:}(\Psi \atl{m} *)) \arrow S)\) \hfill by the def.~of~sim.~subst.\\

    \item Subcase: \(m < n\)\\
      \(\ctxapp{\Gamma_1}{\Phi} \vde \sigma : \Gamma_0\) \hfill by assumption\\
      \(\ctxrestrict{(\ctxapp{\Gamma_1}{\Phi})}{m} \vde \sigma' : \Gamma'_0\) \hfill by substitution chop~\ref{lem:substchop}\\
      \(\ctxapp{\Gamma_1}{\ctxrestrict{\Phi}{m}} \vde \sigma' : \Gamma'_0\) \hfill by context chop distrib.~\ref{lem:chopdistrib}\\
      \(\vde \ctxapp{\ctxapp{\Gamma_1}{\ctxrestrict{\Phi}{m}}}{[\sigma'/\hat\Gamma'_0]\Psi}\) \hfill by IH (sim.~subst.~(\ref{itm:simsubst-context}))\\
      \(\vde \ctxapp{\ctxrestrict{(\ctxapp{\Gamma_1}{\Phi})}{m}}{[\sigma'/\hat\Gamma'_0]\Psi}\) \hfill by context chop distrib.~\ref{lem:chopdistrib}\\
      \(\vde \ctxapp{\ctxrestrict{(\ctxapp{\Gamma_1}{\Phi})}{m}}{\Psi}\) \hfill \(\Psi\) contains type variables only\\
      \(\ctxinsert{(\ctxapp{\Gamma_1}{\Phi})}{\alpha{:}(\Psi \atl{m} *)} \vde \ctxinsert{\sigma}{\alpha^m} : \ctxinsert{\Gamma_0}{\alpha{:}(\Psi \atl{m} *)}\) \hfill by substitution extension~\ref{lem:substext}\\
      \(\ctxapp{\Gamma_1}{(\ctxinsert{\Phi}{\alpha{:}(\Psi \atl{m} *)})} \vde \ctxinsert{\sigma}{\alpha^m} : \ctxinsert{\Gamma}{\alpha{:}(\Psi \atl{m} *)}\) \hfill by context merge reordering~\ref{lem:mergereorder}\\
      \(\ctxapp{\Gamma_1}{(\ctxinsert{\Phi}{\alpha{:}(\Psi \atl{m} *)})} \vde [\ctxinsertatl{\sigma/\hat\Gamma_0}{\alpha/\alpha}{m}]S\) \hfill by IH (sim.~subst.~(\ref{itm:simsubst-type}))\\
      \(\ctxinsert{(\ctxapp{\Gamma_1}{\Phi})}{\alpha{:}(\Psi \atl{m} *)} \vde [\ctxinsertatl{\sigma/\hat\Gamma_0}{\alpha/\alpha}{m}]S\) \hfill by context merge reordering~\ref{lem:mergereorder}\\
      \(\ctxapp{\Gamma_1}{\Phi} \vde (\alpha{:}(\Psi \atl{m} *)) \arrow [\ctxinsertatl{\sigma/\hat\Gamma_0}{\alpha/\alpha}{m}]S\) \hfill by typing rule\\
      \(\ctxapp{\Gamma_1}{\Phi} \vde [\sigma/\hat\Gamma_0]((\alpha{:}(\Psi \atl{m} *)) \arrow S)\) \hfill by the def.~of~sim.~subst.\\
    \end{itemize}
  \end{itemize}

  \paragraph{Cases for simultaneous substitution~(\ref{itm:simsubst-term}):}
  \begin{itemize}
  \item Case:
    \[
      \infer[m > 0]%
      {\ctxapp{\Gamma_1}{\Gamma_0} \vde \letboxm{(\hat\Psi^m.u)}{e_1}{e_2} : S_2}%
      {%
        \deduce%
        {\ctxapp{\Gamma_1}{\Gamma_0} \vde e_1 : \cbox{\Psi \atl{m} S_1}}%
        {\mathcal{D}_1}%
        &%
        \deduce%
        {\ctxinsert{(\ctxapp{\Gamma_1}{\Gamma_0})}{u{:}(\Psi \atl{m} S_1)} \vde e_2 : S_2}%
        {\mathcal{D}_2}%
      }%
    \]
    Let \(n\) be the minimum level of variables in \(\Gamma_1\)
    \begin{itemize}
    \item Subcase: \(m \ge n\)\\
      \(\ctxapp{\Gamma_1}{\Gamma_0} \vde e_1 : \cbox{\Psi \atl{m} S_1}\) \hfill by \(\mathcal{D}_1\)\\
      \(\ctxapp{\Gamma_1}{\Phi} \vde [(\sigma/\hat\Gamma_0)]e_1 : [(\sigma/\hat\Gamma_0)]\cbox{\Psi \atl{m} S_1}\) \hfill by IH\\
      \(\ctxapp{\Gamma_1}{\Phi} \vde [(\sigma/\hat\Gamma_0)]e_1 : \cbox{\Psi \atl{m} S_1}\) \hfill by the def.~of sim.~subst.\\
      \(\ctxinsert{(\ctxapp{\Gamma_1}{\Gamma_0})}{u{:}(\Psi \atl{m} S_1)} \vde e_2 : S_2\) \hfill by \(\mathcal{D}_2\)\\
      \(\ctxapp{(\ctxinsert{\Gamma_1}{u{:}(\Psi \atl{m} S_1)})}{\Gamma_0} \vde e_2 : S_2\) \hfill by context merge reordering~\ref{lem:mergereorder} \(m \ge n\)\\
      \(\ctxapp{(\ctxinsert{\Gamma_1}{u{:}(\Psi \atl{m} S_1)})}{\Phi} \vde [(\sigma/\hat\Gamma_0)]e_2 : [\sigma/\hat\Gamma_0]S_2\) \hfill by IH\\
      \(\ctxinsert{(\ctxapp{\Gamma_1}{\Phi})}{u{:}(\Psi \atl{m} S_1)} \vde [(\sigma/\hat\Gamma_0)]e_2 : [\sigma/\hat\Gamma_0]S_2\) \hfill by context merge reordering~\ref{lem:mergereorder} \(m \ge n\)\\
      \(\ctxapp{\Gamma_1}{\Phi} \vde \letboxm{(\hat\Psi^m.u)}{[\sigma/\hat\Gamma_0]e_1}{[(\sigma/\hat\Gamma_0)]e_2} : [\sigma/\hat\Gamma_0]S_2\) \hfill by typing rule\\
      \(\ctxapp{\Gamma_1}{\Phi} \vde [\sigma/\hat\Gamma_0](\letboxm{(\hat\Psi^m.u)}{e_1}{e_2}) : [\sigma/\hat\Gamma_0]S_2\) \hfill by the def.~of sim.~subst.\\

    \item Subcase: \(m < n\)\\
      Let \((\sigma'/\hat\Gamma'_0) = \ctxrestrict{(\sigma/\hat\Gamma_0)}{m}\)\\
      \(\ctxapp{\Gamma_1}{\Gamma_0} \vde e_1 : \cbox{\Psi \atl{m} S_1}\) \hfill by \(\mathcal{D}_1\)\\
      \(\ctxapp{\Gamma_1}{\Phi} \vde [\sigma/\hat\Gamma_0]e_1 : [\sigma/\hat\Gamma_0]\cbox{\Psi \atl{m} S_1}\) \hfill by IH\\
      \(\ctxapp{\Gamma_1}{\Phi} \vde [\sigma/\hat\Gamma_0]e_1 : \cbox{[\sigma'/\hat\Gamma'_0]\Psi \atl{m} [\ctxmerge{(\sigma'/\hat\Gamma'_0)}{(\id{(\hat\Psi^m)}/\hat\Psi)}]S_1}\) \hfill by the def.~of sim.~subst.\\
      \(\ctxapp{\Gamma_1}{\Phi} \vde \sigma : \Gamma_0\) \hfill by assumption\\
      \(\ctxinsert{(\ctxapp{\Gamma_1}{\Phi})}{u{:}([\sigma'/\hat\Gamma'_0]\Psi \atl{m} [\ctxmerge{(\sigma'/\hat\Gamma'_0)}{(\id{(\hat\Psi^m)}/\hat\Psi)}]S_1)}\)\\
      \mbox{} \qquad \(\vde (\ctxinsertatl{\sigma}{u}{m}) : (\ctxinsert{\Gamma_0}{u{:}(\Psi \atl{m} S_1)})\) \hfill by substitution extension~\ref{lem:substext}\\
      \(\ctxapp{\Gamma_1}{(\ctxinsert{\Phi}{u{:}([\sigma'/\hat\Gamma'_0]\Psi \atl{m} [\ctxmerge{(\sigma'/\hat\Gamma'_0)}{(\id{(\hat\Psi^m)}/\hat\Psi)}]S_1)})}\)\\
      \mbox{} \qquad \(\vde (\ctxinsertatl{\sigma}{u}{m}) : (\ctxinsert{\Gamma_0}{u{:}(\Psi \atl{m} S_1)})\) \hfill by context merge reordering~\ref{lem:mergereorder} \(m < n\)\\
      \(\ctxinsert{(\ctxapp{\Gamma_1}{\Gamma_0})}{u{:}(\Psi \atl{m} S_1)} \vde e_2 : S_2\) \hfill by \(\mathcal{D}_2\)\\
      \(\ctxapp{\Gamma_1}{(\ctxinsert{\Gamma_0}{u{:}(\Psi \atl{m} S_1)})} \vde e_2 : S_2\) \hfill by context merge reordering~\ref{lem:mergereorder} \(m < n\)\\
      \(\ctxapp{\Gamma_1}{(\ctxinsert{\Phi}{u{:}([\sigma'/\hat\Gamma'_0]\Psi \atl{m} [\ctxmerge{(\sigma'/\hat\Gamma'_0)}{(\id{(\hat\Psi^m)}/\hat\Psi)}]S_1)})}\)\\
      \mbox{} \qquad \(\vde [\ctxinsertatl{(\sigma/\hat\Gamma_0)}{u/u}{m}]e_2 : [\ctxinsertatl{(\sigma/\hat\Gamma_0)}{u/u}{m}]S_2\) \hfill by IH\\
      \(\ctxapp{\Gamma_1}{(\ctxinsert{\Phi}{u{:}([\sigma'/\hat\Gamma'_0]\Psi \atl{m} [\ctxmerge{(\sigma'/\hat\Gamma'_0)}{(\id{(\hat\Psi^m)}/\hat\Psi)}]S_1)})}\)\\
      \mbox{} \qquad \(\vde [\ctxinsertatl{(\sigma/\hat\Gamma_0)}{u/u}{m}]e_2 : [\sigma/\hat\Gamma_0]S_2\) \hfill type ignores term substitution~\ref{lem:typeignoresubst}\\
      \(\ctxinsert{(\ctxapp{\Gamma_1}{\Phi})}{u{:}([\sigma'/\hat\Gamma'_0]\Psi \atl{m} [\ctxmerge{(\sigma'/\hat\Gamma'_0)}{(\id{(\hat\Psi^m)}/\hat\Psi)}]S_1)}\)\\
      \mbox{} \qquad \(\vde [\ctxinsertatl{(\sigma/\hat\Gamma_0)}{u/u}{m}]e_2 : [\sigma/\hat\Gamma_0]S_2\) \hfill by context merge reordering~\ref{lem:mergereorder} \(m < n\)\\
      \(\ctxapp{\Gamma_1}{\Phi} \vde \letboxm{(\hat\Psi^m.u)}{[\sigma/\hat\Gamma_0]e_1}{[\ctxinsertatl{(\sigma/\hat\Gamma_0)}{u/u}{m}]e_2} : [\sigma/\hat\Gamma_0]S_2\) \hfill by typing rule\\
      \(\ctxapp{\Gamma_1}{\Phi} \vde [\sigma/\hat\Gamma_0](\letboxm{(\hat\Psi^m.u)}{e_1}{e_2}) : [\sigma/\hat\Gamma_0]S_2\) \hfill by the def.~of sim.~subst.\\
    \end{itemize}
  \end{itemize}
\end{proof}

\section{Structural Equality Modulo Constraints}\label{appendix:typeq}

We summarize here (Fig~\ref{fig:typeqmod}, Fig~\ref{fig:ctxeqmod}, and
Fig~\ref{fig:substeqmod}) the structural equality rules modulo
constraints. Most rules are straightforward and as expected. The only
interesting case is : $\alpha[\sigma] = T$ where $\alpha{:=}(\Phihat^n.S) : (\Phi \atl{n} *) \in \Psi $.
In this case, we continue to compare $T$ with $[\sigma/\Phihat]S$.

\begin{figure}[ht]
  \centering
  \caption{Structural equality for types modulo constraints: \fbox{$\Psi \vde T = S$} where $\Psi \vde T$ and $\Psi \vde S$}
\[
  \begin{array}{c}
\infer{\Psi \vde \alpha[\sigma] = \alpha[\sigma']}{
\alpha:(\Phi \atl{n} *) \in \Psi &
\Psi \vde \sigma = \sigma' : \Phi
}
\quad
\infer{\Psi \vde \alpha[\sigma] = T}
{\alpha{:=}(\Phihat^n.S){:}(\Phi \atl{n} *) \in \Psi &
\Psi \vde [\sigma/\Phihat]S = T}
\\[1em]
\infer{\Psi \vde T = \alpha[\sigma]}
{\alpha{:=}(\Phihat^n.S){:}(\Phi \atl{n} *) \in \Psi &
\Psi \vde T = [\sigma/\Phihat]S}
\quad
\infer{\Psi \vde T = S}
{\# \in \Psi}
\\[1em]
\infer{\Psi \vde T_1 \arrow T_2 = S_1 \arrow S_2}{
\Psi \vde T_1 = S_1 &
\Psi \vde T_2 = S_2
}
\quad
\infer{\Psi \vde
(\alpha{:}(\Phi\atl{n} *)) \arrow T =
(\alpha{:}(\Phi'\atl{n} *)) \arrow S}
{\ctxrestrict{\Psi}{n} \vde \Phi = \Phi' &
 \ctxinsert{\Psi}{\alpha{:}(\Phi\atl{n} *)}
\vde T = S
}
\\[1em]
\infer{\Psi \vde
\cbox{\Phi\atl{n} T} = \cbox{\Phi' \atl{n} T'}}
{\ctxrestrict{\Psi }{n} \vde \Phi = \Phi' &
 \ctxappend{\ctxrestrict{\Psi}{n}}{\Phi} \vde T = T'}
  \end{array}
\]
  \label{fig:typeqmod}
\end{figure}

\begin{figure}[ht]
  \centering
  \caption{Structural equality for variable
    contexts modulo constraints: \fbox{$\Psi \vde \Phi_1 = \Phi_2$}~~
    where $\m{level}(\Psi) \ge \m{level}(\Phi_i)$ and $\Phi_i$ are pure contexts, i.e. they contain only variable declarations, but no constraints themselves}
\[
  \begin{array}{c}
\infer{\Psi \vde (\Phi_1, \alpha{:}(\Phi_1' \atl{n} *))
               = (\Phi_2, \alpha{:}(\Phi_2' \atl{n} *))}
{\Psi \vde \ctxapp{\ctxrestrict{\Phi_1}n}{\Phi_1'} =
      \ctxapp{\ctxrestrict{\Phi_2}n}{\Phi'_2} &
\Psi \vde \Phi_1 = \Phi_2}
\quad
\infer{\Psi \vde \Phi_1 = \Phi_2}{\# \in \Psi}
\quad
\infer{\Psi \vde \cdot = \cdot}{ }
\\[1em]
\infer{\Psi \vde (\Phi_1, x{:}(\Phi_1' \atl{n} T))
               = (\Phi_2, x{:}(\Phi_2' \atl{n} S))}
{\Psi \vde \ctxapp{\ctxrestrict{\Phi_1}n}{\Phi_1'} =
      \ctxapp{\ctxrestrict{\Phi_2}n}{\Phi'_2} &
 \ctxapp{\Psi}{\ctxapp{\ctxrestrict{\Phi_1}n}{\Phi_1'}} \vde T = S &
\Psi \vde \Phi_1 = \Phi_2}
  \end{array}
\]
  \label{fig:ctxeqmod}
\end{figure}

\begin{figure}[ht]
  \centering
  \caption{Structural equality for substitution modulo constraints: \fbox{$\Psi \vde \sigma_1 = \sigma_2 : \Phi$}~~ where \(\Phi\) is a pure context}
\[
  \begin{array}{c}
\\[1em]
\infer%
{\Psi \vde \cdot = \cdot : \cdot}%
{}%
\quad%
\infer%
{\Psi \vde \sigma_1 = \sigma_2 : \Phi}%
{\# \in \Psi}%
\quad%
\infer%
{\Psi \vde \sigma_1;\beta^k = \sigma_2;\beta^k : \ctxapp{\Phi}{\alpha{:}(\Gamma \atl{k} *)}}%
{\Psi \vde \sigma_1 = \sigma_2 : \Phi &%
 \beta{:}(\Gamma \atl{k} *) \in \Psi%
}%
\\[1em]
\infer%
{\Psi \vde \sigma_1;\beta^k = \ctxapp{\sigma_2}{(\hat\Gamma^k . T)} : \ctxapp{\Phi}{\alpha{:}(\Gamma \atl{k} *)}}%
{\Psi \vde \sigma_1 = \sigma_2 : \Phi &%
 \ctxapp{\ctxrestrict{\Psi}{k}}{\Gamma} \vde \beta[\id(\hat\Gamma)] = T%
}%
\\[1em]
\infer%
{\Psi \vde \ctxapp{\sigma_1}{(\hat\Gamma^k . T)} = \sigma_2;\beta^k : \ctxapp{\Phi}{\alpha{:}(\Gamma \atl{k} *)}}%
{\Psi \vde \sigma_1 = \sigma_2 : \Phi &%
 \ctxapp{\ctxrestrict{\Psi}{k}}{\Gamma} \vde T = \beta[\id(\hat\Gamma)]%
}%
\\[1em]
\infer%
{\Psi \vde \sigma_1;\beta^k = \ctxapp{\sigma_2}{(\hat\Gamma^k . T)} : \ctxapp{\Phi}{\alpha{:}(\Gamma \atl{k} *)}}%
{\Psi \vde \sigma_1 = \sigma_2 : \Phi &%
 \beta{:=}(\hat\Gamma^k . S){:}(\Gamma \atl{k} *) \in \Psi &%
 \ctxapp{\ctxrestrict{\Psi}{k}}{\Gamma} \vde S = T%
}%
\\[1em]
\infer%
{\Psi \vde \ctxapp{\sigma_1}{(\hat\Gamma^k . T)} = \sigma_2;\beta^k : \ctxapp{\Phi}{\alpha{:}(\Gamma \atl{k} *)}}%
{\Psi \vde \sigma_1 = \sigma_2 : \Phi &%
 \beta{:=}(\hat\Gamma^k . S){:}(\Gamma \atl{k} *) \in \Psi &%
 \ctxapp{\ctxrestrict{\Psi}{k}}{\Gamma} \vde T = S%
}%
\\[1em]
\infer%
{\Psi \vde \sigma_1;\beta^k = \ctxapp{\sigma_2}{(\hat\Gamma^k . T)} : \ctxapp{\Phi}{\alpha{:}(\Gamma \atl{k} *)}}%
{\Psi \vde \sigma_1 = \sigma_2 : \Phi &%
 \beta{:=}(\hat\Gamma^k . S){:}(\Gamma \atl{k} *) \in \Psi &%
 \ctxapp{\ctxrestrict{\Psi}{k}}{\Gamma} \vde S = T%
}%
\\[1em]
\infer%
{\Psi \vde \ctxapp{\sigma_1}{(\hat\Gamma^k . T_1)} = \ctxapp{\sigma_2}{(\hat\Gamma^k . T_2)} : \ctxapp{\Phi}{\alpha{:}(\Gamma \atl{k} *)}}%
{\Psi \vde \sigma_1 = \sigma_2 : \Phi &%
 \ctxapp{\ctxrestrict{\Psi}{k}}{\Gamma} \vde T_1 = T_2%
}%
  \end{array}
\]
  \label{fig:substeqmod}
\end{figure}

\section{Unification of Substitution}\label{appendix:substunify}
We define unification for substitution in this section. The definition in
Fig~\ref{fig:substunify} contains only substitution for type variables, as
we used unification only on types. Symmetric cases are omitted for brevity.

\begin{figure}
  \caption{Substitution Unification \fbox{\(\Gamma; \Phi \vde \sigma_1 = \sigma_2 : \Psi \searrow \Gamma'\)}}
  \centering
\[
  \begin{array}{c}
\infer%
{\Gamma; \Phi \vde \cdot = \cdot : \cdot \searrow \Gamma}%
{}%
\quad%
\infer%
{\Gamma; \Phi \vde \sigma_1;\beta^k = \sigma_2;\beta^k : \ctxapp{\Psi}{\alpha{:}(\Psi' \atl{k} *)} \searrow \Gamma'}%
{\Gamma; \Phi \vde \sigma_1 = \sigma_2 : \Psi \searrow \Gamma'%
}%
\\[1em]
\infer%
{\Gamma; \Phi \vde \sigma_1;\beta^k = \ctxapp{\sigma_2}{(\hat\Psi'^k . T)} : \ctxapp{\Psi}{\alpha{:}(\Psi' \atl{k} *)} \searrow \Gamma''}%
{\Gamma; \Phi \vde \sigma_1 = \sigma_2 : \Psi \searrow \Gamma'%
 \Gamma; \ctxapp{\ctxrestrict{\Phi}{k}}{\Psi'} \vde \beta[id(\hat\Psi')] = T \searrow \Gamma'%
}%
\\[1em]
\infer%
{\Gamma; \Phi \vde \ctxapp{\sigma_1}{(\hat\Psi'^k . S)} = \ctxapp{\sigma_2}{(\hat\Psi'^k . T)} : \ctxapp{\Psi}{\alpha{:}(\Psi' \atl{k} *)} \searrow \Gamma''}%
{\Gamma; \Phi \vde \sigma_1 = \sigma_2 : \Psi \searrow \Gamma' &%
 \Gamma'; \ctxapp{\ctxrestrict{\Phi}{k}}{\Psi'} \vde S = T \searrow \Gamma''
}%
  \end{array}
\]
  \label{fig:substunify}
\end{figure}

\section{Proof of Lemma~\ref{lem:welldefunify}~and~\ref{lem:unifysound}}\label{sec:proofunify}

\begin{thmn}[\ref{lem:welldefunify}]\mbox{\textnormal{\textsc{(Well-defined Unification).}}}\quad
  \begin{enumerate}
  \item
    If
    $\mathcal{D}:\;\vde \ctxapp{\Gamma}{\ctxapp{\Gamma_0}{\Psi}}$ and
    $\mathcal{E}:\;\vde \ctxapp{\Gamma}{\ctxapp{\Gamma_0}{\Phi}}$\\
    then
    there is a $\Gamma'$ s.t. $\Gamma ; \Gamma_0 \vde \Psi = \Phi \searrow \Gamma'$.
  \item
    If
    $\vde \ctxapp{\Gamma}{\Psi}$ and
    $\mathcal{D}:\;\ctxapp{\Gamma}{\Psi} \vde T$ and
    $\mathcal{E}:\;\ctxapp{\Gamma}{\Psi} \vde S$\\
    then
    there is a $\Gamma'$ s.t. $\Gamma ; \Psi \vde T = S \searrow \Gamma'$
  \end{enumerate}
\end{thmn}

\begin{proof}
  By lexicographic induction on the structures of \(\mathcal{D}\) and \(\mathcal{E}\). Either \(\mathcal{D}\) decreases or \(\mathcal{D}\) remains the same and \(\mathcal{E}\) decreases. We show here the unification variable case only.

  \begin{itemize}
  \item Case: \(\Gamma = \ctxapp{\Gamma_1}{\ctxapp{\alpha{:}(\Psi \atl{k} *)}{\Gamma_0}}\) and
    \[
      \mathcal{D} =
      \begin{array}{c}
        \infer%
        {\ctxapp{\Gamma}{\Psi} \vde \alpha[\id(\hat\Psi)]}%
        {%
          \deduce%
          {\alpha{:}(\Psi \atl{k} *) \in (\ctxapp{\Gamma}{\Psi})}%
          {\mathcal{D}_1}%
          &%
          \deduce%
          {\ctxapp{\Gamma}{\Psi} \vde \id(\hat\Psi) : \Psi}%
          {\mathcal{D}_2}%
        }%
      \end{array}
    \]
    \begin{itemize}
    \item Subcase: \(\ctxapp{\Gamma_1}{\Psi} \vde S\)\\
      \(\Gamma ; \Psi \vde \alpha \not\in S\) \hfill as \(\alpha \not\in \ctxapp{\Gamma_1}{\Psi}\)\\
      \(\Gamma ; \Psi \vde \alpha[\id(\hat\Psi)] = S \searrow \ctxapp{\Gamma_1}{\ctxapp{\alpha{:=}(\hat\Psi . S){:}(\Psi \atl{k} *)}{\Gamma_0}}\) \hfill by unification rule\\
    \item Subcase:
      \[
        \mathcal{E} =
        \begin{array}{c}
          \infer%
          {\ctxapp{\Gamma}{\Psi} \vde \alpha[\id(\hat\Psi)]}%
          {%
            \deduce%
            {\alpha{:}(\Psi \atl{k} *) \in (\ctxapp{\Gamma}{\Psi})}%
            {\mathcal{D}_1}%
            &%
            \deduce%
            {\ctxapp{\Gamma}{\Psi} \vde \id(\hat\Psi) : \Psi}%
            {\mathcal{D}_2}%
          }%
        \end{array}
      \]
      \(\Gamma ; \Psi \vde \alpha[\id(\hat\Psi)] = \alpha[\id(\hat\Psi)] \searrow \Gamma\) \hfill by unification rule\\
    \item Subcase: \(\Gamma_0 = \ctxapp{\Gamma''_0}{\ctxapp{\beta{:}(\Psi \atl{k} *)}{\Gamma'_0}}\) and
      \[
        \mathcal{E} =
        \begin{array}{c}
          \infer%
          {\ctxapp{\Gamma}{\Psi} \vde \beta[\id(\hat\Psi)]}%
          {%
            \deduce%
            {\beta{:}(\Psi \atl{k} *) \in (\ctxapp{\Gamma}{\Psi})}%
            {\mathcal{D}_1}%
            &%
            \deduce%
            {\ctxapp{\Gamma}{\Psi} \vde \id(\hat\Psi) : \Psi}%
            {\mathcal{D}_2}%
          }%
        \end{array}
      \]
      \(\ctxapp{\ctxapp{\ctxapp{\Gamma_1}{\alpha{:}(\Psi \atl{k} *)}}{\Gamma''_0}}{\Psi} \vde \alpha[\id(\hat\Psi)]\) \hfill by typing rule\\
      \(\Gamma ; \Psi \vde \beta \not\in \alpha[\id(\hat\Psi)]\) \hfill as \(\beta \not\in \Psi\)\\
      \(\Gamma ; \Psi \vde \alpha[\id(\hat\Psi)] = \beta[\id(\hat\Psi)] \searrow \ctxapp{\Gamma_1}{\ctxapp{\alpha{:}(\Psi \atl{k} *)}{\ctxapp{\Gamma''_0}{\ctxapp{\beta{:=}(\hat\Psi . \alpha[id(\hat\Psi)]){:}(\Psi \atl{k} *)}{\Gamma'_0}}}}\)\\
      \mbox{} \hfill by unification rule\\
    \item Subcase: \(\Gamma_0 = \ctxapp{\Gamma''_0}{\ctxapp{\beta{:=}(\hat\Psi . S){:}(\Psi \atl{k} *)}{\Gamma'_0}}\) and
      \[
        \mathcal{E} =
        \begin{array}{c}
          \infer%
          {\ctxapp{\Gamma}{\Psi} \vde \beta[\id(\hat\Psi)]}%
          {%
            \deduce%
            {\beta{:=}(\hat\Psi . S){:}(\Psi \atl{k} *) \in \ctxapp{\Gamma_1}{\ctxapp{\alpha{:}(\Psi \atl{k} *)}{\Gamma_0}}}%
            {\mathcal{D}_1}%
            &%
            \deduce%
            {\ctxapp{\Gamma}{\Psi} \vde \sigma : \Psi}%
            {\mathcal{D}_2}%
          }%
        \end{array}
      \]
      \(\ctxapp{\ctxapp{\ctxapp{\Gamma_1}{\alpha{:}(\Psi \atl{k} *)}}{\Gamma''_0}}{\Psi} \vde S\) \hfill by inversion on context\\
      \(\ctxapp{\Gamma}{\Psi} \vde S\) \hfill by weakening\\
      \(\Gamma ; \Psi \vde \alpha[\id(\hat\Psi)] = S \searrow \Gamma'\) \hfill by IH\\
      \(\Gamma ; \Psi \vde \beta \not\in \alpha[\id(\hat\Psi)]\) \hfill as \(\beta \not\in \Psi\)\\
      \(\Gamma ; \Psi \vde \alpha[\id(\hat\Psi)] = \beta[\id(\hat\Psi)] \searrow \Gamma'\) \hfill by unification rule\\
    \item Other subcases:\\
      \(\Gamma ; \Psi \vde \alpha[\id(\hat\Psi)] = S \searrow \Gamma, \#\) \hfill by unification rule\\
    \end{itemize}
  \end{itemize}
\end{proof}

\begin{thmn}[\ref{lem:unifysound}]\mbox{\textnormal{\textsc{(Soundness of Unification).}}}\quad
  \begin{enumerate}
  \item
    If
    $\Gamma ; \Gamma_0 \vde \Phi = \Psi \searrow \Gamma'$\\
    then
    $\ctxapp{\Gamma'}{\Gamma_0} \vde \Phi = \Psi$ and
    $\Gamma' \ext \Gamma$.
  \item
    If
    $\vde \ctxapp{\Gamma}{\Phi}$ and
    $\Gamma ; \Phi \vde T = S \searrow \Gamma'$\\
    then
    $\vde \Gamma'$ and
    $\ctxapp{\Gamma'}{\Phi} \vde T = S$ and
    $\Gamma' \ext \Gamma$.
  \end{enumerate}
\end{thmn}

\begin{proof}
  By induction on the unification derivation. We show here the unification variable case only.

  \begin{itemize}
  \item Case:
    \[
      \infer%
      {\Gamma ; \Phi \vde \alpha[\id(\hat\Phi)] = S \searrow \Gamma'}%
      {%
        \deduce%
        {\Gamma = \ctxapp{\Gamma_1}{\ctxapp{\alpha{:}(\Phi \atl{k} *)}{\Gamma_0}}}%
        {\mathcal{D}_1}%
        &%
        \deduce%
        {\ctxapp{\Gamma_1}{\Phi} \vde S}%
        {\mathcal{D}_2}%
        &%
        \deduce%
        {\Gamma ; \Phi \vde \alpha \not\in S}%
        {\mathcal{D}_3}%
        &%
        \deduce%
        {\Gamma' = \ctxapp{\Gamma_1}{\ctxapp{\alpha{:=}(\hat\Phi.S){:}(\Phi \atl{k} *)}{\Gamma_0}}}%
        {\mathcal{D}_4}%
      }%
    \]
    \(\ctxapp{\Gamma_1}{\Phi} \vde S\) \hfill by \(\mathcal{D}_2\)\\
    \(\vde \Gamma_1, \alpha{:=}(\hat\Phi.S):(\Phi \atl{n} *), \Gamma_0\) \hfill by well-formedness rule\\
    \(\vde \Gamma'\) \hfill by \(\mathcal{D}_4\)\\
    \(\ctxapp{\Gamma'}{\Phi} \vde S\) \hfill by weakening\\
    \([\id(\hat\Phi)/\hat\Phi]S = S\) \hfill by property of id.~subst.~\ref{lem:idsubstprop}\\
    \(\ctxapp{\Gamma'}{\Phi} \vde [\id(\hat\Phi)/\hat\Phi]S = S\) \hfill by reflexivity of struct.~eq.\\
    \(\alpha{:=}(\hat\Phi.S){:}(\Phi \atl{n} *) \in \Gamma'\)\hfill by the def.~of \(\in\)\\
    \(\ctxapp{\Gamma'}{\Phi} \vde \alpha[\id(\hat\Phi)] = S\) \hfill by structural eq.\\
    \(\Gamma' \ext \Gamma\) \hfill by the def.~of \(\ext\)\\
  \end{itemize}
\end{proof}


}
\end{document}